\documentclass[
longbibliography,
reprint,
amsmath,amssymb,
aps,
pra,
superscriptaddress,
nofootinbib
]{revtex4-2}

\usepackage{graphicx}
\usepackage{bm}
\usepackage{braket}
\usepackage[utf8]{inputenc}
\usepackage{color}
\usepackage{comment}
\usepackage{physics}
\usepackage{bbold}
\usepackage{amsmath, amssymb, amsthm}
\usepackage{mathtools}
\usepackage{graphicx}
\usepackage[normalem]{ulem}
\usepackage[colorlinks,linkcolor=blue,urlcolor=blue, citecolor=blue]{hyperref}
\usepackage{rotating}

\theoremstyle{plain}
\newtheorem{theorem}{Theorem}
\newtheorem{lemma}[theorem]{Lemma}

\newtheorem{assumption}{Assumption}
\theoremstyle{definition}
\newtheorem{definition}[theorem]{Definition}

\theoremstyle{remark}

\usepackage{algorithm}
\usepackage{algorithmicx}
\usepackage{algpseudocode}

\newcommand{\an}{\text{~and~}}
\newcommand{\thre}{\mathrm{th}}
\newcommand{\loc}{\mathrm{loc}}
\newcommand{\FT}{\mathrm{FT}}
\newcommand{\prep}{\mathrm{prep}}
\newcommand{\gat}{\mathrm{gate}}
\newcommand{\meas}{\mathrm{meas}}
\newcommand{\synd}{\mathrm{synd}}
\newcommand{\org}{\mathrm{org}}
\newcommand{\EC}{\mathrm{EC}}
\newcommand{\wait}{\mathrm{wait}}

\newcommand{\sync}{\mathrm{sync}}
\newcommand{\dec}{\mathrm{dec}}
\newcommand{\bw}{\mathrm{bw}}
\newcommand{\fw}{\mathrm{fw}}
\newcommand{\re}{\mathrm{re}}
\newcommand{\prece}{\mathrm{pre}}

\algnewcommand{\Initialize}[1]{
\State \textbf{initialize:}
\State \hspace*{\algorithmicindent}\parbox[t]{0.8\linewidth}{\raggedright{#1}}
}

\begin{document}

\begin{abstract}
A major challenge in fault-tolerant quantum computation (FTQC) is to reduce both space overhead---the large number of physical qubits per logical qubit---and time overhead---the long physical gate sequences per logical gate.
We prove that a protocol using non-vanishing-rate quantum low-density parity-check (LDPC) codes, combined with concatenated Steane codes, achieves constant space overhead and polylogarithmic time overhead, even when accounting for non-zero classical computation time.
This protocol offers an improvement over existing constant-space-overhead protocols, which have polynomial time overhead using quantum LDPC codes and quasi-polylogarithmic time overhead using concatenated quantum Hamming codes.
To ensure the completeness of this proof, we develop a technique called partial circuit reduction, which enables error analysis for the entire fault-tolerant circuit by examining smaller parts composed of a few gadgets.
With this technique, we resolve a previously unaddressed logical gap in the existing arguments and complete the proof of the threshold theorem for the constant-space-overhead protocol with quantum LDPC codes.
Our work highlights that the quantum-LDPC-code approach can realize FTQC with a negligibly small slowdown and a bounded overhead of physical qubits, similar to the code-concatenation approach, underscoring the importance of a comprehensive comparison of the future realizability of these two approaches.
\end{abstract}

\title{Polylog-time- and constant-space-overhead fault-tolerant quantum computation\\with quantum low-density parity-check codes}

\author{Shiro Tamiya}
\email{shiro.tamiya01@gmail.com}
\affiliation{Department of Applied Physics, Graduate School of Engineering, The University of Tokyo, 7-3-1 Hongo, Bunkyo-ku, Tokyo 113-8656, Japan}
\affiliation{Nanofiber Quantum Technologies, Inc., 1-22-3 Nishiwaseda, Shinjuku-ku, Tokyo 169-0051, Japan.}

\author{Masato Koashi}
\email{koashi@qi.t.u-tokyo.ac.jp}
\affiliation{Department of Applied Physics, Graduate School of Engineering, The University of Tokyo, 7-3-1 Hongo, Bunkyo-ku, Tokyo 113-8656, Japan}
\affiliation{Photon Science Center, Graduate School of Engineering, The University of Tokyo, 7-3-1 Hongo, Bunkyo-ku, Tokyo 113-8656, Japan}

\author{Hayata Yamasaki}
\email{hayata.yamasaki@gmail.com}
\affiliation{Department of Physics, Graduate School of Science, The University of Tokyo, 7-3-1 Hongo, Bunkyo-ku, Tokyo 113-0033, Japan}
\affiliation{Nanofiber Quantum Technologies, Inc., 1-22-3 Nishiwaseda, Shinjuku-ku, Tokyo 169-0051, Japan.}

\maketitle

\tableofcontents

\section{Introduction}

Quantum computation has promising potential for accelerating the solving of certain classes of computational problems compared to classical computation~\cite{Montanaro_2016,dalzell2023quantum}.
However, the inherent fragility of quantum systems poses a significant challenge in implementing quantum computations on quantum devices.
Fault-tolerant quantum computation (FTQC) has been developed as a powerful framework to overcome this obstacle, enabling reliable computation even in the presence of errors~\cite{gottesman2010introduction, nielsen00}.
FTQC simulates an original quantum circuit using logical qubits encoded within a quantum error-correcting code, instead of directly using physical qubits as qubits in the original circuit.
This approach ensures accurate quantum computation even in the presence of errors that may accumulate as the size of the original circuit increases~\cite{gottesman2010introduction, nielsen00}.

Currently, two prominent schemes have been proposed for FTQC\@: one is a concatenated code scheme~\cite{aliferis2005quantum, gottesman2010introduction, Yamasaki_2024} and the other is a quantum low-density parity-check (LDPC) code scheme~\cite{1997RuMaS..52.1191K,Kitaev_2003,bravyi1998quantum,Litinski_2019,Gottesman2014Constant,Bombin2015Single,bravyi2024high,Fowler_2012,Breuckmann_2021, Tillich_2014, Leverrier2015QuantumExpander}. 
For both schemes, threshold theorems~\cite{1997RuMaS..52.1191K,aharonov1999faulttolerant, aliferis2005quantum,Knill_1998,reichardt2005faulttolerance,Terhal_2005,PhysRevA.78.052331,aliferis2006faulttolerant,aliferis2007accuracy,Preskill_1998} guarantee that the failure probability of the fault-tolerant simulation can be arbitrarily suppressed, provided that the physical error rate is below a certain threshold.
However, conventional FTQC schemes, such as those using surface codes~\cite{bravyi1998quantum, Litinski_2019} and concatenated Steane codes~\cite{aliferis2005quantum}, require a substantial increase in the number of physical qubits per logical qubit, which scales polylogarithmically with the size of the original quantum circuit~\cite{Fowler_2012,aliferis2005quantum}.
This poses a challenge, as the number of physical qubits in quantum devices is limited, making space overhead the primary obstacle to realizing FTQC\@.
In addition to space overhead, reducing time overhead, which is defined as the ratio of the physical time step in simulating an original circuit to the time step in the original circuit, is also important to retain the speedup of quantum computation.
The conventional FTQC schemes scale polylogarithmically in time overhead with the size of the original circuit~\cite{Fowler_2012,aliferis2005quantum}.
Reducing space and time overhead in FTQC is of great interest from both practical and theoretical perspectives.

One of the main questions in the field of FTQC is how short a time overhead we can achieve while simultaneously maintaining a constant space overhead.
In recent years, there have been significant advances in this problem.
Reference~\cite{Gottesman2014Constant} clarified the properties non-vanishing-rate quantum LDPC codes must retain to achieve FTQC with a constant space overhead combined with concatenated Steane codes.
Subsequently, Refs.~\cite{Fawzi_2018,grospellier:tel-03364419} showed that quantum expander codes~\cite{Leverrier2015QuantumExpander, Fawzi_2018_eff, Fawzi_2018} can be used as quantum LDPC codes for this protocol.
Although this protocol maintains a constant space overhead, it compromises the parallelism of gate operations by limiting the number of logical gates performed per time step. 
In these protocols, logical gates of non-vanishing-rate quantum LDPC codes are executed using gate teleportation with auxiliary states encoded in quantum LDPC codes~\cite{gottesman2010introduction,knill2005scalable,Knill_2005}.
Preparing these encoded auxiliary states in a fault-tolerant manner has been challenging without relying on the conventional protocol using concatenated Steane codes~\cite{aliferis2005quantum, gottesman2010introduction}.
However, the conventional protocol incurs a growing space overhead~\cite{aliferis2005quantum, gottesman2010introduction}, which undermines complete gate parallelism when we need to maintain a constant space overhead.
As a result, the existing analyses in~\cite{Gottesman2014Constant, grospellier:tel-03364419, Fawzi_2018} require sequential gate implementation, leading to a polynomial increase in the time overhead of this protocol.
More recently, Reference~\cite{Yamasaki_2024} resolved this bottleneck by developing a new constant-space-overhead protocol with complete gate parallelism based on concatenated quantum Hamming codes.
However, even with this fully parallel execution of gates, the protocol only achieves quasi-polylogarithmic time overhead.
Thus, it remains an open question whether it is possible to design an even faster constant-space-overhead protocol that achieves polylogarithmic time overhead, which is even shorter than the quasi-polylogarithmic time overhead and comparable to the time overhead of the conventional protocols with polylogarithmic space overhead.

From a more fundamental perspective, this question involves a trade-off relation between space and time overhead in FTQC, originally raised in~\cite{Gottesman2014Constant} and partially addressed in~\cite{Yamasaki_2024}.
On the one hand, there exist fault-tolerant protocols in which both the space overhead and the time overhead are polylogarithmic~\cite{Fowler_2012, aliferis2005quantum}. 
On the other hand, when space overhead is reduced to a constant, an ideal trade-off would be that the time overhead remains polylogarithmic, with only a higher degree polynomial, but existing constant-space-overhead protocols exhibit either the polynomial time overheads~\cite{Gottesman2014Constant,Fawzi_2018,grospellier:tel-03364419} or quasi-polylogarithmic time overheads~\cite{Yamasaki_2024}, resulting in apparently redundant time overhead. 
This issue highlights the challenge of understanding the trade-off between space and time overhead in FTQC\@.

Increasing the parallelism of logical gates in non-vanishing-rate LDPC codes offers a promising approach, yet achieving this using existing techniques remains challenging.
For instance, even with structured non-vanishing-rate quantum LDPC codes, applying logical gates without auxiliary qubits in a fault-tolerant way, which may enable complete gate parallelism, is difficult~\cite{Tillich_2014, Quintavalle_2023}.
Although hyperbolic surface codes~\cite{Breuckmann_2016} possess the potential for achieving complete gate parallelism~\cite{Lavasani_2019,breuckmann2022foldtransversal}, it remains challenging to use them for FTQC due to the lack of efficient decoding algorithms in the presence of noise on syndrome measurements~\cite{hastings2013decoding,Gottesman2014Constant,Lavasani_2019}.
Alternative approaches to performing logical gates may be code deformation~\cite{Bombin_2009} and lattice surgery~\cite{Horsman_2012}.
Code deformation transforms one code into another, where the resulting codes complement each other in terms of the logical gates that can be implemented fault-tolerantly.
While code deformation techniques have been proposed for non-vanishing-rate quantum LDPC codes~\cite{Krishna_2021}, there is no guarantee that the gate implementation time step is shorter than that of gate teleportation.
Lattice surgery performs logical gates by measuring multiqubit logical Pauli operators.
A lattice surgery technique for a general class of quantum LDPC codes has also been proposed in~\cite{Cohen_2022}. Nevertheless, this technique requires a large number of auxiliary qubits for parallel operation on all encoded qubits and thus cannot achieve sufficient gate parallelism while maintaining constant space overhead.

In this work, we prove that a polylogarithmic time overhead can be achieved with quantum LDPC codes while maintaining a constant space overhead, thus eliminating the trade-off between space overhead and time overhead. To achieve this, we analyze a hybrid fault-tolerant protocol that combines concatenated Steane codes~\cite{aliferis2005quantum} and non-vanishing-rate quantum LDPC codes with an efficient decoding algorithm, specifically quantum expander codes~\cite{Leverrier2015QuantumExpander}, as discussed in~\cite{Gottesman2014Constant,Fawzi_2018,grospellier:tel-03364419}. In this hybrid protocol, non-vanishing-rate quantum LDPC codes are used as registers to store and protect logical qubits, while concatenated codes are employed to prepare auxiliary states encoded in the non-vanishing-rate quantum LDPC code, facilitating the implementation of logical gates on these logical qubits via gate teleportation~\cite{gottesman2010introduction,knill2005scalable,Knill_2005}.
Our analysis enables the polylogarithmic time overhead by increasing gate parallelism in the gate operations using gate teleportation, compared to previous analyses~\cite{Gottesman2014Constant,Fawzi_2018,grospellier:tel-03364419}, while still maintaining a constant space overhead. This improves upon the polynomial time overhead presented in earlier works~\cite{Gottesman2014Constant, Fawzi_2018, grospellier:tel-03364419}.

The contribution of our analysis to the theory of FTQC is to refine and fully complete the analysis of the threshold theorem for this protocol, advancing over~\cite{Gottesman2014Constant,Fawzi_2018,grospellier:tel-03364419}.
The existing analyses of threshold theorems for protocols with quantum LDPC codes~\cite{Gottesman2014Constant, Fawzi_2018, grospellier:tel-03364419} assume that classical computations are performed instantaneously on arbitrarily large scales, and prior to our work, the nonzero runtime of classical computations in implementing FTQC was explicitly taken into account only in the analysis of the threshold theorem for concatenated codes~\cite{Yamasaki_2024}.
By contrast, we prove that the existence of a threshold and the aforementioned scaling of time overhead can be established even when accounting for the runtime of classical computations during the execution of FTQC, e.g., those for decoding.
To achieve this, we employ the constant-time decoding algorithm in Ref.~\cite{grospellier:tel-03364419} and bound the runtime of all classical computations required to perform the FTQC\@.
Our analysis shows that even when considering the runtime of classical computations, our constant-space-overhead protocol exhibits a polylogarithmic time overhead.

To analyze the fault tolerance of our hybrid protocol, we identify the fault-tolerance conditions for gadgets of quantum LDPC codes.
Unlike the fault-tolerance conditions for protocols of concatenated codes~\cite{aliferis2005quantum, Yamasaki_2024}, our conditions for quantum LDPC codes do not strictly require non-propagation of errors.
In the cases of concatenated codes, transversality, i.e., a property that gadgets are implemented by a tensor product of gates acting individually on each physical qubit of each code block, has been crucial for avoiding the propagation of errors in the code block.
By contrast, for the quantum LDPC codes, our fault-tolerance conditions require that gadgets should have a constant depth, but errors can propagate via physical operations to a constant number of physical qubits in each code block, which we believe will also be useful for analyzing other fault-tolerant protocols with quantum LDPC codes to establish the threshold theorem.

Moreover, our analysis newly addresses and fixes a previously overlooked problem in the existing analyses~\cite{Gottesman2014Constant,Fawzi_2018, grospellier:tel-03364419} of the threshold theorem for the quantum LDPC codes.
In the analyses of FTQC, it is conventional to consider the local stochastic Pauli error model, where the correlated errors may occur across the entire fault-tolerant quantum circuit, as in~\cite{Gottesman2014Constant,Fawzi_2018,grospellier:tel-03364419}.
For each code block of the quantum LDPC codes, the existing analyses in~\cite{Gottesman2014Constant,Fawzi_2018,grospellier:tel-03364419} argue that the physical error rate of the codeword is suppressed after noisy syndrome measurements followed by quantum error correction using decoding algorithms, but this codeword state is on specific code blocks, i.e., only in a small part of the entire fault-tolerant circuit on which the local stochastic Pauli error model is defined.
This argument overlooks the correlations between this part of the circuit and the rest of the circuit surrounding this part, leaving the overall proof of the threshold theorem incomplete.
To address this issue, we introduce a new method, called partial circuit reduction.
This method enables error analysis of the entire fault-tolerant circuit through the examination of individual gadgets on the code blocks of quantum LDPC codes.
Our approach allows us to leverage the existing results from the decoding algorithms~\cite{Gottesman2014Constant,Fawzi_2018, grospellier:tel-03364419} as a black box, so that we fully complete the proof of the threshold theorem for the constant-space-overhead protocol with the quantum LDPC codes.
Furthermore, by combining this method with theoretical advances of the decoding algorithms in~\cite{Fawzi_2018,grospellier:tel-03364419}, we show that it is indeed possible to achieve higher parallelization of the logical gates compared to the existing analyses~\cite{Gottesman2014Constant, Fawzi_2018, grospellier:tel-03364419}, resulting in a polylogarithmic time overhead.

These results indicate a promising potential for low-overhead FTQC using the hybrid approach that combines non-vanishing-rate quantum LDPC codes and concatenated codes, in addition to relying solely on the code-concatenation approach as proposed in~\cite{Yamasaki_2024, yoshida2024concatenate}.
Our findings contribute to a fundamental understanding of the minimal overhead and the space-time trade-off required to realize FTQC within the framework of quantum mechanics.
Furthermore, with regard to practical implementation, our results clarify the competing time overheads of these two constant-space-overhead approaches, highlighting the importance of a comprehensive investigation into the physical realizability of both approaches.

The rest of this paper is organized as follows.
In Sec.~\ref{sec:preliminaries}, we present preliminaries for this work.
In Sec.~\ref{sec:setting}, we describe the setting of FTQC\@.
In Sec.~\ref{sec: Description of polylog-time constant-space overhead protocol}, we present the hybrid protocol using non-vanishing-rate quantum LDPC codes and concatenated Steane codes to achieve polylogarithmic-time- and constant-space-overhead FTQC\@.
In Sec.~\ref{sec:gadgets}, we introduce conditions of fault-tolerant gadgets for quantum LDPC codes and construct gadgets for our protocol.
In Sec.~\ref{sec: threshold theorem based on quantum LDPC codes}, we prove the threshold theorem for this protocol and bound its time and space overhead; in particular, in Sec.~\ref{sec: parial circuit reduction}, we introduce the technique of partial circuit reduction, which is central for establishing the threshold theorem.
Finally, in Sec.~\ref{sec:conclusion}, we conclude our work.
In Appendix~\ref{appendix: fault-tolerant protocol for open quantum circuits}, we also provide a rigorous analysis of the fault-tolerant protocol using concatenated Steane codes for implementing open quantum circuits (which do not terminate with measurements), which is used in our hybrid protocol to prepare encoded states of non-vanishing-rate quantum LDPC codes.

\section{Preliminaries\label{sec:preliminaries}}

In this section, we present preliminaries used in our paper.
In Sec.~\ref{sec: Pauli group and Clifford group}, we start with introducing the Pauli group and the Clifford group, along with their binary representations.
In Sec.~\ref{sec: Concatenated Steane codes and non-vanishing-rate quantum LDPC codes}, we introduce the quantum error-correcting codes used in our fault-tolerant protocol.

\subsection{Pauli group and Clifford group\label{sec: Pauli group and Clifford group}}

We present the basics of the Pauli group and the Clifford group.
Their binary representations will be used throughout our fault-tolerant protocol.

We recall the definition of the Pauli group as follows.
Let $X,Y$, and $Z$ be the Pauli operators defined as
\begin{equation}
X\coloneqq\begin{bmatrix}
  0 & 1  \\
  1 & 0  \\
\end{bmatrix},\quad
Z\coloneqq\begin{bmatrix}
  1 & 0  \\
  0 & -1  \\
\end{bmatrix},\quad
Y\coloneqq\mathrm{i}XZ=\begin{bmatrix}
  0 & -\mathrm{i}  \\
  \mathrm{i} & 0  \\
\end{bmatrix}.
\end{equation}
We also let $I$ denote an identity operator with size $2\times 2$, which we may also consider one of the Pauli operators for brevity of presentation.
When we need to specify the dimension of the identity operator, the identity operator of size $2^n\times 2^n$ is denoted by $I_n$.
The Pauli group on $n$ qubits, denoted by $\mathcal{P}_n$, consists of $2^n\times 2^n$ matrices in the form of 
\begin{equation}
    P=\alpha\bigotimes_{i=0}^{n-1}P_i,
\label{eq: example of Pauli operator}
\end{equation}
where $\alpha\in\{\pm 1,\pm \mathrm{i}\}$, $P_i\in\{I,X,Y,Z\}$, and $\mathrm{i}^2=-1$.
Also, let $\langle \mathrm{i}I_n\rangle$ be the center of the Pauli group $\mathcal{P}_n$, generated by $\mathrm{i}I_n$.
The projective Pauli group is defined as the Pauli group where global phases are ignored, i.e.,
\begin{equation}
    \Tilde{\mathcal{P}}_n\coloneqq \mathcal{P}_n/\langle \mathrm{i}I_n\rangle.
\end{equation}
Each representative of $\tilde{\mathcal{P}}_n$ is selected by decomposing a Pauli gate $P'\in\mathcal{P}_n$ as 
\begin{equation}
    P'=\beta\bigotimes_{i=0}^{n-1}P_i^{\prime},
\end{equation}
where $P_i^{\prime}\in\{I,X,Z,ZX\}$ and fixing $\beta=1$.

The eigenvectors of the Pauli-$Z$ operator serve as a particular basis for a qubit, which is referred to as the computational basis.
This basis is denoted by $\{\ket{0}, \ket{1}\}$, where $\ket{0}$ and $\ket{1}$ are the eigenstates of the Pauli-$Z$ operator, corresponding to the eigenvalues $+1$ and $-1$, respectively.
As a special case of projective measurement, we often use the measurement of the Pauli operator $P\in \mathcal{P}_n$ whose projective operators $\{\Pi_m\}_m$ associated with the measurement outcomes $m\in\{0,1\}$ are defined as
\begin{equation}
    \Pi_0\coloneqq(I_{n}+P)/2 \text{,~~} \Pi_{1}\coloneqq(I_{n}-P)/2.
\end{equation}
For a given Pauli operator $P \in \mathcal{{\mathcal{P}}}_n$, the weight of $P$, denoted by $|P|$, is defined as the number of qubits on which $P$ acts non-trivially, i.e.,
\begin{equation}
    |P|\coloneqq\#\{i\in\{1,\ldots,n\}\colon P_i\neq I\},
\end{equation}
where $P_i$ is given by~\eqref{eq: example of Pauli operator}.

The Clifford group is defined as the group whose elements map Pauli operators to Pauli operators under conjugation.
Let $\mathbb{U}(d)$ be the unitary group on the set $\mathbb{C}^d$ of complex vectors. 
The Clifford group on $n$ qubits is the normalizer of the $n$-qubit Pauli group,
\begin{equation}
    \mathcal{C}_n\coloneqq \{C\in \mathbb{U}(2^n) \mid C P C^{\dagger} \in \mathcal{P}_n, \forall P \in \mathcal{P}_n\}.
\end{equation}
where $C^{\dagger}$ represents the adjoint of $C$.
The elements of the Clifford group are called Clifford operators.
Furthermore, the projective Clifford group is defined as
\begin{equation}
    \Tilde{\mathcal{C}}_n\coloneqq \mathcal{C}_n/\mathbb{U}(1),
\end{equation}
where our analysis does not need to specify the representative of $\tilde{\mathcal{C}}_n$ while the representative of $\tilde{\mathcal{C}}_n/\tilde{\mathcal{P}}_n$ will be specified later in~\eqref{def: Symplectic matrix representation of Clifford operators}.

The symplectic representation of Pauli operators allows us to perform several calculations on Pauli operators using a binary vector.
The definition of the symplectic representation of Pauli operators is as follows.
\begin{definition}(Symplectic representation of Pauli operators)

Let $P=\bigotimes_{i=1}^{n}P_i\in\Tilde{\mathcal{P}}_n$ be a Pauli operator.
The mapping $\phi\colon \Tilde{\mathcal{P}}_n\to \mathbb{F}_2^{2n}$ provides the symplectic representation of $P$, where $\mathbb{F}_2$ is the finite field of order $2$ representing a bit.
In this representation, $P$ is mapped to a pair of row vectors $x, z\in\mathbb{F}_2^{n}$, denoted by $\phi(P)\coloneqq[x,z]\in\mathbb{F}^{2n}_2$, according to the following rules:
\begin{alignat}{3}
&\text{if}\; P_i=I, &&\;x_i=0 \; &&\text{and} \; z_i=0,\\
&\text{if}\; P_i={X}, &&\;x_i=1 \; &&\text{and} \; z_i=0,\\
&\text{if}\; P_i={ZX}, &&\;x_i=1 \; &&\text{and} \; z_i=1,\\
&\text{if}\; P_i={Z}, &&\;x_i=0 \; &&\text{and} \; z_i=1.
\end{alignat}
\end{definition}
\noindent Since $\Tilde{\mathcal{P}}_n\cong\mathbb{F}_2^{2n}$, when we ignore the global phase of the Pauli operator, the multiplication in two Pauli operators $P_1,P_2\in\Tilde{\mathcal{P}}_n$ can be calculated with their symplectic representation as
\begin{equation}
    \phi(P_1)\oplus\phi(P_2),
\label{eq: multiplication of Pauli operators}
\end{equation}
where $\oplus$ denotes the bitwise exclusive OR (XOR).
Moreover, the multiplication in $P_1,P_2\in\Tilde{\mathcal{P}}_n$ can also be expressed as
\begin{equation}
    P_1 P_2={(-1)}^{xz'^{\top}+x'z^{\top}}P_2 P_1,
\label{eq: commutant with symplectic rep}
\end{equation}
where $\phi(P_1)=[x,z]\in\mathbb{F}_2^{2n}$, $\phi(P_2)=[x',z']\in\mathbb{F}_2^{2n}$, and $z^{\top}$ denotes the transpose of $z$.
Thus, the commutator between two Pauli operators corresponds to $xz'^{\top}+x'z^{\top}$ in~\eqref{eq: commutant with symplectic rep}, which is defined as the symplectic inner product.

The conjugation of a Pauli operator in $\Tilde{\mathcal{P}}_n$ by a Clifford operator in $\Tilde{\mathcal{C}}_n$ is carried out in a way that maintains the symplectic inner product.
From Refs.~\cite{de_Beaudrap_2013,Rengaswamy_2018}, we have $\Tilde{\mathcal{C}}_n/\mathcal{\Tilde{P}}_n\cong\mathrm{Sp}(2n,\mathbb{F}_2)$, where the equivalence relation is defined by conjugation, and $\mathrm{Sp}(2n,\mathbb{F}_2)$ represents the group of $2n\times 2n$ symplectic matrices over $\mathbb{F}_2$.
In addition, the corresponding $\Gamma\in\mathrm{Sp}(2n,\mathbb{F}_2)$ can be identified by mapping $[e_i,0]\mapsto[x,z]$ and $[0,e_j]\mapsto[x',z']$, where $e_i$ is the standard basis vector of $\mathbb{F}_2^{n}$ that has an entry $1$ in the $i$-th column and $0$ otherwise~\cite{Rengaswamy_2018}.
These facts provide the symplectic matrix $\gamma(C)\in\mathbb{F}_2^{2n\times 2n}$ of a representative of a Clifford operator $C\in\Tilde{\mathcal{C}}_n/\Tilde{\mathcal{P}}_n$~\cite{Rengaswamy_2018} as
    \begin{widetext}
    \begin{equation}
    \label{def: Symplectic matrix representation of Clifford operators}
        \gamma(C)\coloneqq
        \begin{bmatrix}
            & & & & &\\
            & & & & &\\
            \phi(C X_n^{(1)} C^{\dagger})^{\top}&\cdots& \phi(C X_n^{(n)} C^{\dagger})^{\top}& \phi(C Z_n^{(1)} C^{\dagger})^{\top}&\cdots& \phi(C Z_n^{(n)} C^{\dagger})^{\top}\\
            & & & & &\\
            & & & & &\\
        \end{bmatrix},
    \end{equation}
    \end{widetext}
where $X_n^{(i)}(Z_n^{(i)})\in\Tilde{\mathcal{P}}$ represents an $n$-qubit Pauli operator that acts as $X(Z)$ on the $i$-th qubit, and as $I$ otherwise.
\noindent The conjugation of a Pauli operator $P\in\Tilde{\mathcal{P}}_n$ by a Clifford operator $C\in\Tilde{\mathcal{C}}_n/\Tilde{\mathcal{P}}_n$ can be calculated by multiplying the matrix $\gamma(C)\in\mathbb{F}_2^{2n\times 2n}$ with the vector $\phi(P)\in\mathbb{F}_2^{2n}$ from the right, i.e.,
\begin{equation}
    \phi(P)\gamma(C).
\label{eq: matrix-vector multiplication}
\end{equation}

The group $\Tilde{\mathcal{C}}_n/\Tilde{\mathcal{P}}_n$ is generated by the $H$ operator, $S$ operator, and either CNOT or CZ operator defined as
\begin{align}
H&\coloneqq\frac{1}{\sqrt{2}}\begin{bmatrix}
  1 & 1  \\
  1 & -1  \\
\end{bmatrix},\\
S&\coloneqq\begin{bmatrix}
  e^{-\mathrm{i}\frac{\pi}{4}} & 0  \\
  0 & e^{\mathrm{i}\frac{\pi}{4}}  \\
\end{bmatrix}=R_Z\left(\frac{\pi}{2}\right),\\
\mathrm{CNOT}&\coloneqq \begin{bmatrix}
  1 & 0 & 0 & 0  \\
  0 & 1 & 0 & 0 \\
  0 & 0 & 0 & 1 \\
  0 & 0 & 1 & 0 \\
\end{bmatrix},\\
\mathrm{CZ}&\coloneqq \begin{bmatrix}
  1 & 0 & 0 & 0  \\
  0 & 1 & 0 & 0 \\
  0 & 0 & 1 & 0\\
  0 & 0 & 0 & -1 \\
\end{bmatrix}=(I\otimes H)\mathrm{CNOT}(I\otimes H),
\end{align}
where $R_Z(\theta)\coloneqq e^{-\mathrm{i}\frac{\pi}{2}\theta Z}=\begin{bmatrix}
  e^{-\mathrm{i}\frac{\pi}{2}\theta} & 0  \\
  0 & e^{\mathrm{i}\frac{\pi}{2}\theta}  \\
\end{bmatrix}$.
Each representative of $\Tilde{\mathcal{C}}_n/\Tilde{\mathcal{P}}_n$ is selected by decomposing a Clifford operator $C\in \mathcal{C}_n$ into a sequence of $H$, $S$, and CNOT operators followed by a Pauli operator $P\in\mathcal{P}_n$, and removing the Pauli operator $P$.

\subsection{Quantum codes\label{sec: Concatenated Steane codes and non-vanishing-rate quantum LDPC codes}}
In this section, we present the basics of stabilizer codes that will be used in this paper.
In Sec.~\ref{sec: stabilizer codes}, we summarize stabilizer codes and Calderbank-Shor-Steane (CSS) codes.
In Sec.~\ref{sec:Concatenated Steane codes}, we present the Steane code and its code concatenation.
In Sec.~\ref{sec:Quantum LDPC codes}, we present quantum LDPC codes.

\subsubsection{Stabilizer codes and Calderbank-Shor-Steane (CSS) codes\label{sec: stabilizer codes}}
A quantum error-correcting code $\mathcal{Q}$ with $N$ physical qubits is a subspace $\mathcal{Q}$ of a $2^N$-dimensional Hilbert space, i.e., $\mathcal{Q} \subseteq (\mathbb{C}^2)^{\otimes N}$, where $\mathbb{C}^2=\mathrm{span}\{\ket{0},\ket{1}\}$ represents a qubit.
We may refer to $N$ as a code size.
We consider a stabilizer code, which is specified by its stabilizer $\mathcal{S}\subset \mathcal{P}_N$, which is an Abelian subgroup of $\mathcal{P}_N$ satisfying $-I\notin \mathcal{S}$.
The centralizer $C(\mathcal{S})$ consists of all Pauli operators that commute with all elements in $\mathcal{S}$.
If $\mathcal{S}$ is generated by $N-K$ independent elements, then the corresponding space $\mathcal{Q}(\mathcal{S})$ is a $2^K$-dimensional subspace that is invariant under the action of $\mathcal{S}$.
In this case, the subspace is isomorphic to a space of $K$ qubits, which are called the \textit{logical qubits}.
Logical operators are elements of the set $C(\mathcal{S})\setminus \mathcal{S}$.
For each $k\in[1,\ldots,K]$, one can choose the associated logical $X$ and $Z$ operators acting on the $k$-th logical qubit in $C(\mathcal{S})\setminus\mathcal{S}$ that obey the Pauli commutation relation.
We refer to the subspace $\mathcal{Q}(\mathcal{S})$ as the \textit{code space} and a state of logical qubits as a \textit{codeword}.
The minimum number of physical qubits upon which any non-trivial logical operator acts is referred to as the \textit{distance}.
A stabilizer code $\mathcal{Q}(\mathcal{S})$ with the parameters of the number of physical qubits $N$, the number of logical qubits $K$, and distance $D$ is denoted by an $[[N, K, D]]$ code.
Furthermore, the rate $R$ of the code $\mathcal{Q}(\mathcal{S})$ is defined as $R\coloneqq K/N$.

A Calderbank-Shor-Steane (CSS) code~\cite{Calderbank_1996,Steane_1996} is a stabilizer code that can be constructed from a pair of classical linear codes.
A classical linear code with a block length $N$ and dimension $K$ is a linear subspace $C$ of a vector space $\mathbb{F}_2^{N}$.
A linear code $C$ is defined as the kernel of an $M\times N$ parity-check matrix $H$, i.e., $C=\{x\in \mathbb{F}_2^{N}: Hx^{\top}=0\}$, where $M\geq N-K$ holds equality when $H$ is full rank, and $x$ is a row vector representing a codeword.
In addition, a linear subspace of a vector space $\mathbb{F}_2^{N}$ spanned by the set of vectors that are orthogonal to all codewords in $C$ is known as the dual code $C^{\perp}$ of $C$, defined as $C^{\perp}\coloneqq \{d\in\mathbb{F}_2^{n}\colon d\oplus c=0, \forall c\in C\}$.
Given a pair of classical linear codes $C_X=\ker{H_X}$ and $C_Z=\ker{H_Z}$ with a block length $N$ satisfying $C_Z^{\perp} \subseteq C_X$, we can define a parity-check matrix $H\in \mathbb{F}_2^{M\times 2N}$ of the CSS code as
\begin{equation}
    H = \begin{bmatrix}
    H_X & 0 \\
    0 & H_Z
\end{bmatrix},
\label{eq: parity-check matrix of the CSS code}
\end{equation}
where $M=M_X+M_Z$.
The parity-check matrix $H\in\mathbb{F}_2^{M\times 2N}$ of the CSS code gives the stabilizer group $\mathcal{S}$ of the code generated by $\{g_i\}_i$ such that
\begin{equation}
\begin{split}
    H_X^{\top}&=[\phi(g_1),\ldots,\phi(g_{M_X})],\\
    H_Z^{\top}&=[\phi(g_{M_X+1}),\ldots,\phi(g_{M})].
\end{split}
\label{eq: relation between parity-check matrices and generators}
\end{equation}
From the construction in~\eqref{eq: parity-check matrix of the CSS code}, the stabilizer generators of the CSS code can be classified into $X$-type generators of
\begin{equation}
    g_{m}^{X}\in\left\{\bigotimes_i P_i \colon P_i\in\{I,X\}\right\}\; \mathrm{for}\;m^\prime\in \{1,\ldots,M_X\},
\end{equation}
and $Z$-type generators of
\begin{equation}
    g_{m^\prime}^{Z}\in\left\{\bigotimes_i P_i \colon P_i\in\{I,Z\}\right\}\; \mathrm{for}\;m\in \{M_X+1,\ldots,M\}.
\end{equation}
The condition that these generators commute with each other is equivalent to $H_Z H_X^{\top}=0$, which is guaranteed by the requirement of $C_Z^{\perp} \subseteq C_X$.
The code space of the CSS codes $\mathcal{Q}$ is
\begin{equation}
    \mathrm{span}\left\{\sum_{y\in C_Z^{\perp}}\ket{x\oplus y}: x\in C_X\right\},
\end{equation}
the dimension of a CSS code is $K=K_X+K_Z-N$, where $K_X$ and $K_Z$ are the dimension of $C_X$ and $C_Z$, respectively, and the distance of the code is $D=\min\{D_X, D_Z\}$, where $D_X=\min_{x\in C_X\backslash C_Z^{\perp}}|x|$, $D_Z=\min_{x\in C_Z\backslash C_X^{\perp}}|x|$, and $|x|$ is the Hamming weight of $x$.

\subsubsection{Concatenated Steane codes\label{sec:Concatenated Steane codes}}
A $[[7,1,3]]$ Steane's $7$-qubit code~\cite{Steane_1996}, or simply the Steane code, is a CSS code whose parity-check matrices $H_X$ and $H_Z$ are both $[7,4,3]$ Hamming code~\cite{Hamming1950}, i.e.,
\begin{equation}
    H_X=H_Z=\begin{bmatrix}
  1 & 0 & 1 & 0 & 1 & 0 & 1 \\
  0 & 1 & 1 & 0 & 0 & 1 & 1 \\
  0 & 0 & 0 & 1 & 1 & 1 & 1 \\
\end{bmatrix},
\label{eq: parity-check matrix for [[7,4,1]]-Hamming code}
\end{equation}
which satisfy $H_Z H_X^{\top}=0$.
The stabilizer generators of the Steane code are given by
\begin{equation}
\begin{split}
    g^X_{1}=&X\otimes I\otimes X\otimes I \otimes X\otimes I\otimes X,\\
    g^X_{2}=&I\otimes X\otimes X\otimes I\otimes I\otimes X\otimes X,\\
    g^X_{3}=&I\otimes I\otimes I\otimes X\otimes X\otimes X\otimes X,\\
    g^Z_{1}=&Z\otimes I\otimes Z\otimes I \otimes Z\otimes I\otimes Z,\\
    g^Z_{2}=&I\otimes Z\otimes Z\otimes I\otimes I\otimes Z\otimes Z,\\
    g^Z_{3}=&I\otimes I\otimes I\otimes Z\otimes Z\otimes Z\otimes Z,\\\\
\end{split}
\label{eq: stabilizer generator for Steane code.}
\end{equation}
The logical $X$ and $Z$ operators of the Steane code acting on the logical qubit are described by the following operator acting on the $7$ physical qubits, respectively,
\begin{equation}
\begin{split}
    &I\otimes I\otimes I\otimes I\otimes X\otimes X \otimes X\\
    &I\otimes I\otimes I\otimes I\otimes Z\otimes Z \otimes Z.
\end{split}
\label{eq: logical opeartor for Steane code}
\end{equation}
A logical gate is called \textit{transversal} if it can be implemented by a tensor product of gates acting individually on each physical qubit.
Transversality is useful for FTQC since the logical gate can be performed in a single time step without additional qubits and, more importantly, they inherently prevent error propagation occurring in a code block.
The feature of the Steane code is that the logical Clifford gates $H$, $S$, and CNOT can be executed transversally by applying $H^{\otimes 7}$, $(S^{\dagger})^{\otimes 7}$, and CNOT$^{\otimes 7}$, respectively, as well as the logical $Z$ and $X$.

A decoding algorithm for the Steane code is as follows.
The CSS code corrects errors using two types of stabilizer generators, $X$-type generators and $Z$-type generators, independently.
The $X$-type generators are used for correcting phase-flip ($Z$) errors and the $Z$-type generators are used for bit-flip ($X$) errors, respectively.
In the following, we will explain error correction using $X$-type generators.
By definition, the $X$-type generator is derived from the parity-check matrix $H_Z$ of the $[7,4,3]$ Hamming code in~\eqref{eq: parity-check matrix for [[7,4,1]]-Hamming code}, and thus we can use the same decoding algorithm for the $[7,4,1]$ Hamming code.
Suppose that a single-qubit $Z$ error occurs, and the syndrome bits $\sigma(E) = (s_1, s_2, s_3)\in\mathbb{F}_2^{3}$ are obtained, where the $i$-th element corresponds to the measurement outcome $m\in\{0,1\}$ of $g_i^{X}$ in~\eqref{eq: stabilizer generator for Steane code.}.
Due to the properties of the Hamming code, the syndrome bits indicate the position of the qubit where the error occurred, given by
\begin{equation}
    \Tilde{i}=\sum_{i=1}^{3} s_i 2^{i-1},
\label{eq: decoding algorithm for Steane code}
\end{equation}
which gives the $Z$ operator on the $\Tilde{i}$-th qubit as the recovering operation.

The concatenated Steane codes~\cite{knill1996concatenated,aliferis2005quantum} are constructed by recursively encoding each physical qubit of a Steane code into a logical qubit of the Steane code.
Starting with a single qubit at level $L$, this recursive encoding process is repeated for each level $l\in\{L,\ldots,1\}$, resulting in the level-$L$ concatenated Steane code, denoted by $\mathcal{Q}^{(L)}$, with parameters,
\begin{equation}
[[N=7^L,K=1,D=3^L]].
\end{equation}
We refer to $L$ as the concatenation level.

\subsubsection{Quantum LDPC codes\label{sec:Quantum LDPC codes}}
A family of CSS codes with parameters $[[N_i, K_i, D_i]]$ indexed by an integer $i$  ($N_i\rightarrow\infty$) is said to be an $(r,c)$ quantum low-density parity-check (LDPC) code if the parity-check matrices of the CSS codes in the family have at most $r=O(1)$ non-zero elements in each row and at most $c=O(1)$ non-zero elements in each column as $i\rightarrow\infty$.
That is, the $X$-type and $Z$-type stabilizer generators of the quantum LDPC codes have only $r=O(1)$ weights, and only a constant number $c=O(1)$ of stabilizer generators act nontrivially on each physical qubit of the quantum LDPC codes. 
If the rate $R_i$ of a family of quantum LDPC codes converges to a finite positive value $R$ as $i\rightarrow\infty$,
\begin{equation}
    \lim_{i\rightarrow \infty}{R_i}=R>0,
\end{equation}
a code in the family is referred to as a \textit{non-vanishing-rate} quantum LDPC code.
Conversely, if the rate $R$ converges to zero, a code in the family is referred to as a \textit{vanishing-rate} quantum LDPC code.

Vanishing-rate quantum LDPC codes, such as surface codes~\cite{Kitaev_2003,bravyi1998quantum,Litinski_2019,Fowler_2012,Vasmer_2019,Dennis_2002} and color codes~\cite{Bombin_2006,Bombin_2007,Kubica_2015,bombin2015gauge}, are capable of implementing most or all logical Clifford gates transversally~\cite{Kubica_2015,bombin2015gauge,Vasmer_2019,Dennis_2002,Bombin_2011}.
However, a drawback of these codes is that the number of physical qubits required to protect a single logical qubit diverges asymptotically as the code size grows.
In contrast, non-vanishing-rate quantum LDPC codes are advantageous since such codes can be used for constant-space-overhead FTQC, where the rate of the code converges to a non-zero value~\cite{Fawzi_2018, Gottesman2014Constant,grospellier:tel-03364419}.
Although implementing logical Clifford gates transversally with non-vanishing-rate LDPC codes is challenging~\cite{Krishna_2021,Quintavalle2023partitioningqubits, Yamasaki_2024, yoshida2024concatenate}, these codes can implement logical gates via gate teleportation if the required auxiliary encoded states can be prepared in a fault-tolerant way.

\section{Setting of fault-tolerant quantum computation}
\label{sec:setting}
In this section, we present the setting of fault-tolerant quantum computation in our work.
Quantum computation is represented by quantum circuits, and we consider a sequence $\{C_n^{\mathrm{org}}\}_n$ of original circuits specified by an integer $n$.
Each original circuit is composed of the following quantum operations: $\ket{0}$-state preparation, $Z$-basis measurement, Clifford gates ($H$, $S$, and $\mathrm{CNOT}$), non-Clifford gates ($T\coloneqq R_Z(\pi/4)$ and $T^{\dagger}$), and a wait operation $I$.
We consider original circuits that start with the $\ket{0}$-state preparations and end with the $Z$-basis measurements, without any measurements in the middle of the circuits~\cite{nielsen00}.
The unitary part of the original circuits sandwiched in between is composed of Clifford gates ($H$, $S$, and CNOT), non-Clifford gates ($T\coloneqq R_Z(\pi/4)$ and $T^{\dagger}$), and a wait operation $I$~\cite{nielsen00}.
These circuits, which output measurement outcomes, are referred to as \textit{closed} quantum circuits, distinguished from \textit{open} quantum circuits that do not end with measurements and output a quantum state.
Note that we consider closed circuits as original circuits, even if not explicitly stated.

The quantum operations in each original circuit can be executed simultaneously on all qubits, under the condition that at most one operation acts on each qubit in the circuit at any given time step.
The total number of the time steps of $C_n^\org$ is referred to as the \textit{depth}, denoted by $D(n)$, and the number of qubits involved at any time step is referred to as the \textit{width}, denoted by $W(n)$.
The size (number of locations) of $C_n^{\mathrm{org}}$, denoted by $|C_n^\org|$.
Note that since there are no measurements in the middle of the circuits, the width of an original circuit remains unchanged throughout the circuit.
In addition, we assume that the original circuit satisfies 
\begin{align}
\label{eq:assumption_W}
W(n)&\to\infty,\\
D(n)&=O(\mathrm{poly}(W(n))),
\end{align} 
as $n\to\infty$; this leads to
\begin{align}
    |C_n^\org|= W(n)D(n)=O(\mathrm{poly}(W(n))).
    \label{eq:assumption_circuit_size}
\end{align}

The goal of FTQC is to simulate an original circuit $C^{\mathrm{org}}_n$, i.e., to sample from a probability distribution that is close to that of $C^{\mathrm{org}}_n$ within $\varepsilon$ in the total variation distance for given $\varepsilon>0$.
To achieve this, we provide a fault-tolerant protocol that explicitly constructs a physical circuit $C^{\mathrm{FT}}_n$ with fault tolerance (referred to as a \textit{fault-tolerant circuit}) to be executed on quantum devices that may suffer from noise.
When executing a physical circuit during computation, noise may impair the results of the computation, which the FTQC aims to overcome.
We assume that the physical circuit satisfies the following conditions, following the convention of the previous works~\cite{aliferis2005quantum,aliferis2007accuracy,Gottesman2014Constant,Fawzi_2018, Yamasaki_2024}.
\begin{assumption}[Physical circuit]\label{assump: physical circuit}
A physical circuit satisfies the following conditions.
\begin{enumerate}
    \item Set of physical operations:
    
    A physical circuit is composed of the following set of physical operations: $\ket{0}$-state preparation, $\ket{T}$-state preparation, $Z$-basis measurement, Pauli gates $X$, $Y$, and $Z$, Clifford gates $H$, $S$, $S^{\dagger}$, and $\mathrm{CNOT}$, and the wait operation $I$, where $\ket{T}\coloneqq\frac{1}{\sqrt{2}}\qty(\mathrm{e}^{-\mathrm{i}\pi/8}\ket{0}+\mathrm{e}^{\mathrm{i}\pi/8}\ket{1})$.
    Each physical operation in a physical circuit is called a physical location; i.e., the physical circuit may be considered as a set of physical locations.

    \item No geometrical constraint on CNOT gates:
    
    We assume that the CNOT gate can be applied to an arbitrary pair of physical qubits at a single time step.
    This setting is motivated by physical platforms allowing long-range interactions such as neutral atom systems~\cite{sunami2024scalablenetworkingneutralatomqubits,Bluvstein_2023, xu_constant-overhead_2024}, trapped-ion systems~\cite{PhysRevX.13.041052,PhysRevX.11.041058,egan2021faulttolerant}, and photonic systems~\cite{Bourassa2021blueprintscalable,litinski2022active,yamasaki2020polylogoverhead}.
    
    \item Parallel physical operations:
    
    The physical operations can be performed simultaneously on all qubits in a single time step as long as each qubit is involved in at most one physical operation at any given time step.
    The depth of a physical circuit is determined by these time steps.

    \item Allocation of qubits and bits:
    
    We allocate physical qubits via the $\ket{0}$-state preparations and deallocate them through the $Z$-basis measurements. 
    Bits are allocated in a classical register to store measurement outcomes obtained from the $Z$-basis measurements.
    If no longer used in the rest of the computation, the bits are deallocated from the classical register.
    Tracing out qubits is considered to be the $Z$-basis measurements followed by such deallocation to forget measurement outcomes.
    The number of physical qubits at any given time is those that have been allocated previously and have not yet been deallocated.

    \item Noiseless classical computation with non-negligible runtime:
    
    For simplicity in our analysis, we assume that classical computation can be performed without faults.
    Additionally, during the classical computation, we assume that the wait operations are performed on all allocated physical qubits in a physical circuit.
    The depth of the physical circuit performing the wait operations is determined by the runtime of the classical computation. 
    Specifically, this depth is limited to a value less than or equal to a constant multiple of the runtime of the classical computation.
\end{enumerate}
\end{assumption}

We define time steps of the circuits, as shown in Fig.~\ref{fig:circuit-time}.
In a depth-$T$ circuit, for each $t\in\{0,1,\ldots,T\}$, operations at the $(t+1)$-th depth is placed between the time steps $t$ and $t+1$.
The time step $t=0$ is before the first operations, and $t=T$ is after the last operations.

When analyzing the impact of noise on FTQC with quantum LDPC codes, it is conventional to consider the \textit{local stochastic Pauli error model} on a fault-tolerant circuit $C^{\mathrm{FT}}_n$~\cite{Gottesman2014Constant, Fawzi_2018, grospellier:tel-03364419}, which is also used in this paper, as defined in Def.~\ref{def: Local stochastic error model on physical circuit} below.
This error model has a stochastic property, where a Pauli error is applied to a quantum state with probability $p$, and the state remains unchanged with probability $1-p$.
It also satisfies a local property, where the probability of an error occurring at any given set of locations decreases exponentially in terms of the size of the set. 
However, it does not impose any other constraints; in particular, errors may be correlated in both time and space in our setting.

\begin{definition}[Local stochastic Pauli error model on fault-tolerant circuit]
\label{def: Local stochastic error model on physical circuit}
    Let $C^{\mathrm{FT}}$ be a fault-tolerant circuit.
    For simplicity, the set of all the physical locations of the fault-tolerant circuit is also denoted by $C^{\mathrm{FT}}$.
    A set of faulty locations is given by a random variable $F\subseteq C^{\mathrm{FT}}$ of locations of the circuit, where any Pauli errors may occur.
    We say that $C^{\mathrm{FT}}$ experiences a local stochastic Pauli error model with error parameters $\{p_i\}_{i\in C^{\mathrm{FT}}}$ if the following conditions hold.
    \begin{enumerate}
        \item Each physical location $i\in C^{\mathrm{FT}}$ has an error parameter $p_i\in(0,1]$ such that for all $A\subseteq C^{\mathrm{FT}}$, the probability that $F$ contains $A$ satisfies
    \begin{equation}
        \mathbb{P}[F\supseteq A]\leq \prod_{i\in A} p_i.
    \label{eq: local stochastic error model on locations}
    \end{equation}
    Note that in the case of $A=\varnothing$, it holds that $\prod_{i \in \varnothing } p_i=1$, even without explicitly stating in this paper.

    \item The faulty operations in $F$ may suffer from errors represented by the assignment of arbitrary Pauli operators, which are applied to the state at each time step (between the locations) in the following way.
    \begin{itemize}
        \item If a location in $F$ is a $\ket{0}$-state preparation, any Pauli operator may be applied to the qubit after the $\ket{0}$-state preparation is performed.
        \item If a location in $F$ is a gate operation (i.e., Pauli gates, Clifford gates, non-Clifford gates, and a wait operation), then after the gate operation is performed, any Pauli operators may be applied to the qubit(s) on which the gate operation has acted.
        \item If a location in $F$ is a $Z$-basis measurement operation, a bit-flip (or identity) operation may be applied to the measurement outcome.
    \end{itemize}
    The physical locations in $C^{\mathrm{FT}}\backslash F$ behave in the same way as the case without faults. 
    \end{enumerate}
\end{definition}

We refer to a physical circuit that suffers from local stochastic Pauli errors as a \textit{faulty} physical circuit.
An example of a faulty physical circuit is shown in Fig.~\ref{fig:circuit-time}.
In the following, we assume that a fault-tolerant circuit $C_n^{\mathrm{FT}}$ experiences the local stochastic Pauli error model, where the error parameter at each location is assumed to be upper bounded by the same parameter
$p_{\mathrm{loc}}\in(0,1]$, i.e., $p_i\leq p_{\mathrm{loc}}$ for all $i\in C_n^{\mathrm{FT}}$.

\begin{assumption}
\label{assumption: fault-tolerant circuit experiences the local stochastic Pauli error model}
For a constant $p_\mathrm{loc}\in(0,1]$, a fault-tolerant circuit $C_n^{\mathrm{FT}}$ experiences the local stochastic Pauli error model with error parameters such that
    \begin{align}
        p_i\leq p_{\mathrm{loc}}
    \end{align}
    for all locations $i\in C_n^{\mathrm{FT}}$.
\end{assumption}

\begin{figure}
    \centering
    \includegraphics[width=3.1in]{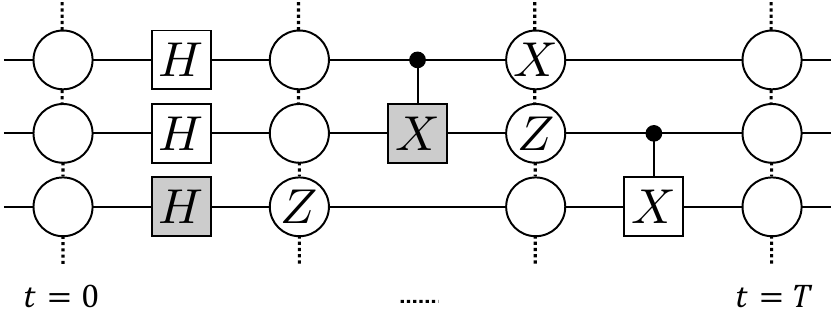}
    \caption{An example of a faulty physical circuit. The operations shaded in grey represent the faulty locations. The circles after the faulty locations represent Pauli errors on wires at a given time step.}
    \label{fig:circuit-time}
\end{figure}

A fault-tolerant protocol provides a fault-tolerant physical circuit $C_n^{\mathrm{FT}}$ to simulate an original circuit $C_n^{\mathrm{org}}$ by replacing qubits in an original circuit with logical qubits of a quantum error-correcting code. This process requires using multiple physical qubits per logical qubit for redundancy. 
At the same time, operations in an original circuit are replaced by logical operations acting on the logical qubits.
Implementing these logical operations in a fault-tolerant circuit requires additional time steps.
This procedure incurs an overhead with respect to space and time.
Let $W_{\mathrm{FT}}(n,\varepsilon)$ denote the width and $D_{\mathrm{FT}}(n,\varepsilon)$ denote the depth of the fault-tolerant circuit $C_n^{\mathrm{FT}}$.
Then, the space overhead is defined as the ratio of the width $W_{\mathrm{FT}}(n,\varepsilon)$ of the fault-tolerant circuit to the width $W(n)$ of the corresponding original circuit,
\begin{equation}
\frac{W_{\mathrm{FT}}(n,\varepsilon)}{W(n)}.
\end{equation}
On the other hand, the time overhead is defined as the ratio of the depth $D_{\mathrm{FT}}(n,\varepsilon)$ of the fault-tolerant circuit, taking into account wait operations for the runtime of classical computation, to the depth $D(n)$ of the original circuit,
\begin{equation}
\frac{D_{\mathrm{FT}}(n,\varepsilon)}{D(n)}.
\end{equation}
For any fixed target error $\varepsilon>0$, the fault-tolerant
protocol is said to achieve a polylogarithmic time overhead if the time overhead of the fault-tolerant circuit is
\begin{align}
\label{eq:time_overhead_definition}
    \frac{D_{\mathrm{FT}}(n,\varepsilon)}{D(n)}&=O\left(\mathrm{polylog}\left(\frac{|C_n^\mathrm{org}|}{\varepsilon}\right)\right)\quad\text{as $n\rightarrow\infty$},
\end{align}
where $|C_n^\mathrm{org}|$ in~\eqref{eq:assumption_circuit_size} is the size of the original circuit.
The fault-tolerant protocol is said to achieve a constant space overhead if
\begin{equation}
\label{eq:space_overhead_definition}
    \frac{W_{\mathrm{FT}}(n,\varepsilon)}{W(n)}=O(1)\quad\text{as $n\to\infty$}.
\end{equation}
Note that in our analysis, we fix $\varepsilon$ and analyze asymptotic scaling of the space and time overheads as $n \to \infty$; however, the same analysis also holds when $\varepsilon$ decays as $\varepsilon(n)=O\qty(\frac{1}{\mathrm{poly}(W(n))})$.
To account for this, we will explicitly include the scaling of $\varepsilon$ in our bounds of the overheads as in~\eqref{eq:time_overhead_definition} and ~\eqref{eq:space_overhead_definition}.

\section{Hybrid fault-tolerant protocol of quantum LDPC codes and concatenated codes\label{sec: Description of polylog-time constant-space overhead protocol}}

In this section, we describe our fault-tolerant protocol that achieves polylogarithmic time and constant space overhead, using non-vanishing-rate quantum LDPC codes and concatenated codes in combination.
In Sec.~\ref{sec: Compilation of ideal quantum circuit into fault-tolerant quantum circuit},
we explain the compilation procedure from an original circuit into a fault-tolerant circuit.
In Sec.~\ref{sec: Construction of abbreviations}, we describe the constructions of abbreviations for applying Clifford gates, and $T$- and $T^{\dagger}$-gates used in this compilation while each gadget will be shown in Sec.~\ref{sec:gadgets}.
Note that this protocol will be proven to achieve polylogarithmic time and constant space overhead with a non-zero threshold value in Sec.~\ref{sec: threshold theorem based on quantum LDPC codes}.

\subsection{Compilation of original quantum circuit into fault-tolerant quantum circuit \label{sec: Compilation of ideal quantum circuit into fault-tolerant quantum circuit}}

We describe the procedure of our fault-tolerant protocol, which compiles an original circuit $C_n^{\mathrm{org}}$ with width $W(n)$, depth $D(n)$, and a total number of locations $|C_n^{\mathrm{org}}|=W(n)D(n)$, where $|C_n^{\mathrm{org}}|\rightarrow\infty$ as $i\rightarrow\infty$, into a fault-tolerant circuit that achieves a target error for fault-tolerant simulation
\begin{equation}
    \varepsilon>0.
\label{eq: target error of fault-tolerant simulation}
\end{equation}
We use a family of non-vanishing-rate CSS $(r,c)$ LDPC codes
$\{\mathcal{Q}_i\}_i$
in combination with the protocol with the concatenated Steane codes to simulate open circuits in Appendix~\ref{appendix: fault-tolerant protocol for open quantum circuits}.
Here, assume each $\mathcal{Q}_i$ has parameters
\begin{equation}
    [[N_i, K_i=\Theta(N_i), D_i=\Theta(N_i^{\gamma})]],
\end{equation}
where $N_i$ represents the number of phyiscal qubits, $K_i$ represents the number of logical qubits, $D_i$ is the code distance, and $\gamma>0$ is a constant.
We will choose $N_i$ depending on $\varepsilon$ and $n$ so as to achieve a polylogarithmic time and constant space overhead.
The detailed assumptions on the family $\{\mathcal{Q}_i\}_i$ used in our protocol are provided in Assumption \ref{assump: non-vanishing-rate quantum LDPC with efficient decoding algorithm}.
In the following, for simplicity, we will consider a specific code $\mathcal{Q}$ in $\{\mathcal{Q}_i\}_i$ with 
\begin{equation}
    \label{eq: code parameters of quantum LDPC codes}
    [[N=N_i, K=K_i, D=D_i]],
\end{equation}
omitting the index $i$ unless specifically stated.

The procedure of our fault-tolerant protocol with quantum LDPC codes for compiling an original quantum circuit to a fault-tolerant circuit is as follows, which is also shown in Fig.~\ref{fig: compilation procedure for ideal closed circuits}.
The compilation procedure consists of two steps: compiling the original circuit into an intermediate circuit using the quantum LDPC code $\mathcal{Q}$ and then compiling the intermediate circuit into a fault-tolerant circuit using the quantum LDPC code $\mathcal{Q}$ and the concatenated Steane code in combination.
In the compilation to the intermediate circuit, the qubits of the original circuit are grouped into registers, with each register containing at most $K$ qubits, where $K$ is the number of logical qubits of the quantum LDPC code $\mathcal{Q}$ in \eqref{eq: code parameters of quantum LDPC codes}.
Elementary operations for the quantum LDPC code are specified as $\ket{0}^{\otimes K}$-state preparation, Clifford-state preparations, magic-state preparations, Pauli-gate operations, a CNOT-gate operation, a $Z_K$-measurement operation, a Bell-measurement operation, and a wait operation.
These elementary operations are defined as operations acting collectively on the qubits in the registers.

For simplicity of presenting our protocol, we also define abbreviations, which are combinations of elementary operations for implementing Clifford and non-Clifford gates via the gate teleportation protocol~\cite{gottesman2010introduction,knill2005scalable,Knill_2005}.
In particular, we will introduce two-register Clifford-gate abbreviations to implement Clifford gates on two registers and $U_T$-gate abbreviations to implement non-Clifford gates on a register.
The intermediate circuits are described using these abbreviations and some of the elementary operations for  $\mathcal{Q}$, where each abbreviation is identified with the corresponding combination of elementary operations.
The construction of the abbreviations will be clarified in Sec.~\ref{sec: Construction of abbreviations}.

For each elementary operation for $\mathcal{Q}$ in the intermediate circuit, we construct a gadget, which is a physical circuit intended to perform the corresponding logical operation acting on logical qubits in a code block of $\mathcal{Q}$.
Also, we construct an error-correction (EC) gadget for the quantum LDPC code intended to carry out quantum error correction on a code block of $\mathcal{Q}$.
The gadgets must be carefully designed to satisfy the fault-tolerance conditions, which we will introduce in Sec.~\ref{sec: Conditions of fault-tolerant gadgets on quantum LDPC codes}.
In addition, the code $\mathcal{Q}$ must have an efficient decoding algorithm for the protocol to have a threshold, even when taking into account nonzero runtime of classical computation for the decoding.
We will also formally define the requirements for an efficient decoding algorithm that the code $\mathcal{Q}$ should fulfill in Sec.~\ref{sec: Conditions of fault-tolerant gadgets on quantum LDPC codes}.
The construction of the gadgets will be detailed in Sec.~\ref{sec: Construction of fault-tolerant gadgets and abbreviations(qLDPC)}.

In the fault-tolerant circuit, each elementary operation in the intermediate circuit is replaced with the corresponding gadget, and an EC gadget acting on each code block is inserted after each gadget(except after the gadgets of the $Z_K$-measurement and the Bell-measurement operations; see Fig~\ref{fig: compilation procedure for ideal closed circuits}).
Each part of the physical fault-tolerant circuit composed of a gadget of each elementary operation for $\mathcal{Q}$ followed by the EC gadgets is referred to as a \textit{rectangle}.
The replacement of each elementary operation in the intermediate circuit with the corresponding rectangle provides a fault-tolerant circuit for simulating the original circuit.
In the following, we will describe these two compilation steps in more detail.

\subsubsection{Compilation from original circuit to intermediate circuit \label{sec:Compilation from original circuit to intermediate circuit}}
As the first step of compilation, the original circuit is compiled into an intermediate circuit that acts on registers of the qubits.
First, we introduce elementary operations and abbreviations used in our protocol, along with their corresponding diagrams.
In the diagrams shown below, a dashed input line changing to a solid output line indicates the allocation of a register, while a solid line changing to a dashed line indicates the deallocation. 
If both the input and output lines are dashed, the corresponding register is used as a workspace for performing elementary operations or abbreviations.
A double output line represents bits allocated by an elementary operation.

The elementary operations are represented as follows.

\begin{itemize}
    \item $\ket{0}^{\otimes K}$-state preparation
\begin{align}
    \includegraphics{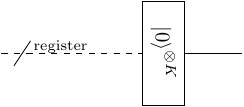}\label{Fig: 0-state-preparation-ldpc}
\end{align}
The $\ket{0}^{\otimes K}$-state preparation allocates a single register that is prepared in the state $\ket{0}^{\otimes K}$.

\item Clifford-state preparations
\begin{align}
    \includegraphics{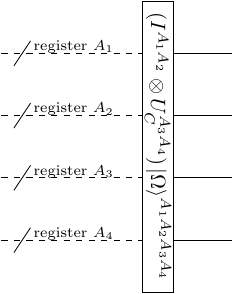}\label{Fig: clifford-state-preparation-ldpc}
\end{align}
A Clifford-state preparation allocates four registers $A_1,A_2,A_3,A_4$ in the state
\begin{equation}
    (I^{A_1 A_2}\otimes U_{C}^{A_3 A_4})\ket{\Omega}^{A_1 A_2 A_3 A_4},
\label{eq: Clifford state in the definition of elementary operation}
\end{equation}
where 
\begin{equation}
    \ket{\Omega}^{A_1 A_2 A_3 A_4}\coloneqq \ket{\Phi}^{A_1 A_3}\otimes\ket{\Phi}^{A_2 A_4},
\label{eq: maximally entangled state between 4 registers}
\end{equation}
$\ket{\Phi}^{A_i A_{i'}}$, which is a maximally entangled state between the registers $A_i$ and $A_{i'}$, is given by
\begin{equation}
    \ket{\Phi}^{A_i A_{i'}}=\frac{1}{\sqrt{2^K}}\sum_{m=0}^{2^K-1}\ket{m}^{A_i}\otimes\ket{m}^{A_{i'}},
\end{equation}
$I^{A_1 A_2}$ is an identity operator acting on the qubits in the two registers $A_1$ and $A_2$, and $U_{C}^{A_3 A_4}\in \Tilde{\mathcal{C}}_{2K}/\Tilde{\mathcal{P}}_{2K}$ is an arbitrary Clifford unitary acting on the qubits in the two registers $A_3$ and $A_4$.

\item Magic-state preparations
\begin{align}
    \includegraphics{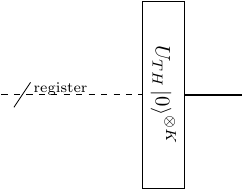}\label{Fig: magic-state-preparation-ldpc}
\end{align}
A magic-state preparation allocates a single register in the state of
\begin{equation}
    U_{TH}\ket{0}^{\otimes K},
\label{eq: magic state in the definition of elementary operation}
\end{equation}
where $U_{TH}$ is a tensor product of any combinations of $TH$ and $I$ on each of the $K$ qubit.
Note that $TH$ acting on $\ket{0}$ yields a conventional magic state $\ket{T}=TH\ket{0}=\frac{1}{\sqrt{2}}\qty(\mathrm{e}^{-\mathrm{i}\pi/8}\ket{0}+\mathrm{e}^{\mathrm{i}\pi/8}\ket{1})$ for the $T$ gate.

\item CNOT-gate operation
\begin{align}
    \includegraphics{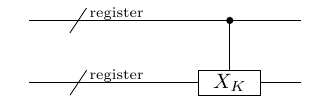}\label{Fig: cnot-gate-ldpc}
\end{align}
The CNOT-gate operation applies the $\text{CNOT}^{\otimes K}$ gates acting on $2K$ qubits in the two registers.
Here, a controlled-$X_{K}$ gate represents $\text{CNOT}^{\otimes K}$ gates.
The $K$ target qubits are all contained in a single register.

\item $Z_K$-measurement operation
\begin{align}
    \includegraphics{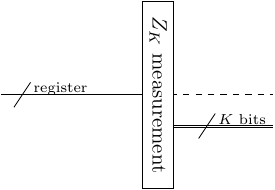}\label{Fig: Z-basis measurements-ldpc.}
\end{align}
The $Z_K$-measurement operation performs collective $Z$-basis measurements on $K$ qubits contained in a register, deallocates the register, and outputs a $K$-bit string of the measurement outcomes.
Here, $Z_K$ denotes a label representing these collective $Z$-basis measurements on $K$ qubits.

\item Bell-measurement operation
\begin{align}
    \includegraphics{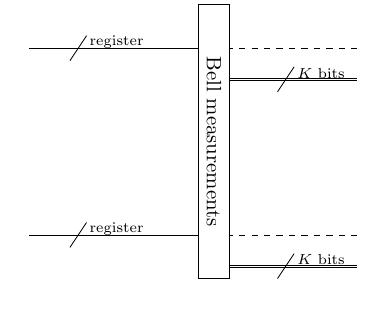}\label{Fig: bell-measurement-ldpc}
\end{align}
The Bell-measurement operation performs the Bell measurements on $K$ pairs of qubits of an $X\otimes X$ operator and a $Z\otimes Z$ operator, where each pair is shared by the two registers. It deallocates these registers and outputs a $2K$-bit string of the measurement outcomes.
We assume that the upper register outputs the $K$-bit string representing the measurement outcomes of the $X\otimes X$ operators, and the lower register outputs the $K$-bit string representing the measurement outcomes of the $Z\otimes Z$ operators.

\item Pauli-gate operations

\begin{align}
    \includegraphics{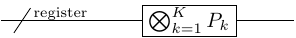}\label{Fig: pauli-gate-ldpc}
\end{align}
A Pauli-gate operation performs a tensor product of arbitrary Pauli gates
\begin{equation}
    \bigotimes_{k=1}^{K} P_k\in\Tilde{\mathcal{P}}_K,
\label{eq: Pauli-gate operations for quantum LDPC codes}
\end{equation}
where $P_k\in\{I, X,Y,Z\}$ is a single-qubit Pauli operator acting on the $k$-th qubit in a register.

\item Wait operation

\begin{align}
    \includegraphics{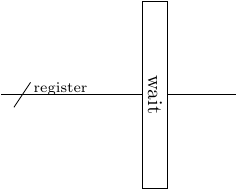}\label{Fig: wait-ldpc}
\end{align} or simply
\begin{align}
    \includegraphics{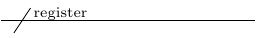}\label{Fig: wait operation}
\end{align}
The wait operation performs the $I$ gate on a register, which is regarded as a special case of the Pauli-gate operation.
\end{itemize}

In addition, the abbreviations, composed of these elementary operations, are represented as follows.
The constructions of the abbreviations will be given in Sec.~\ref{sec: Construction of abbreviations}.
\begin{itemize}
\item Two-register Clifford-gate abbreviations
\begin{align}
    \includegraphics{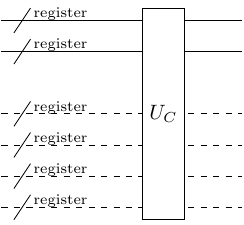}\label{Fig: two-register-clifford-abbreviation-ldpc.}
\end{align}
A two-register Clifford-gate abbreviation applies an arbitrary Clifford unitary $U_C\in \Tilde{\mathcal{C}}_{2K}/\Tilde{\mathcal{P}}_{2K}$ to $2K$ qubits in the two registers represented by the solid lines.
The four registers represented by dashed lines are used as a workspace.
    
\item $U_T$-gate abbreviations
\begin{align}
    \includegraphics{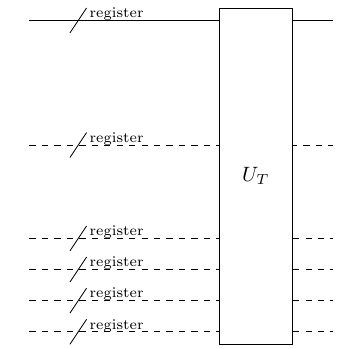}\label{Fig: magic-abbreviation-ldpc.}
\end{align}
A $U_T$-gate abbreviation applies a $U_{T}$ gate to $K$ qubits in the register represented by the solid line, where $U_{T}$ is the tensor product of any combination of $T$, $T^{\dagger}$, and $I$.
The five registers represented by dashed lines are used as a workspace.
\end{itemize}

Some of the elementary operations may take as input a bitstring that is determined during the execution of quantum computation, and the specification of the elementary operation is determined on the fly by this input.
Such elementary operations are called \textit{on-demand} elementary operations.
We also introduce on-demand abbreviations in the same way.
In the following, we list the on-demand elementary operations and the on-demand abbreviations used in our protocol.
\begin{itemize}
    \item On-demand Pauli-gate operations
    
    An on-demand Pauli-gate operation, which appears in Sec.~\ref{sec: Two-register Clifford-gate abbreviation} and Sec.~\ref{sec: U_T-gate abbreviation}, is used for the correction operation required for gate teleportation in the two-register Clifford-gate abbreviation in \eqref{Fig: two-register-clifford-abbreviation-ldpc.} and in the $U_T$-gate abbreviation in \eqref{Fig: magic-abbreviation-ldpc.}.
    The on-demand Pauli-gate operation receives a $2K$-bit string representing the symplectic representation $\phi(P)\in\mathbb{F}_2^{2K}$, where $P\in\Tilde{\mathcal{P}}_K$ is a Pauli operator,
    \begin{equation}
        P=\bigotimes_{k=1}^{K}P_k \in\Tilde{\mathcal{P}}_K
    \label{eq: on-demand Pauli gate operations}
    \end{equation}
    to be applied to $K$ qubits in a register.
    
    \item On-demand two-register Clifford-gate abbreviations
    
An on-demand two-register Clifford-gate abbreviation, which will be presented in Sec~\ref{sec: Two-register Clifford-gate abbreviation}, is used for a correction operation in gate teleportation for the $U_T$-gate abbreviation in \eqref{Fig: magic-abbreviation-ldpc.}.

An on-demand two-register Clifford-gate operation applies a Clifford unitary $U_C \in \Tilde{\mathcal{C}}_{2K}/\Tilde{\mathcal{P}}_{2K}$ on $2K$ qubits across two registers. This gate is expressed as a tensor product of $2K$ single-qubit Clifford gates,
\begin{equation}
U_C = \bigotimes_{k=1}^{K} (C_k \otimes I_k),
\label{eq: on-demand Clifford operation}
\end{equation}
where $C_k \in \{I, S\}$ acts on the $k$-th qubit of the first register, and $I_k$ acts on the $k$-th qubit of the second register. 
The operation takes a $K$-bit string as input: if the $k$-th bit of the string is $1$, then $C_k = S$; if the $k$-th bit is $0$, then $C_k = I$.

    \item On-demand Clifford-state preparation
    An on-demand Clifford-state preparation is  used for preparing a Clifford state in the form of 
    \begin{equation}
        (I\otimes U_C)\ket{\Omega},
    \label{eq: on-demand state preparation}
    \end{equation}
    as described in \eqref{eq: Clifford state in the definition of elementary operation}, where $U_C$ takes the form given in \eqref{eq: on-demand Clifford operation}.
     The on-demand Clifford-state preparation is invoked by an on-demand two-register Clifford-gate abbreviation used for a correction operation for gate teleportation in the $U_T$-gate abbreviation in \eqref{Fig: magic-abbreviation-ldpc.}.
    The on-demand Clifford-state preparation receives a $K$-bit string as input to specify a Clifford unitary $U_C$ for the correction in \eqref{eq: on-demand Clifford operation}.
\end{itemize}

In the compilation prior to starting the execution of quantum computation, elementary operations (and abbreviations) that are not on-demand, i.e., do not require classical input, are referred to as scheduled operations (and abbreviations), which are distinguished from on-demand ones.
All scheduled elementary operations and abbreviations are determined during compilation, independently of on-demand operations, and remain unchanged during the execution of quantum computation.
This reduces the time overhead of waiting for the classical computations during execution.

\begin{figure*}[t]
    \includegraphics[width=\textwidth]{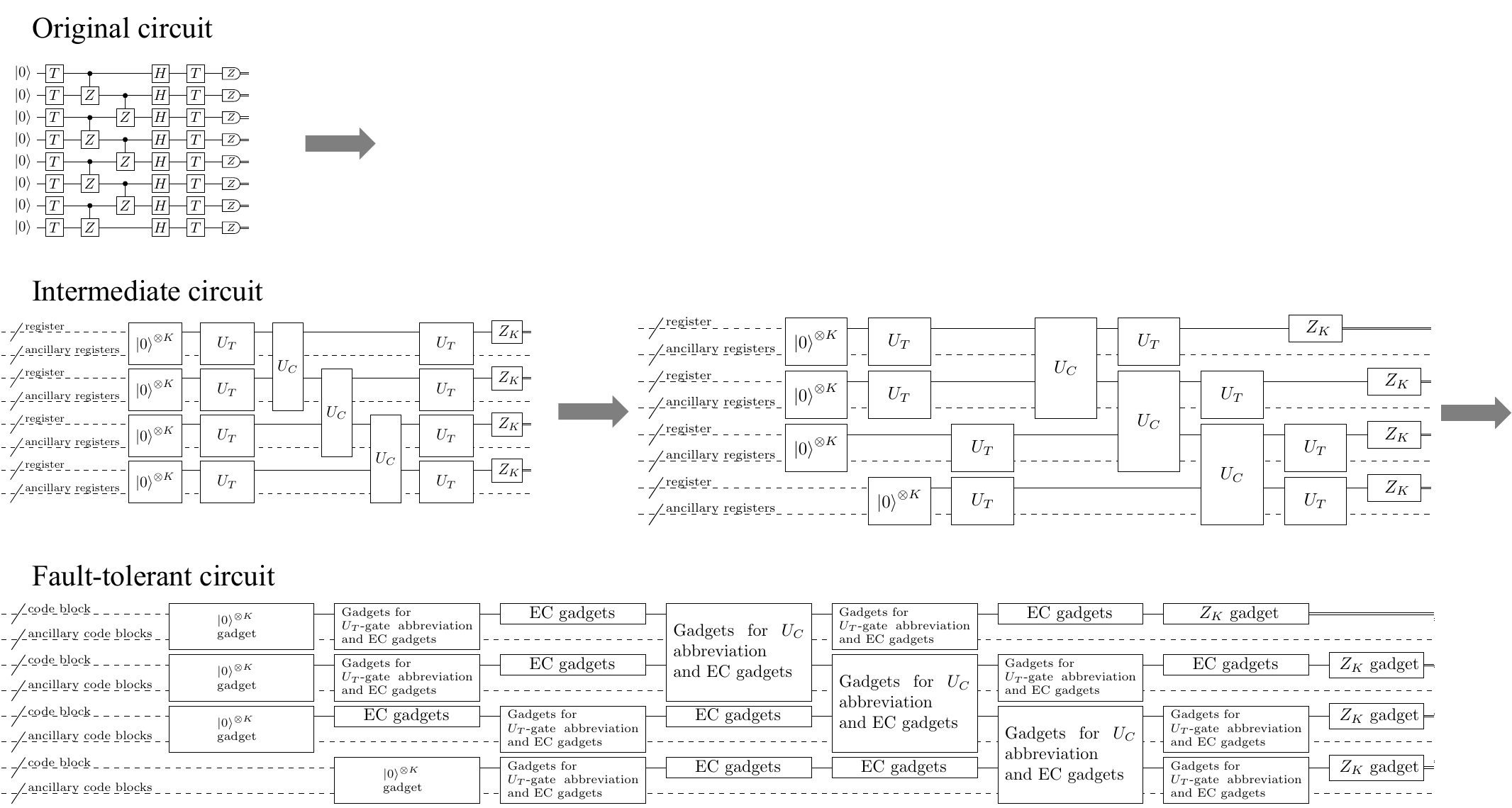}
    \caption{
    Compilation procedure of our fault-tolerant protocol. 
    First, we compile an original circuit into an intermediate circuit that consists only of intermediate operations acting on qubits in registers.
    In this intermediate circuit, a Clifford gate $U_C$ acting on qubits contained in two registers and a $U_T$ gate, which is a tensor product of $T$, $T^{\dagger}$, and $I$, acting on qubits in a single register are abbreviations of a sequence of elementary operations implemented via gate teleportation.
    Then, we limit the number of intermediate operations at each time step to $\Theta(W(n)/\mathrm{polylog}(|C_n^{\mathrm{org}}|))$ so as to achieve constant space overhead, where $W(n)$ is the width of the original circuit and $|C_n^{\mathrm{org}}|$ is the number of locations in the original circuit.
    Finally, we compile the intermediate circuit with reduced parallelism into a fault-tolerant circuit by replacing each elementary operation with the corresponding rectangle.}
    \label{fig: compilation procedure for ideal closed circuits}
\end{figure*}

In constructing an intermediate circuit, we require it should be composed only of the two-register Clifford-gate abbreviations~\eqref{Fig: two-register-clifford-abbreviation-ldpc.}, the $U_T$-gate abbreviations~\eqref{Fig: magic-abbreviation-ldpc.}, the $\ket{0}^{\otimes K}$-state preparation~\eqref{Fig: 0-state-preparation-ldpc}, the $Z_K$-measurement operation~\eqref{Fig: Z-basis measurements-ldpc.}, and the wait operation~\eqref{Fig: wait-ldpc}.
For simplicity, we refer to these abbreviations and elementary operations as intermediate operations.

We compile the original circuit into the intermediate circuit using the following procedure (see the middle of Fig.~\ref{fig: compilation procedure for ideal closed circuits}).
First, $W(n)$ qubits in the original circuit are divided into
\begin{equation}
    \label{eq: kappa in compilation}\kappa(n,\varepsilon)\coloneqq\left\lceil \frac{W(n)}{K}\right\rceil
\end{equation}
registers, where each register contains at most $K$ qubits, and $\lceil{}\cdots{}\rceil$ represents the ceiling function.
For each of these registers, we also allow the intermediate circuit to use five auxiliary registers that can be allocated and used as workspace for the abbreviations.

Next, we replace the operations in the original circuit with the corresponding intermediate operations (see the middle left of Fig.~\ref{fig: compilation procedure for ideal closed circuits}).
The $\ket{0}$-state preparations at the beginning of the original circuit are replaced with $\ket{0}^{\otimes K}$-state preparations.
We replace the $Z$-basis measurements at the end of the original circuit with $Z_K$-measurement operations.
To replace the operations in the middle of the circuit, we first decompose each single-depth part of the original circuit, located between the $\ket{0}$-state preparations and the $Z$-basis measurements, into two single-depth circuits: one composed of $T$, $T^{\dagger}$, and $I$ gates, and the other composed of Clifford gates.
We then replace each single-depth circuit described by $T$, $T^{\dagger}$ and $I$ gates with $U_T$-gate abbreviations, and each described by Clifford gates with two-register Clifford-gate abbreviations.
Note that a single use of a two-register Clifford-gate abbreviation can indeed replace an arbitrarily long sequence of Clifford gates acting on qubits in the same pair of registers of the intermediate circuit, but for simplicity of presentation, we here use it to replace only a one-depth part of Clifford gates, leaving such an optimization for future work.
In this replacement, if a one-depth part of the original circuit contains multiple Clifford gates acting on qubits in different pairs of registers, the corresponding part of the intermediate circuit may require multiple two-register Clifford-gate abbreviations placed in series, as will be described in Sec.~\ref{sec: threshold theorem based on quantum LDPC codes}.

After replacing operations in the original circuit with intermediate operations, we limit the number of non-trivial intermediate operations (i.e., operations other than the wait operations) that can be applied in parallel at a single time step, and insert wait operations for all registers where a non-trivial intermediate operation is not applied (see the middle right of Fig.~\ref{fig: compilation procedure for ideal closed circuits}).
Note that completely parallel execution of $Z_K$-measurement operation does not ruin the constant-space overhead FTQC, but for the simplicity of the analysis we perform the operation in a sequential way.
The upper bound of the number of non-trivial intermediate operations applied in parallel at a single time step is denoted by a parallelization parameter
\begin{equation}
\tau(n,\varepsilon)\in\{1,2,\ldots, W(n)\},
\label{eq: the maximum number of non-trivial intermediate operations L(n,varepsilon)}
\end{equation}
which will turn out to be able to be chosen as~\eqref{eq: tau in threshold theorem} later in our analysis; in particular, progressing beyond the existing analyses~\cite{Gottesman2014Constant, Fawzi_2018, grospellier:tel-03364419} with $\tau(n,\varepsilon)=\Theta(W(n)/\mathrm{poly}(|C_n^{\mathrm{org}}|/\varepsilon))$, we will show that an increased parallelization $\tau(n,\varepsilon)=\Theta(W(n)/\mathrm{polylog}(|C_n^{\mathrm{org}}|/\varepsilon))$ is possible to prove that polylogarithmic time overhead is achievable.

\subsubsection{Compilation from intermediate circuit to fault-tolerant circuit \label{subsubsec: Compilation from intermediate circuit to fault-tolerant circuit}}
In the next step of the compilation, we compile the intermediate circuit into a fault-tolerant circuit.
In this step, all elementary operations in the intermediate circuit are replaced with the gadgets (with EC gadgets inserted to form rectangles).

For each elementary operation, we define a gadget satisfying the fault-tolerant conditions that we will introduce in Sec.\ref{sec: Conditions of fault-tolerant gadgets on quantum LDPC codes}, where the constructions are described in Sec.~\ref{sec: Construction of fault-tolerant gadgets and abbreviations(qLDPC)}.
We may use the same diagrams as the elementary operations on quantum LDPC codes to represent the corresponding gadgets.
We here introduce EC gadgets on quantum LDPC codes are represented as follows.
\begin{itemize}
    \item EC gadget
\begin{align}
    \includegraphics[width=0.25\textwidth]{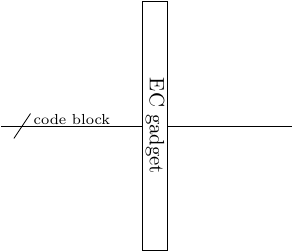}\label{Fig: EC gadget for quantum LDPC codes}
\end{align}
An EC gadget performs quantum error correction on a code block $\mathcal{Q}$ consisting of $N$ physical qubits.
The construction of the EC gadget will be described in \ref{sec: error-correction gadget}.
The EC gadget for quantum LDPC codes does not require auxiliary code blocks, unlike the EC gadget for concatenated Steane codes shown in Appendix~\ref{appendix: fault-tolerant protocol for open quantum circuits}.
\end{itemize}

Each gadget consists of a \textit{quantum} part and a \textit{classical} part.
The quantum part is described by physical operations that need to be applied to physical qubits, whereas the classical part involves the necessary classical computations for executing the gadget.
During the execution of the classical part, the wait operations act on physical qubits in the same way as the quantum part.

\begin{figure*}[t]
    \centering
    \includegraphics[width=\textwidth]{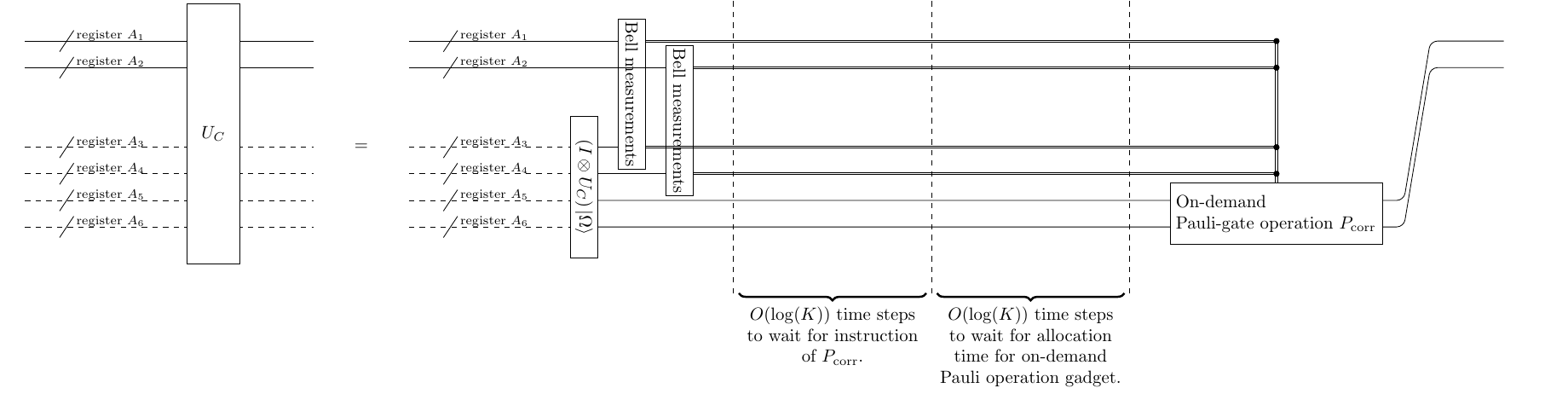}
    \caption{
    The construction of the two-register Clifford-gate abbreviation in~\eqref{Fig: two-register-clifford-abbreviation-ldpc.} performs a Clifford gate $U_C$ applied to the two registers $A_1$ and $A_2$ through gate teleportation.
First, a Clifford-state preparation in \eqref{Fig: clifford-state-preparation-ldpc} prepares the state $(I\otimes U_{C})\ket{\Omega}$ in auxiliary registers $A_3$, $A_4$, $A_5$, and $A_6$.
Next, Bell-measurement operations in~\eqref{Fig: bell-measurement-ldpc} are performed on two pairs of registers, $A_1, A_3$ and $A_2, A_4$. Following the Bell-measurement operations, wait operations with depth $O(\log(K))$ are inserted. These wait operations provide time for calculating the Pauli correction operation $P_{\mathrm{corr}}\in\Tilde{\mathcal{P}}_{2K}$ based on the measurement outcomes, which is necessary for performing gate teleportation. Furthermore, additional wait operations with depth $O(\log(K))$ are inserted to allocate time for the classical computation required by the on-demand Pauli-gate gadget to determine the physical Pauli operations that implement the logical Pauli operation corresponding to $P_{\mathrm{corr}}$ during the execution of the quantum computation.
The on-demand version reduces the wait time required to obtain the instruction for $P_{\mathrm{corr}}$ to $O(1)$.}
    \label{Fig: two-register-clifford-abbreviation.}
\end{figure*}
Compilation from the intermediate circuit to a fault-tolerant circuits is performed by the following procedure (see the bottom of Fig.~\ref{fig: compilation procedure for ideal closed circuits}).
First, for a given time step of the intermediate circuit, the abbreviations belonging to the time step are replaced with a sequence of elementary operations, while elementary operations remain unchanged. 
Since the depths of the intermediate operations may differ, the differences in depth that arise when expanding intermediate operations into a sequence of elementary operations are synchronized using wait elementary operations.
Next, sequences of elementary operations belonging to the same depth in the intermediate circuit is compiled to fault-tolerant circuits. 
This compilation process involves replacing each elementary operation with the corresponding rectangle; however, since the physical depths of different gadgets may differ, synchronization is required.
Assuming the depth of an EC gadget is a fixed constant $d$, wait physical operations are added so that the depth of all gadgets becomes an odd multiple of $d$, which is greater than or equal to the physical depth of the original gadgets.
With the EC gadget, rectangles are defined as follows to ensure synchronized depth alignment across gadgets. A gate rectangle, for the CNOT-gate operation in \eqref{Fig: cnot-gate-ldpc}, Pauli-gate operations in \eqref{Fig: pauli-gate-ldpc}, and the wait operation in \eqref{Fig: wait-ldpc}, is defined as a gate gadget, with wait operations added as needed to achieve a depth of $d$, followed by an EC gadget acting on each of the code blocks. 
Preparation rectangles, for the $\ket{0}^{\otimes K}$-state preparation in \eqref{Fig: 0-state-preparation-ldpc}, the Clifford-state preparations in \eqref{Fig: clifford-state-preparation-ldpc}, and the magic-state preparations in \eqref{Fig: magic-state-preparation-ldpc}, are similarly padded to reach an odd multiple of 
$d$, followed by an EC gadget acting on the code blocks.
Measurement rectangles for the $Z_K$-measurement operation in \eqref{Fig: Z-basis measurements-ldpc.} and the Bell-measurement operation in \eqref{Fig: bell-measurement-ldpc} do not require padding with wait operations within the gadget itself. 
However, for synchronization, measurement rectangles are treated as having a depth that is an odd multiple of $d$ to account for any necessary classical computations after physical measurement operations are completed.
Using these rectangles, we replace the sequence of elementary operations as follows.
First, all elementary operations, excluding wait operations, are replaced with their corresponding rectangles. Then, if the corresponding elementary operations originally assigned to the same depth have differing physical depths, these gadgets are padded with wait rectangles to achieve synchronized depth. Finally, wait elementary operations are replaced with wait rectangles of depth $2d$ (comprising a depth $d$ wait operation and a depth $d$ EC gadget) to align with the rectangles of the corresponding elementary operations within the same time step. 
Applying this procedure for all the time steps of the intermediate circuit results in the fault-tolerant circuit.

\subsection{Constructions of abbreviations\label{sec: Construction of abbreviations}}
In the following, we explain the constructions of abbreviations.
In Sec.~\ref{sec: Two-register Clifford-gate abbreviation}, we present the construction of two-register Clifford-gate abbreviations in~\eqref{Fig: two-register-clifford-abbreviation-ldpc.}.
In Sec.~\ref{sec: U_T-gate abbreviation}, we present the construction of $U_T$-gate abbreviations in~\eqref{Fig: magic-abbreviation-ldpc.}.
We also check that the abbreviations have the depth
\begin{align}
O(\log K),
\end{align}
which is dominated by the runtime of classical computation in the two-register Clifford-gate abbreviations.

\subsubsection{Two-register Clifford-gate abbreviations \label{sec: Two-register Clifford-gate abbreviation}}

A two-register Clifford-gate abbreviation is used to apply a Clifford unitary $U_C\in \Tilde{\mathcal{C}}_{2K}/\Tilde{\mathcal{P}}_{2K}$ acting on $2K$ qubits contained in two registers.
The construction of the abbreviation is shown in Fig.~\ref{Fig: two-register-clifford-abbreviation.}.

The abbreviation is based on the gate teleportation protocol~\cite{gottesman2010introduction,knill2005scalable,Knill_2005}.
In the beginning, we have two registers $A_1$ and $A_2$ on which we want to perform $U_C$, and the four auxiliary registers $A_3, A_4, A_5, A_6$ in a state
\begin{equation}
\begin{split}
    \ket{\Psi_{U_C}}^{A_3 A_4 A_5 A_6} = (I^{A_3 A_4}\otimes U_{C}^{A_5 A_6})\ket{\Omega}^{A_3 A_4 A_5 A_6},
\end{split}
\end{equation}
which is prepared by the Clifford-state preparation in~\eqref{Fig: clifford-state-preparation-ldpc}, where $\ket{\Omega}^{A_3 A_4 A_5 A_6}$ is the maximally entangled state between $A_3$, $A_5$ and $A_4$, $A_6$ as in~\eqref{eq: maximally entangled state between 4 registers}.
After preparing the auxiliary state, we perform the Bell-measurement operation in~\eqref{Fig: bell-measurement-ldpc} on two pairs of registers $A_1, A_3$ and $A_2, A_4$.
These Bell measurements output a pair of $2K$-bit outcomes as
\begin{equation}
    (x^{A_1 A_3},z^{A_1 A_3})\in\mathbb{F}_2^{2K} \text{~and~} (x^{A_2 A_4},z^{A_2 A_4})\in\mathbb{F}_2^{2K},
\label{eq: a pair of K-bit strings of Bell-measurement in two-register Clifford-gate abbreviation}
\end{equation}
where the $K$-bit string $x$ represents the measurement outcomes of $X\otimes X$, and the $K$-bit string $z$ represents measurement outcomes of $Z\otimes Z$.

From the pair of $2K$-bit outcomes, the correction operation $P_{\mathrm{corr}}\in\mathcal{\Tilde{P}}_{2K}$ for gate teleportation is classically computed during the execution of FTQC, which is in the form of
\begin{equation}
    P_{\mathrm{corr}}\coloneqq U_C^{A_5 A_6}\left( \bigotimes_{k=1}^{K}P_{k}^{A_5}\otimes \bigotimes_{l=1}^{K}P_{l}^{A_6}\right) \left(U_C^{A_5 A_6}\right)^{\dagger},
\label{eq: Pauli operator for correction operation in two-register Clifford-gate abbreviation}
\end{equation}
where $P_k^{A_j}\in\{I,X,Y,Z\}$ with $j\in\{5,6\}$ is a Pauli operator acting on the $k$-th qubit in the register $A_j$.
The correction operation in~\eqref{eq: Pauli operator for correction operation in two-register Clifford-gate abbreviation} can be calculated via multiplication of the symplectic matrix $\gamma(U_C^{A_5A_6})\in\mathbb{F}_2^{4K\times 4K}$ from the right of the row vector of the symplectic representation of $\phi\left(\left(\bigotimes_{k=1}^{K}P_{k}^{A_5}\otimes \bigotimes_{l=1}^{K}P_{l}^{A_6}\right)\right)\in\mathbb{F}_2^{4K}$ as 
\begin{equation}
    \phi(P_{\mathrm{corr}})= \phi\left(\left(\bigotimes_{k=1}^{K}P_{k}^{A_5}\otimes \bigotimes_{l=1}^{K}P_{l}^{A_6}\right)\right)\gamma(U_C^{A_5A_6}).
\label{eq: multiplication for calculating correction operation}
\end{equation}
The resulting $4K$-dimensional row vector in \eqref{eq: multiplication for calculating correction operation} is used as an input bitstring to specify the on-demand Pauli-gate operation in~\eqref{eq: on-demand Pauli gate operations} for the correction operation.
Using $O(K^2)$ parallel processes, this classical computation can be performed within runtime
\begin{equation}
    O(\log{K}).
\label{eq: runtime of calculating the Pauli correction operation}
\end{equation}
The wait operation in~\eqref{Fig: wait-ldpc} is performed during this classical computation.

Furthermore, $O(\log K)$ wait operations follow after the classical computation.
Although no classical computations are performed during these wait operations, the allocated time is used by an on-demand Pauli gate gadget during the execution of the quantum computation to receive the $4K$-bit string in \eqref{eq: multiplication for calculating correction operation} and calculate the physical Pauli operations needed to implement the logical Pauli operation corresponding to $\phi(P_{\mathrm{corr}})$.
The details of this process will be explained in Sec~\ref{sec: Pauli-gate gadgets}.
Finally, the Pauli-gate operation \eqref{eq: Pauli operator for correction operation in two-register Clifford-gate abbreviation} for the correction operation is applied to the registers $A_5,A_6$ completing the abbreviation.

As a result, the depth of the abbreviation is bounded by
\begin{equation}
    O(\log{K}).
\end{equation}

\begin{figure*}[t]
    \centering
    \includegraphics[width=\textwidth]{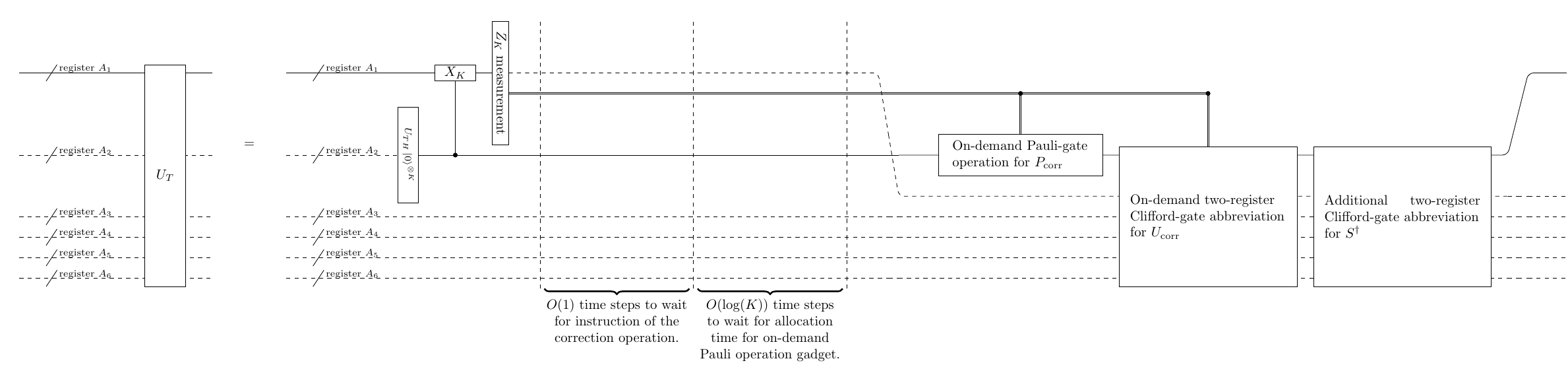}
    \caption{The construction of the $U_{T}$-gate abbreviation in~\eqref{Fig: magic-abbreviation-ldpc.} to perform a $U_{T}$ gate on the register $A_1$, where $U_T$ is a tensor product of $T$, $T^{\dagger}$, and $I$.
    First, a magic-state preparation in~\eqref{Fig: magic-state-preparation-ldpc} prepares a state $U_{TH}\ket{0}^{\otimes K}$ in the auxiliary register $A_2$.
    Next, the CNOT-gate operation in \eqref{Fig: cnot-gate-ldpc} is performed on the registers $A_1$ and $A_2$, followed by the $Z_K$-measurement operation in \eqref{Fig: Z-basis measurements-ldpc.} on the register $A_1$.
    After the $Z_K$ measurement operation, a wait operation with depth $O(1)$ is inserted. This wait time allows for the generation of an input bitstring used to determine both the Pauli correction operation $P_{\mathrm{corr}}$ and the Clifford correction operation $U_{\mathrm{corr}} \in \Tilde{\mathcal{C}}_{2K}/\Tilde{\mathcal{P}}_{2K}$ based on the measurement outcomes, which are necessary for gate teleportation.
    Additionally, wait operations with a depth of $O(\log(K))$ are inserted to provide sufficient time for the classical computation required by the on-demand Pauli-gate gadget to calculate the physical Pauli operations that implement the logical Pauli operation corresponding to $P_{\mathrm{corr}}$ during the quantum computation.
    Then, the Pauli correction operation is applied using an on-demand Pauli-gate operation, and the Clifford correction operation $U_{\mathrm{corr}}$ is applied using an on-demand two-register Clifford-gate abbreviation. At this point, $T$ has been applied to the qubit where $T^{\dagger}$ is intended to be performed. Thus, the additional Clifford gates $S^{\dagger}$ must be applied to each qubit where $T^{\dagger}$ is intended, using a two-register Clifford-gate abbreviation in \eqref{Fig: two-register-clifford-abbreviation-ldpc.} after the Clifford correction operation.
    }
    \label{Fig: t-gate-abbreviation-ldpc.}
\end{figure*}

Next, we explain the on-demand version of the two-register Clifford abbreviation.
A on-demand two-register Clifford abbreviation is used to apply a Clifford unitary $U_C \in \Tilde{\mathcal{C}}_{2K}/\Tilde{\mathcal{P}}_{2K}$ on $2K$ qubits across two registers on demand, where $U_C$ is expressed as a tensor product of $2K$ single-qubit Clifford gates, 
\begin{equation} 
U_C = \bigotimes_{k=1}^{K} (C_k \otimes I_k), \label{eq: on-demand Clifford operation U_C} 
\end{equation} 
where $C_k \in \{I, S\}$ acts on the $k$-th qubit of the first register, and $I_k$ acts on the $k$-th qubit of the second register. 
The operation receives a $K$-bit string as input, which is then fed into the on-demand Clifford-state preparation.

One feature of the on-demand version is the classical computation required to determine the Pauli correction operation for gate teleportation.
Since $U_C$ is expressed as a tensor product of single-qubit Clifford gates, the computations in \eqref{eq: multiplication for calculating correction operation} can be parallelized across qubit. 
With $O(K)$ parallel processes, this classical computation can be executed within runtime $O(1)$.

Since the other procedures are the same as in the two-register Clifford gate abbreviation, the depth is similarly bounded by 
\begin{equation} 
O(\log(K)). 
\end{equation}

\subsubsection{$U_T$-gate abbreviation \label{sec: U_T-gate abbreviation}}

A $U_T$-gate abbreviation is used to apply the $U_{T}$ gate, which is a tensor product of $T$, $T^{\dagger}$ and $I$ acting on the qubits in a register.
The construction of the $U_T$-gate abbreviation is shown in Fig.~\ref{Fig: t-gate-abbreviation-ldpc.}.

The $U_T$-gate abbreviation is based on the gate teleportation protocol~\cite{gottesman2010introduction,knill2005scalable,Knill_2005}. 
At the beginning of the protocol, we have a register $A_1$ containing the qubits on which we want to perform $U_{TH}$, and an auxiliary register $A_2$ in a state
\begin{equation}
\ket{\Psi_{U_{TH}}}= U_{TH}\ket{0}^{\otimes K},
\end{equation}
prepared by the magic-state preparation in~\eqref{Fig: magic-state-preparation-ldpc}. 
After preparing the auxiliary register $A_2$, we perform the CNOT-gate operation in \eqref{Fig: cnot-gate-ldpc}, followed by the $Z_K$-measurement operation in \eqref{Fig: Z-basis measurements-ldpc.}. This $Z_K$-measurement operation outputs a $K$-bit string as
\begin{equation}
z\in\mathbb{F}_2^{K}.
\label{eq: measurement outcomes of Z-basis measurement in the magic-gate abbreviation}
\end{equation}
For $i \in [1, \ldots, K]$, if the measurement outcome $z_i=1$, the correction operation $SX$ is applied to the $i$-th qubit in the register $A_2$. 

To implement these correction operations, we use the on-demand Pauli-gate operation in \eqref{eq: on-demand Pauli gate operations} to apply
\begin{equation}
P_{\mathrm{corr}}=\bigotimes_{k=1}^{K}P_k,
\label{eq: on-demand Pauli-gate operation in U_T gate abbreviation}
\end{equation}
where $P_k \in \{I, X\}$, and the on-demand two-register Clifford-gate operation as explained in Sec.~\ref{sec: Two-register Clifford-gate abbreviation} to apply
\begin{equation}
U_{\mathrm{corr}}=\bigotimes_{k=1}^{K}(C_{k}\otimes I_k),
\label{eq: on-demand Clifford operation in U_T gate abbreviation}
\end{equation}
where $C_k \in \{I, S\}$ acts on the $k$-th qubit in the register $A_2$ and the identity operator $I_k$ acts on the $k$-th qubit in the register $A_1$. 
Using $O(K)$ parallel processors, the runtime of classical computation to generate the inputs for the on-demand Pauli-gate operation and the two-register Clifford-gate operation can be bounded by $O(1)$.
After this classical computation, as with the two-register Clifford gate abbreviation, additional wait operations with depth $O(\log(K))$ are inserted to allocate time for the classical computation required by the on-demand Pauli-gate gadget.

At this point, $T$ has been applied to the qubit where $T^{\dagger}$ is intended to be performed. 
Thus, the additional Clifford gates $S^{\dagger}$ must be applied to each qubit to which $T^{\dagger}$ is intended to be applied.
The additional Clifford gates are applied using the two-register Clifford-gate abbreviation, acting non-trivially on the register $A_2$.

As a result, the depth of the $U_T$-gate abbreviations is bounded by
\begin{equation}
    O(\log(K)).
\end{equation}

\section{Fault-tolerant gadgets on quantum LDPC codes\label{sec:gadgets}}

In this section, we introduce conditions and constructions of fault-tolerant gadgets on quantum LDPC codes.
In Sec.~\ref{sec: Conditions of fault-tolerant gadgets on quantum LDPC codes}, we specify the fault-tolerance conditions of gadgets for quantum LDPC codes.
Then, in Sec.~\ref{sec: Construction of fault-tolerant gadgets and abbreviations(qLDPC)}, we present the constructions of the gadgets for quantum LDPC codes that satisfy the fault-tolerance conditions.

\subsection{Conditions of fault-tolerant gadgets for quantum LDPC codes\label{sec: Conditions of fault-tolerant gadgets on quantum LDPC codes}}

In this section, we define the fault-tolerance conditions of gadgets for quantum LDPC codes.
Before we provide the definition, we introduce a property of a decoding algorithm that is required for constructing our protocol.

The decoding algorithm used for quantum LDPC codes needs to consider cases where the syndrome bits may be erroneous due to noise in measuring the syndrome bits of the physical circuit.
In our setting, we consider a CSS LDPC code obtained from a pair of classical linear codes, $C_X=\ker{H_X}$ and $C_Z=\ker{H_Z}$ satisfying $C_Z^{\perp} \subseteq C_X$, where $H_X\in\mathbb{F}_2^{M_Z\times N}$ and $H_Z\in \mathbb{F}_2^{M_Z\times N}$ are parity-check matrices. 
Let $e=\phi(E)\in\mathbb{F}_2^{2N}$ be a symplectic representation of a Pauli error $E\in\mathcal{\Tilde{P}}_N$ on physical qubits, and $\Delta=(\Delta_X, \Delta_Z)\in\mathbb{F}_2^{M}$ ($M=M_Z+M_X$) be a bitstring representing the locations of errors on syndrome bits.
The ideal syndrome bits 
\begin{equation}
    \sigma =(\sigma_X,\sigma_Z)\in\mathbb{F}_2^{M}
\label{eq: ideal syndrome bits}
\end{equation}
for the error $e\in\mathbb{F}_2^{2N}$ is given by
\begin{equation}
        \sigma_X=H_Ze_X \quad\mathrm{and} \quad \sigma_Z=H_Xe_Z.
\end{equation}
However, the ideal syndrome bits $\sigma\in\mathbb{F}_2^{M}$ can be corrupted by syndrome errors $\Delta\in\mathbb{F}_2^{M}$, resulting in erroneous syndrome bits 
\begin{equation}
    \Tilde{\sigma}=(\Tilde{\sigma}_X, \Tilde{\sigma}_Z)\in\mathbb{F}_2^{M}
\label{eq: noisy syndrome}
\end{equation}
which is given by
\begin{equation}
    \Tilde{\sigma}_X=\sigma_X \oplus \Delta_X\quad\mathrm{and}\quad\Tilde{\sigma}_Z=\sigma_Z\oplus \Delta_Z.
\end{equation}
Using the erroneous syndrome $\Tilde{\sigma}$, the decoding algorithm aims to calculate the recovery operation $R\in \mathcal{P}_N$ with the symplectic representation $r=\phi(R)$.

In the analysis of the decoding algorithm, we use the local stochastic Pauli error model for the pair $(e,\Delta)$ representing the locations of errors on physical qubits and syndrome bits, which is defined as follows.
\begin{definition}[Local stochastic Pauli error model on physical qubits and syndrome bits]\label{def: Local stochastic Pauli error model on physical qubits and syndrome bits}
    Let $V$ be the set of physical qubits and $W_{X(Z)}$ be the set of $X(Z)$-type stabilizer generators.
    Let $e=(e_X, e_Z)\in\mathbb{F}_{2}^{2N}$ be a bitstring representing the locations of $X$- and $Z$-type Pauli errors, $\Delta=(\Delta_X,\Delta_Z)\in\mathbb{F}_2^{M}$ be a bitstring representing the location of errors of syndrome bits for correcting $X$ and $Z$ errors.
Let $\mathrm{supp}(e_{X(Z)}) \subseteq V$ be the set of indices, i.e., locations, where $e_{X(Z)}$ has a value of $1$, and $\mathrm{supp}(\Delta_{X(Z)}) \subseteq W_{Z(X)}$ be the set of indices where $\Delta_{X(Z)}$ has a value of $1$.
    We say that errors $(e, \Delta)$ follow the local stochastic Pauli error model with error parameters $(p_{\mathrm{phys}}, p_{\mathrm{synd}})$ if for all subsets $S\subseteq V$ of physical qubits, $T_X\subseteq W_X$ of syndrome bits for $X$, and $T_Z\subseteq W_Z$ of syndrome bits for $Z$, the following relations are satisfied.
    \begin{align}
        \mathbb{P}\left[\mathrm{supp}(e_X)\supseteq S~\mathrm{and}~\mathrm{supp}(\Delta_X)\supseteq T_X \right]&\leq p_{\mathrm{phys}}^{|S|}p_{\mathrm{synd}}^{|T_X|},\\
        \mathbb{P}\left[\mathrm{supp}(e_Z)\supseteq S~\mathrm{and}~\mathrm{supp}(\Delta_Z)\supseteq T_Z \right]&\leq p_{\mathrm{phys}}^{|S|}p_{\mathrm{synd}}^{|T_Z|}.
    \label{eq: local stochastic error model on physical qubits and syndrome bits}
    \end{align}
    In particular, when the errors on the syndrome bits are not considered, we say that errors follow a local stochastic Pauli error model with parameter $p_{\mathrm{phys}}$ if for all subsets $S\subseteq V$ of physical qubits, the errors satisfy
    \begin{align}
        \mathbb{P}\left[\mathrm{supp}(e_X)\supseteq S \right]\leq p_{\mathrm{phys}}^{|S|},\\
        \mathbb{P}\left[\mathrm{supp}(e_Z)\supseteq S \right]\leq p_{\mathrm{phys}}^{|S|}.
    \label{eq: local stochastic error model on physical qubits}
    \end{align}
\end{definition}

Quantum expander codes~\cite{Fawzi_2018,Leverrier_2015} and quantum Tanner codes~\cite{Leverrier2022Tanner}, which are CSS LDPC codes, each have their own decoding algorithms~\cite{Fawzi_2018,grospellier:tel-03364419,Gu2024SingleShot} such that when $\Delta\neq 0$, they can deduce a recovery operation to suppress the residual error $e^{\mathrm{re}}\in\mathbb{F}_2^{2N}$ on physical qubits,
\begin{equation}
    e^{\mathrm{re}}\coloneqq e\oplus r.
\end{equation}
The decoding algorithms for both codes may tolerate measurement errors in syndrome bits due to the single-shot property, but as discussed below, the proof of the threshold theorem requires further properties on bounds of error suppression. 
As we will clarify below, the decoding algorithm for the quantum expander codes~\cite{Leverrier_2015,Fawzi_2018_eff,Fawzi_2018,grospellier:tel-03364419} has all the required properties while that for the quantum Tanner codes~\cite{leverrier2022decoding,Leverrier2023Efficient,Gu2024SingleShot} currently lacks the proof of having some of the properties.

A decoding algorithm with the single-shot property is crucial for feasibility of correcting errors using erroneous syndrome bits $\Tilde{\sigma}$ that are extracted by a noisy physical circuit.
These erroneous syndrome bits are obtained from a \textit{single} round of syndrome measurements for each stabilizer generator~\cite{grospellier:tel-03364419, Fawzi_2018, Gu2024SingleShot,Campbell_2019,Bombin2015Single, Kubica_2022,Brown2016}. 
Using these noisy syndrome bits, existing single-shot decoding algorithms for quantum expander codes and quantum Tanner codes~\cite{grospellier:tel-03364419, Fawzi_2018, Gu2024SingleShot} run in an iterative way; i.e., they repeat $T$ internal loops to output a final estimate of a recovery operation $r\coloneqq r^{(T)}$. 
Each of the $T$ loops can be executed with runtime $O(1)$, using $O(N)$ parallel processes. 
For each internal loop $t\in [1, \ldots, T]$, the algorithm computes a temporally deduced recovery operation $r^{(t)} \in \mathbb{F}_2^{2N}$. 
The decoding algorithm requires that for all $t$, the support of the residual error, derived by using $r^{(t)} \in \mathbb{F}_2^{2N}$, should be less than or equal to the code distance $D$.
Formally, the definition of the single-shot decoding algorithm for non-vanishing-rate quantum LDPC codes is as follows.

\begin{definition}[Single-shot decoding algorithm (Definition~3.3 of~\cite{Gu2024SingleShot})]
\label{def: single-shot decoding algorithm with a fixed weight}
    Let $\{\mathcal{Q}_i\}_i$ be a family of CSS LDPC codes where $\mathcal{Q}_i$ is an $[[N_i,K_i,D_i]]$ code with $N_i\rightarrow\infty$ for $i\rightarrow\infty$.
    Let $e=(e_X, e_Z)\in\mathbb{F}_{2}^{2N}$ be a bitstring representing the locations of Pauli errors, $\Delta=(\Delta_X,\Delta_Z)\in\mathbb{F}_2^{M}$ be a bitstring representing the locations of errors on syndrome bits, and $\Tilde{\sigma}=(\Tilde{\sigma}_X, \Tilde{\sigma}_Z)\in\mathbb{F}_2^{M}$ be the erroneous syndrome bits.
    For $P=\{X,Z\}$, a decoding algorithm returns the recovery operation $r^{(t)}_{P}\in\mathbb{F}_2^{2N}$ at each internal loop $t\in\{1,\ldots, T\}$.
    After the $T$ iterations, the final output of the recovery operation is denoted by $r_P\coloneqq r^{(T)}_P$.
    We write
    \begin{equation}
        u_P\coloneqq e_P\oplus r_P.
    \end{equation}
    A decoding algorithm for $\{\mathcal{Q}_i\}_i$ is said to be single-shot if the following property holds for sufficiently small constants $a$, $b$, and $c$:
    if we have
    \begin{equation}
        |u_P|\leq a|e_P|_\mathrm{R}+b|\Delta_P|\leq c D_i,
    \end{equation}
    then given the erroneous syndrome bits $\tilde{\sigma}_P\in\mathbb{F}_2^{M_P}$, the algorithm after the $T$ iterations can find a recovering operation $r_P\in\mathbb{F}_2^{N}$ satisfying
    \begin{equation}
        |e_P+r_P|_\mathrm{R}\leq\alpha|e_P|+\beta |\Delta_P|,
    \end{equation}
    with 
    \begin{equation}
        \alpha=\exp({-\Omega(T)}),\quad\beta=O(1)\quad\text{as $T\to\infty$}.
    \end{equation}
    Moreover, each internal loop of the decoding algorithm can be performed within runtime $O(1)$ by using $O(N)$ parallel processes.
    Here $|\cdot|_\mathrm{R}$ denotes the stabilizer-reduced weight defined to be the minimum weight of a stabilizer-equivalent error to $e_P\in\mathbb{F}_2^{N}$, i.e.,
    \begin{equation}
        |e_P|_\mathrm{R}\coloneqq\min_{e'\in C^{\perp}}|e_P+e'|.
    \end{equation}
\label{def: single-shot decoding algorithm}
\end{definition}

In the analysis of a decoding algorithm, we consider the local stochastic Pauli error model on the physical qubits and the syndrome bits.
Under this model, the weight of errors typically scales linearly in the number of physical qubits $N$ contained in a code $\mathcal{Q}$.
It appears that a code with a linear distance $D=\Theta(N)$ would be useful for correcting errors of the linear weight. 
However, even if one uses quantum LDPC codes with sublinear distances, there are cases under the local stochastic error model where errors can still be corrected with a high probability when the error parameter is sufficiently small~\cite{Kovalev_2013,Gottesman2014Constant,grospellier:tel-03364419, Fawzi_2018}. 
This is feasible because, at sufficiently low error rates, typical linear-weight errors occurring in quantum LDPC codes tend to form small, separate clusters of errors that are independently correctable by using the decoding algorithms.

In the following, we define two decoding algorithms that quantum LDPC codes used in our protocol should have: the single-shot decoding algorithm with thresholds for reducing residual errors in Def.~\ref{def: single-shot decoding algorithm with thresholds to suppress residual error} and the decoding algorithm with threshold for correcting errors in Def.~\ref{def: Logarithmic-time decoding algorithm with threshold}.
These algorithms are designed to correct local stochastic Pauli errors on physical qubits and syndrome bits with high probability if an error parameter is below a certain threshold value.
The family of quantum expander codes with $D=\Theta(\sqrt{N})$ is the only known family of non-vanishing-rate quantum LDPC codes satisfying the requirements in Defs.~\ref{def: single-shot decoding algorithm with thresholds to suppress residual error} and \ref{def: Logarithmic-time decoding algorithm with threshold}. 
Specifically, the small-set-flip decoding algorithm for the quantum expander codes can work as both the single-shot decoding algorithm for reducing errors with thresholds in Def.~\ref{def: single-shot decoding algorithm with thresholds to suppress residual error} and the decoding algorithm for correcting errors in Def.~\ref{def: Logarithmic-time decoding algorithm with threshold}~\cite{grospellier:tel-03364419, Fawzi_2018} (in particular, see Theorems 5.30 and 5.36 in~\cite{grospellier:tel-03364419})
More recently, Ref.~\cite{Gu2024SingleShot} showed that the family of quantum Tanner codes with $D=\Theta(N)$ also has a single-shot decoding algorithm in the sense of Def.~\ref{def: single-shot decoding algorithm}, but there is no proof that bounds the error parameter of a residual error under the local stochastic error model as in~\eqref{eq: residual error rate after error correction}.
Building on this direction, Ref.~\cite{christandl2024faulttolerantquantuminputoutput} extends the analysis to provide an error bound for error correction with quantum Tanner codes within the EC gadget.
However, Ref.~\cite{christandl2024faulttolerantquantuminputoutput} does not provide a complete error analysis for all gadgets but focuses on analyzing specific parts of the fault-tolerant protocol.
In contrast, our contribution lies in performing a comprehensive analysis for all gadgets, thereby completing the proof of the threshold theorem.

Both decoding algorithms are performed iteratively to reduce or correct errors.
It is important to note that the runtime of each loop does not grow as $N\rightarrow\infty$ when parallel processes are used, and thus, the runtime of the classical computation for decoding depends only on the number of iterations within each algorithm.
These decoding algorithms work in different situations and play different roles in our fault-tolerant protocol.
Specifically, the single-shot decoding algorithm for reducing errors in Def.~\ref{def: single-shot decoding algorithm with thresholds to suppress residual error} is capable of keeping the residual error small even if the syndrome bits are noisy; thus, we will utilize this algorithm in the EC gadget to keep the residual error small in the code block.
On the other hand, the decoding algorithm for correcting errors in Def.~\ref{def: Logarithmic-time decoding algorithm with threshold} is capable of recovering the logical state, but this requires that the syndrome bit errors should be absent; thus, we will use this algorithm in the $Z_K$-measurement gadget and the Bell-measurement gadget to obtain the measurement outcomes of the logical qubits in $\mathcal{Q}$, as will be explained later in Secs.~\ref{sec: $Z$-basis measurements} and~\ref{sec: Bell-measurement gadget}, respectively. 
Since measurement errors are absent in this case, the single-shot property is not required for this algorithm.
To summarize, we define the single-shot decoding algorithm for reducing errors and the (not necessarily single-shot) decoding algorithm for correcting errors as follows.

\begin{definition}[Single-shot decoding algorithm with threshold for reducing errors\label{def: single-shot decoding algorithm with thresholds to suppress residual error}]
Let $\{\mathcal{Q}_i\}_i$ be a family of CSS LDPC codes, where $\mathcal{Q}_i$ is an $[[N_i, K_i, D_i]]$ code with $N_i \rightarrow \infty$ as $i \rightarrow \infty$.
Let $e = (e_X, e_Z) \in \mathbb{F}_2^{2N}$ be a random variable representing the locations of the physical qubits with errors, and $\Delta = (\Delta_X, \Delta_Z) \in \mathbb{F}_2^{M}$ be a random variable representing the locations of the syndrome bits with errors, which are local stochastic with parameters $p_{\mathrm{phys}}\in (0,1]$ and $p_{\mathrm{synd}}\in (0,1]$, i.e.,
\begin{align}
&\forall S \subseteq V, ~\forall T_X \subseteq W_X, \text{~and~} \forall T_Z \subseteq W_Z, \nonumber\\
& \mathbb{P}\left[\mathrm{supp}(e_X)\supseteq S~\mathrm{and}~\mathrm{supp}(\Delta_X)\supseteq T_X \right]\leq p_{\mathrm{phys}}^{|S|}p_{\mathrm{synd}}^{|T_X|},\\
        &\mathbb{P}\left[\mathrm{supp}(e_Z)\supseteq S~\mathrm{and}~\mathrm{supp}(\Delta_Z)\supseteq T_Z \right]\leq p_{\mathrm{phys}}^{|S|}p_{\mathrm{synd}}^{|T_Z|}.
\label{eq: local stochastic error model on physical qubits and syndrome bits}
\end{align}
We say that a decoding algorithm for $\{\mathcal{Q}_i\}_i$ is a single-shot decoding algorithm with threshold for reducing errors if there exists a threshold $p_1^\thre > 0$ satisfying the following properties:
if $p_\mathrm{phys} < p_1^\thre$ and $p_{\mathrm{synd}} < p_1^\thre$, then, for each $P\in\{X,Z\}$, the decoding algorithm returns a bitstring $r_P\in\mathbb{F}_2^{N}$ after $T$ internal loops such that a set of the locations of residual errors
\begin{align}
    e'_P\coloneqq e_P\oplus r_P
\end{align}
satisfies
\begin{align}
&\forall S' \subseteq V,\nonumber\\
&\mathbb{P}[m(e_P, \Delta_P) = 0 \text{~and~} \mathrm{supp}(e'_P) \supseteq S'] \leq \qty(p_1^\thre)^{c(T)|S'|},
\label{eq: residual error rate after error correction}
\end{align}
where $c(T)>1$ (called an error suppression parameter) is a monotonically increasing.

Here, $m \colon \mathbb{F}_2^N\times\mathbb{F}_2^M\rightarrow \{0, 1\}$ is a flag indicating if the decoding succeeds for $e_P$ and $\Delta_P$, and the failure probability of decoding is bounded by
\begin{equation}
\mathbb{P}[m(e_P, \Delta_P) = 1] \leq \delta_1(D_i),
\end{equation}
with a parameter
\begin{equation}
\delta_1(D_i) = \exp(-\Omega(D_i)).
\end{equation}
Each internal loop should be executable with runtime $O(1)$, using $O(N_i)$ parallel processes.
If $m(e_P, \Delta_P) = 1$, we say that the single-shot decoding algorithm fails.
We refer to $\delta_1$ as the failure probability of the single-shot decoding algorithm.
\end{definition}

\begin{definition}[Decoding algorithm with threshold for correcting errors\label{def: Logarithmic-time decoding algorithm with threshold}]
    Let $\{\mathcal{Q}_i\}_i$ be a family of CSS LDPC codes, where $\mathcal{Q}_i$ is an $[[N_i,K_i,D_i]]$ code with $N_i\rightarrow\infty$ as $i\rightarrow\infty$.
    Let $e=(e_X, e_Z)\in\mathbb{F}_2^{2N}$ be a bitstring representing the locations of the physical qubit, which are local stochastic with parameters $p_{\mathrm{phys}}\in (0,1]$, i.e.,
\begin{align}
\forall S \subseteq V,\quad
&\mathbb{P}\left[\mathrm{supp}(e_X)\supseteq S \right]\leq p_{\mathrm{phys}}^{|S|},\\
&\mathbb{P}\left[\mathrm{supp}(e_Z)\supseteq S\right]\leq p_{\mathrm{phys}}^{|S|}.
\end{align}
We say a decoding algorithm is a decoding algorithm with threshold for correcting errors if there exists a threshold $p_2^\thre>0$ satisfying the following properties.
If $p_\mathrm{phys}<p_2^\thre$, then for each $P\in\{X,Z\}$, the decoding algorithm returns a bitstring $r_P\in\mathbb{F}_2^{N}$ 
after $T=O(\log{N_i})$ internal loops such that
\begin{equation}
e_P\oplus r_P\in C_P^{\perp}
\end{equation}
with probability at least
\begin{equation}
        1-\delta_2(D_i),
    \end{equation}
    where 
    \begin{equation}
        \delta_2(D_i)=\exp(-\Omega(D_i)).
    \end{equation}
    Each internal loop should be able to be executed with runtime $O(1)$, using $O(N_i)$ parallel processes. 
        We say $\delta_2$ is the failure probability of the decoding algorithm for correcting errors.
    \end{definition}

Furthermore, the family of CSS LDPC codes with decoding algorithms in Defs.~\ref{def: single-shot decoding algorithm with thresholds to suppress residual error} and \ref{def: Logarithmic-time decoding algorithm with threshold} should have a constrained difference between adjacent code block size $N_i$ and $N_i$, ensuring that $N_i$ does not increase too rapidly between successive indices $i$.
This condition, which is also introduced in Ref.~\cite{Gottesman2014Constant}, can be expressed as
\begin{align}
    0<N_i-N_{i-1}\leq N_{i-1}^\beta,\quad\text{for a constant $\beta>0$}.
\end{align}
This contributes to selecting an appropriate code block size for FTQC.
Note that the above inequality holds for any $\beta>1/2$ in quantum expander codes~\cite{Leverrier2015QuantumExpander}.  

To summarize, we make the following assumption about families of non-vanishing-rate quantum LDPC codes that can be used in our protocol.
\begin{assumption}
    A family of non-vanishing-rate CSS LDPC codes $\{\mathcal{Q}_i\}_i$ used in our protocol has both the single-shot decoding algorithm for reducing errors with thresholds in Def.~\ref{def: single-shot decoding algorithm with thresholds to suppress residual error} and the decoding algorithm for correcting errors with threshold in Def.~\ref{def: Logarithmic-time decoding algorithm with threshold},  where $\mathcal{Q}_i$ is an $[[N_i, K_i, D_i]]$ code with monotonically increasing $N_i$ satisfying $N_i\rightarrow\infty$, $K_i=\Theta(N_i)$, and $D_i=\Theta(N_i^{\gamma})$ $(0<\gamma\leq 1)$ as $i\rightarrow\infty$. In addition, $N_i$ does not increase too rapidly between successive indices $i$, i.e., $0<N_i-N_{i-1}\leq N_{i-1}^\beta$ for a constant $\beta> 0$.
\label{assump: non-vanishing-rate quantum LDPC with efficient decoding algorithm}
\end{assumption}

In order for the fault-tolerant protocol with a family $\{\mathcal{Q}_i\}_i$ of quantum LDPC codes to exhibit fault tolerance, gadgets must be designed so that errors caused by faults in a gadget do not propagate excessively even when the code block size $N_i$ of $\mathcal{Q}_i$ increases as $i\rightarrow\infty$.
In addition, a gadget does not contain any faults, its logical action should be the intended one.
In the following, we formally define the fault-tolerance conditions that each type of gadget in our protocol must satisfy.

Since our protocol aims for small space and time overheads, the width and the depth of each gadget are also crutial.
Thus, we include size requirements in the conditions.
Instead of specifying requirements for a set of gadgets, we state them for each gadget individually, using universal constants ($d_0$, $c_W$, and $c_D$) that apply to all gadgets.
To prove the scaling of overheads, it suffices to ensure the existence of such constants, as formally stated in the following lemma.

\begin{lemma}
\label{lemma: fault-tolerance condition}
There exist positive constants $d_0$, $c_W$, and $c_D$ such that every gadget satisfies either Def.~\ref{def: fault-tolerance condition of gate, measurement for qLDPC codes}, Def.~\ref{def: fault-tolerance conditions of the state preparation gadgets for quantum LDPC codes}, or Def.~\ref{def: fault-tolerance conditions of the EC gadgets for quantum LDPC codes}.
\end{lemma}

The definition of fault tolerance of gate gadgets for the CNOT-gate operation in~\eqref{Fig: cnot-gate-ldpc} and on-demand Pauli-gate gadget in~\eqref{Fig: pauli-gate-ldpc}, measurement gadgets for the $Z_K$-measurement operation in~\eqref{Fig: Z-basis measurements-ldpc.} and the Bell-measurement operation in~\eqref{Fig: bell-measurement-ldpc} is stated as follows.

\begin{definition}[Fault-tolerance conditions of gate and measurement gadgets for quantum LDPC codes]
\label{def: fault-tolerance condition of gate, measurement for qLDPC codes}
    Let $C$ be a physical circuit of a gate or measurement gadget for an $[[N,K,D]]$ code $\mathcal{Q}$.
    For given constants $d_0$ and $c_W$,  we say that the gadget $C$ is fault-tolerant if the following conditions hold.
    \begin{enumerate}
            \item The width of $C$ is no greater than $c_W N$.
        \item The depth of $C$ is no greater than $d_0$.

        \item Each physical operation in $C$ acts on at most two qubits.  

        \item If $C$ has no faults and the input state to $C$ is a codeword of $\mathcal{Q}$, the action of $C$ is exactly the intended logical operation or logical measurement.
       
    \end{enumerate}
\end{definition}

In contrast to the gate or measurement gadgets, 
the fault-tolerance property of the state-preparation gadgets relies on  concatenated Steane codes. 
In Appendix~\ref{appendix: fault-tolerant protocol for open quantum circuits}, we provide an explicit construction for converting an original open circuit $C$ into a fault-tolerant one $\tilde{C}(C, L)$ using level-$L$ concatenated codes, together with Theorem~\ref{Theorem: level-reduction for the circuit that outputs a quantum state} describing its properties.
The definition of fault tolerance of state-preparation gadgets for the $\ket{0}^{\otimes K}$-state preparation in~\eqref{Fig: 0-state-preparation-ldpc}, the Clifford-state preparations in~\eqref{Fig: clifford-state-preparation-ldpc}, and the magic-state preparations in~\eqref{Fig: magic-state-preparation-ldpc} is stated as follows.
\begin{definition}[Fault-tolerance conditions of the state preparation gadgets for quantum LDPC codes]
\label{def: fault-tolerance conditions of the state preparation gadgets for quantum LDPC codes}
    Let $C$ be a physical circuit of a state-preparation gadget for an $[[N,K,D]]$ code $\mathcal{Q}$.
    For given constants $c_W$ and $c_D$, we say that the gadget $C=\tilde{C}(C^\org, L)$ for an integer $L$ and for an original open circuit $C^\org$ is fault-tolerant if $C^\org$ satisfies the following conditions.
       \begin{enumerate}

        \item The width of $C^\org$ is no greater than $c_W N$.

        \item The depth of $C^\org$ is no greater than $c_D N$.

        \item If $C^\org$ suffers from no faults, $C^\org$ exactly prepares the intended logical state.
       
    \end{enumerate}
   \end{definition}

The fault-tolerance of the EC gadgets depends on both the structure of syndrome measurement circuit and the decoding algorithm.
The definition of fault tolerance of the EC gadgets is stated as follows.
\begin{definition}[Fault-tolerance conditions of the EC gadgets for quantum LDPC codes]
\label{def: fault-tolerance conditions of the EC gadgets for quantum LDPC codes}
    Let $C$ be a physical circuit of an EC gadget for an $[[N,K,D]]$ code $\mathcal{Q}$.
    For given constants $d_0$ and $c_W$, we say that the gadget $C$ is fault-tolerant if  
    $C$ is divided into the first part $C_1$ and the second part $C_2$  such that 
    the following conditions hold.
       \begin{enumerate}
       
        \item All the syndrome bits have been obtained at the end of $C_1$.
       
        \item The depth of $C_1$ is no greater than $d_0$.

        \item Each physical operation in $C_1$ acts on at most two qubits.  
        
        \item In $C_2$, a recovery operation is applied according to the single-shot decoding algorithm in Def.~\ref{def: single-shot decoding algorithm with thresholds to suppress residual error} with $T$ internal loops, using the syndrome obtained in $C_1$.
        
        \item The depth of $C_2$ may grow as $T$ increases but is independent of $N$.
                
        \item Each physical operation in $C_2$ acts only on a single qubit.
                
        \item The width of $C$ is no greater than $c_W N$. 

    \end{enumerate}
   \end{definition}

In the next subsection, we explicitly present all gadgets used in our fault-tolerant protocol and verify that they satisfy the properties defined in Defs. \ref{def: fault-tolerance condition of gate, measurement for qLDPC codes}--\ref{def: fault-tolerance conditions of the EC gadgets for quantum LDPC codes}.
We analyze the scaling of 
the depth and the width for each type of gadget, determining the specific coefficients associated with each type. 
Since the number of types is finite, Lemma~\ref{lemma: fault-tolerance condition} follows immediately by choosing the maximum coefficients across all types. 
The gadget properties established in Lemma~\ref{lemma: fault-tolerance condition} and Defs.~\ref{def: fault-tolerance condition of gate, measurement for qLDPC codes}--\ref{def: fault-tolerance conditions of the EC gadgets for quantum LDPC codes} will be used to prove the threshold theorem and to estimate the overheads of our protocol in Sec.~\ref{sec: threshold theorem based on quantum LDPC codes}.

\subsection{Construction of fault-tolerant gadgets \label{sec: Construction of fault-tolerant gadgets and abbreviations(qLDPC)}}
In this section, we present the construction of the gadgets and check their fault-tolerance conditions in Defs.~\ref{def: fault-tolerance condition of gate, measurement for qLDPC codes},~\ref{def: fault-tolerance conditions of the state preparation gadgets for quantum LDPC codes}, and~\ref{def: fault-tolerance conditions of the EC gadgets for quantum LDPC codes}.

\begin{figure*}[t]
    \centering
    \includegraphics[width=\textwidth]{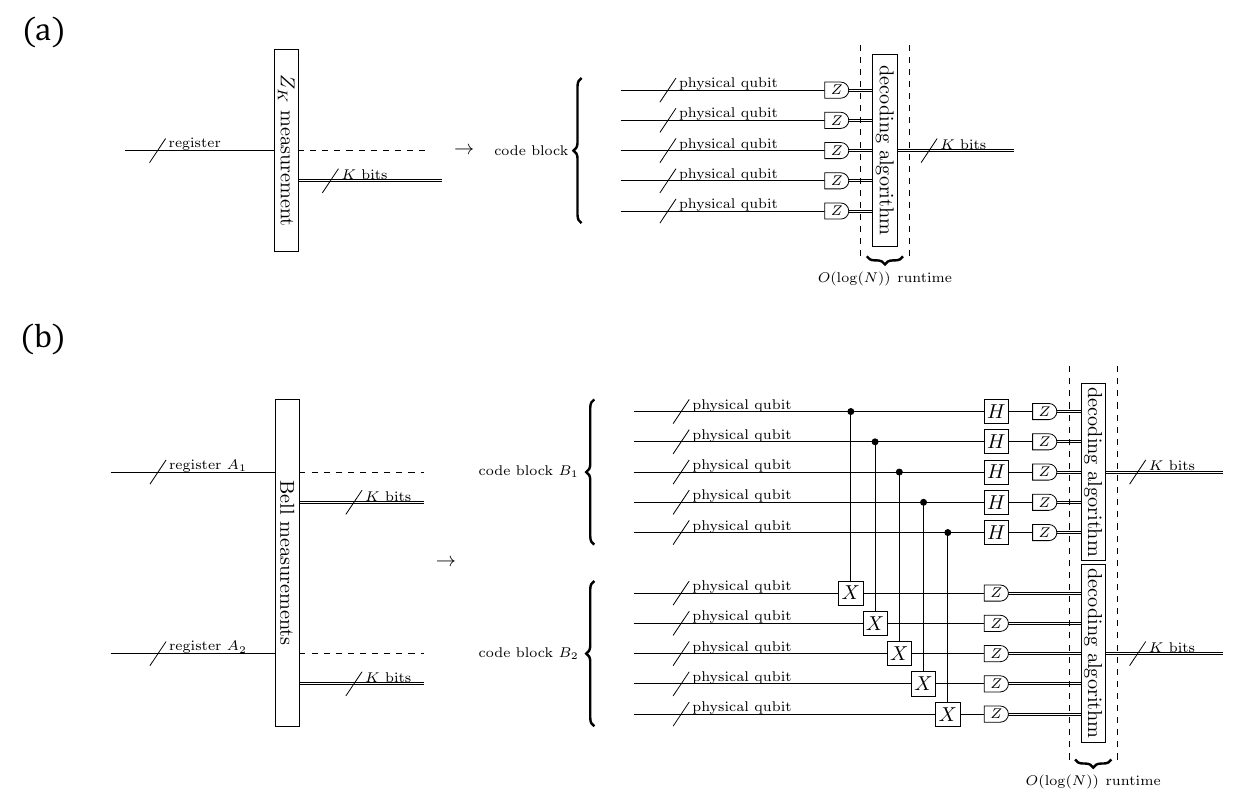}
    \caption{(a) The physical circuit of the gadget for the $Z_K$-measurement operation~\eqref{Fig: Z-basis measurements-ldpc.}. The gadget is implemented by $Z$-basis measurements on $N$ qubits.
    Based on the measurement outcomes, we perform classical computation for recovering measurement outcomes of $K$-bit string, where $i$-th bit corresponds to the measurement outcome of the logical $Z_i$ operator of $\mathcal{Q}$ in \eqref{eq: code parameters of quantum LDPC codes}.
    The runtime of the classical computation is bounded by $O(\log{N})$.
    (b) The physical circuit of the gadget for the Bell-measurement operation is shown in~\eqref{Fig: bell-measurement-ldpc}.
    The gadget is implemented by transversal CNOT gates, transversal $H$ gates on the code block $B_1$, followed by transversal $Z$-basis measurements on $B_1$ and $B_2$.
    Based on the measurement outcomes, we perform classical computation to recover the measurement outcomes of the logical $X \otimes X$ and $Z \otimes Z$ operators, where the runtime is bounded by $O(\log{N})$.}
    \label{fig: destructive-measurement-ldpc}
\end{figure*}

\subsubsection{$Z_K$-measurement gadget \label{sec: $Z$-basis measurements}}
The construction of a $Z_K$-measurement gadget is shown in Fig.~\ref{fig: destructive-measurement-ldpc}~(a).
The gadget is implemented by transversal $Z$-basis measurements, followed by a classical computation for performing the decoding algorithm in Def.~\ref{def: Logarithmic-time decoding algorithm with threshold}.
The transversal measurements give $N$-bit measurement outcomes 
\begin{equation}
    z\in \mathbb{F}_2^{N},
\label{eq: physical measurement outcome of Z-basis measurement operation}
\end{equation}
where $z_i$ corresponds to the measurement outcome of the $Z$-basis measurement on the $i$-th physical qubit and experiences errors.

From $z\in \mathbb{F}_2^{N}$, we need to estimate the $K$-bit logical measurement outcomes 
\begin{equation}
    \bar{z}\in\mathbb{F}_2^{K},
\label{eq: logical measurement outcome of Z-basis measurement operation}
\end{equation}
where $\bar{z}_i$ corresponds to the measurement outcome of the logical $\Bar{Z}_i$ operator on the $i$-th logical qubit.
This can be done in the following steps:
\begin{enumerate}
    \item From $z\in \mathbb{F}_2^{N}$, we calculate the syndrome bits 
    \begin{equation}
        \sigma_X\in\mathbb{F}_2^{M_Z}
    \label{eq: syndrome bits of Z-basis measurement operation}
    \end{equation}
    of the $Z$-type generators by using the relation given by,
    \begin{equation}
        (\sigma_X)_i\coloneqq\bigoplus_{j\in\mathrm{supp}(g^{Z}_i)}z_j,
    \end{equation}
    where $g^{Z}_i$ is the $i$-th $Z$-type stabilizer generator, and $\mathrm{supp}(g^{Z}_i)$ represents the set of indices of the physical qubits on which $g^{Z}_i$ acts nontrivially.
    Here, the syndrome bits $\sigma_X$ in \eqref{eq: syndrome bits of Z-basis measurement operation} are calculated classically, and thus there are no syndrome bit errors.
    Using $M_Z=O(N)$ parallel processes, we can compute the syndrome bits $\sigma_X$ in \eqref{eq: syndrome bits of Z-basis measurement operation} within runtime $O(1)$ due to $|g_i^{Z}|=O(1)$ as $N\rightarrow\infty$, since $\mathcal{Q}$ in \eqref{eq: code parameters of quantum LDPC codes} is a CSS LDPC code.
    \item We perform the decoding algorithm in Def.~\ref{def: Logarithmic-time decoding algorithm with threshold} that takes as input the syndrome bits $\sigma_X\in \mathbb{F}_2^{M_Z}$ and outputs a recovery operation $r_X\in \mathbb{F}_2^{N}$ for correcting $X$-type errors. 
    Then, the corrected bitstring $\tilde{z}\in\mathbb{F}_2^{N}$ is obtained as follows: for all $i\in[1,\ldots,N]$,
    \begin{equation}
        \tilde{z}_i\coloneqq z_i\oplus (r_X)_i.
    \label{eq: calculating correct bits}
    \end{equation}
    Since there are no syndrome bit errors, the decoding algorithm can correct the data error.
    A decoding algorithm for correcting errors with a threshold in Def.~\ref{def: Logarithmic-time decoding algorithm with threshold} can be performed within runtime $O(\log{N})$, and calculating~\eqref{eq: calculating correct bits} for all $i$ can also be performed within runtime $O(1)$ using $O(N)$ parallel processes.
    \item  We calculate the measurement outcomes $\bar{z}\in\mathbb{F}_2^{K}$ of the logical operators as follows: for all $i\in[1,\ldots,K]$,
    \begin{equation}
        \bar{z}_i\coloneqq \bigoplus_{j\in\mathrm{supp}(\Bar{Z}_i)}\Tilde{z}_j,
    \label{eq: calculating the logical measurement outcomes}
    \end{equation}
    where $\Bar{Z}_i$ is the logical-$Z$ operator acting on the $i$-th logical qubit, and $\mathrm{supp}(\Bar{Z}_i)$ denotes the set of indices of the physical qubits on which $\Bar{Z}_i$ acts nontrivially.
    Using $DK=O(N^2)$ parallel processes, we can compute~\eqref{eq: calculating the logical measurement outcomes} for all $i$ within runtime $O(\log{N})$ since $|\Bar{Z}_i|=\Theta(D)=\Theta(N^{\gamma})$ for $0<\gamma\leq 1$ as $N\rightarrow\infty$.
\end{enumerate}
Therefore, using $O(N^2)$ parallel processes, the classical part of the gadget can be performed within runtime
\begin{equation}
    O(\log{N}).
\label{eq: classical part of the Z basis measurement}
\end{equation}

As a result, the width of the quantum part of the gadget is $N$ and the depth of the quantum part is $1$.
Thus, the fault-tolerance conditions in Def.~\ref{def: fault-tolerance condition of gate, measurement for qLDPC codes} are satisfied.
The runtime of the classical part of the gadget is bounded by
\begin{equation}
    O(\log{N}).
\end{equation}

\subsubsection{Bell-measurement gadget \label{sec: Bell-measurement gadget}}
The construction of a Bell-measurement gadget is shown in Fig.~\ref{fig: destructive-measurement-ldpc}~(b).
The gadget is implemented by transversal CNOT gates between the controlled block $B_1$ and the target block $B_2$, transversal $H$ gates on $B_1$, followed by transversal $Z$-basis measurements on $B_1$ and $B_2$.
The transversal measurements on $B_1$ and $B_2$ yield a pair of $N$-bit strings of measurement outcomes, respectively, as
\begin{equation}
(x, z) \in \mathbb{F}_2^{2N},
\label{eq: physical measurement outcome of XX and ZZ measurement operation}
\end{equation}
where $x\in \mathbb{F}_2^{N}$ represents the outcomes from $B_1$ and $z\in \mathbb{F}_2^{N}$ represents those from $B_2$.
Here, $x_i$ and $z_i$ correspond to the measurement outcome of the $X_{i}\otimes X_{i}$ and $Z_{i}\otimes Z_{i}$ operators, respectively, on $N$ pairs of physical qubits.

From $(x,z)$ in~\eqref{eq: physical measurement outcome of XX and ZZ measurement operation}, we deduce a pair of $K$-bit strings of logical measurement outcomes as
\begin{equation}
    (\Bar{x},\Bar{z})\in\mathbb{F}_2^{2K},
\end{equation}
where $\Bar{x}_i$ and $\Bar{z}_i$ correspond to the measurement outcomes of the logical $\Bar{X}_i\otimes \Bar{X}_i$ and $\Bar{Z}_i\otimes \Bar{Z}_i$ operators, respectively, on $K$ pairs of logical qubits.
To obtain the $K$-bit outcomes $\bar{z}\in \mathbb{F}_2^{K}$ from the erroneous measurement outcomes $z\in\mathbb{F}_2^{N}$, we follow the same procedure as described for the $Z_K$-measurement gadget.
Since $\mathcal{Q}$ in~\eqref{eq: code parameters of quantum LDPC codes} is a CSS code, the $K$-bit outcomes $\bar{x}\in \mathbb{F}_2^{K}$ can also be estimated by replacing the $Z$-basis with the $X$-basis and performing the same procedure.

As a result, the width of the quantum part in the gadget is $2N$,
and the depth of the quantum part is $3$.
Thus, the fault-tolerance conditions in Def.~\ref{def: fault-tolerance condition of gate, measurement for qLDPC codes} are satisfied.
The same analysis as in the $Z$-measurements gadget bounds the runtime of the classical part of the Bell-measurement gadget as 
\begin{equation}
    O(\log{N}).
\end{equation}

\subsubsection{On-demand Pauli-gate gadgets\label{sec: Pauli-gate gadgets}}

An on-demand Pauli-gate gadget is intended to perform a Pauli gate
\begin{equation}
P=\bigotimes_{k=1}^{K}{P}_k\in\Tilde{\mathcal{P}}_K,
\label{eq: Pauli operator}
\end{equation}
where $P_k\in\{I,X,Y,Z\}$ is a Pauli operator acting on the $k$-th qubit, as a logical operation acting on logical qubits in a code block as
\begin{equation}
    \bar{P}=\bigotimes_{k=1}^{K}\Bar{P}_k\in\tilde{P}_N,
\end{equation}
where $\Bar{P}_k$ is the logical operator of $P_k$ in \eqref{eq: Pauli operator}.
Each $\Bar{P}_k$ is expressed as a tensor product of $N$-qubit Pauli operators acting on physical qubits, i.e.,
\begin{equation}
    \Bar{P}_k=\bigotimes_{n=1}^{N}P'_{n,k},
\end{equation}
where $P'_{n,k}\in\{I,X,Y,Z\}$ is a Pauli operator acting on the $n$-th physical qubit.
Thus, $\Bar{P}$ can be written as a tensor product of $N$-qubit Pauli operators, i.e.,
\begin{equation}
    \Bar{P}=\bigotimes_{k=1}^{K}\Bar{P}_k=\bigotimes_{k=1}^{K}\left(\bigotimes_{n=1}^{N}P'_{n,k}\right) \in\Tilde{\mathcal{P}}_N.
\label{eq: Pauli-gate operation gadget}
\end{equation}
The construction of on-demand Pauli-gate gadgets is shown in Fig.~\ref{Fig: Pauli-gate gadget.}.
The width of the gadget is $N$, and the depth of the gadget is $1$.
Thus, the fault-tolerance conditions in Def.~\ref{def: fault-tolerance condition of gate, measurement for qLDPC codes} are satisfied.

\begin{figure*}[t]
    \centering
    \includegraphics[width=\textwidth]{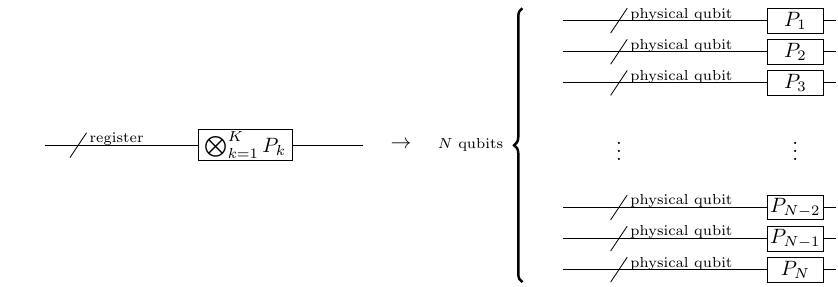}
    \caption{A physical circuit of the gadget of the on-demand Pauli gate operation~\eqref{eq: on-demand Pauli gate operations}.
    The gadget is implemented by an $n$-qubit Pauli operator acting on physical qubits in \eqref{eq: Pauli-gate operation gadget}.}
    \label{Fig: Pauli-gate gadget.}
\end{figure*}

By the constructions of the two-register Clifford-gate abbreviations and $U_T$-gate abbreviations, as explained in Sec.~\ref{sec: Two-register Clifford-gate abbreviation} and \ref{sec: U_T-gate abbreviation}, respectively, wait operations with depth $O(\log{K}) = O(\log{N})$ are always allocated before each on-demand Pauli-gate operation. 
This time is used for the classical computation required to determine the physical Pauli gates that implement the logical Pauli operation. 
The classical computation begins with the input of a $2K$-bit string, representing the symplectic form of the $K$-qubit Pauli operator $P \in \Tilde{\mathcal{P}}_K$ in \eqref{eq: Pauli operator}. The output is a $2N$-bit string representing the symplectic form of $\Bar{P} \in \Tilde{\mathcal{P}}_N$, as described in \eqref{eq: Pauli-gate operation gadget}. 
The task of computing \eqref{eq: Pauli-gate operation gadget} corresponds to summing these $K$ bitstrings according to \eqref{eq: multiplication of Pauli operators}. 
Each $n$-th bit for $n \in\{1, \ldots, N\}$ can be calculated simultaneously using $O(N)$ parallel processes. 
The sums of the $n$-th positions of the $K$ bitstrings are then computed using $O(K)$ parallel processes. 
Thus, with a total of $O(NK) = O(N^2)$ parallel processes, the computation can be completed within runtime $O(\log{K}) = O(\log{N})$,
which enables the generation of the required $2N$-bit string for this gadget.

\subsubsection{CNOT-gate gadget\label{sec: CNOT-gate gadget}}
The CNOT-gate gadget is designed to perform the logical CNOT gate between logical qubits in different code blocks $B_1, B_2$.
The construction of a CNOT-gate gadget is shown in Fig.~\ref{Fig: cnot gadget.}.
Since $\mathcal{Q}$ is a CSS code, the gadget is implemented by transversal CNOT-gate operations.
\begin{figure*}[t]
    \centering
    \includegraphics[width=\textwidth]{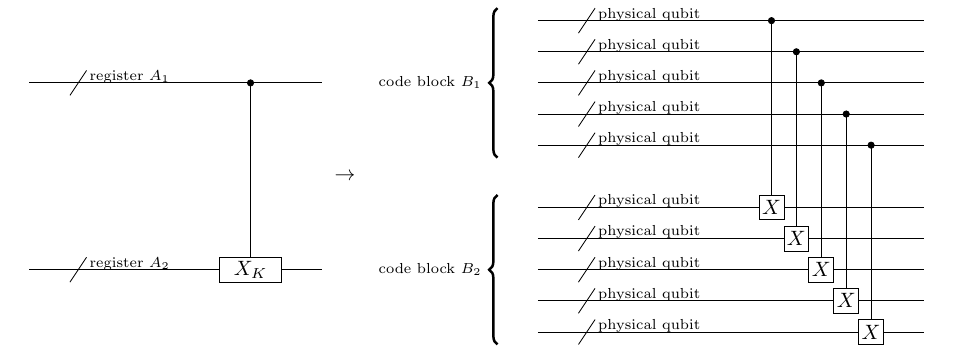}
    \caption{The physical circuit of the gadget for the CNOT-gate operation~\eqref{Fig: cnot-gate-ldpc}. The gadget is implemented by transversal CNOT gates, where the code block $B_1$ contains control qubits and the code block $B_2$ contains target qubits.}
    \label{Fig: cnot gadget.}
\end{figure*}
As a result, the depth of the quantum part of the gadget is $1$
and the width is $2N$.
Thus, the fault-tolerance conditions in Def.~\ref{def: fault-tolerance condition of gate, measurement for qLDPC codes} are satisfied.

\subsubsection{$\ket{0}^{\otimes K}$-state preparation gadget \label{sec: 0^K-state preparation gadget}}
The $\ket{0}^{\otimes K}$-state preparation gadget is designed to simulate an original open circuit $C_{N}^{\ket{0}}$ to prepare a logical state $\ket{\overline{0}}^{\otimes K}$ of $\mathcal{Q}$.
To construct $C_{N}^{\ket{0}}$, we start by considering an open circuit that is composed of state preparation of $\ket{0}^{\otimes N}$ followed by a circuit $U_{\mathrm{encode}}$ that encodes an arbitrary $K$-qubit state $\ket{\psi}$ into a logical state $\ket{\overline{\psi}}$ of $\mathcal{Q}$ 
as
\begin{equation}
    U_{\mathrm{encode}}\left(\ket{\psi}\otimes\ket{0}^{\otimes(N-K)}\right)=\ket{\overline{\psi}},
\label{eq: encoding circuit for stabilizer codes}
\end{equation} 
where we choose $\ket{\psi}=\ket{0}^{\otimes K}$, as shown in Fig.~\ref{Fig: clifford_state_original.}~(a).
For stabilizer codes, we can obtain such encoding circuit $U_{\mathrm{encode}}$ described by a stabilizer circuit~\cite{Cleve_encoding}.
Furthermore, any $n$-qubit stabilizer circuit has an equivalent $n$-qubit stabilizer circuit comprising $7$ rounds, each using only one type of Clifford gate~\cite{Maslov_2018}.
The sequence of gates follows the order: CNOT, CZ, $S$, $H$, $S$, CZ, CNOT, with a two-qubit gate depth of $14n - 4$.
Using this technique, we rewrite the circuit in \eqref{eq: encoding circuit for stabilizer codes} in such a way that it can be described in terms of the elementary operations defined in Appendix~\ref{appendix: fault-tolerant protocol for open quantum circuits}, finally yielding $C_{N}^{\ket{0}}$.
Let $W(C_{N}^{\ket{0}})$ and $D(C_{N}^{\ket{0}})$ denote the width and the depth of $C_{N}^{\ket{0}}$, respectively.
Under this construction, the width and the depth of $C_{N}^{\ket{0}}$ are given by, respectively,
\begin{align}
   &W(C_{N}^{\ket{0}})=N,
\label{width of 0-state-preparation in ldpc codes}\\
   &D(C_{N}^{\ket{0}})\leq c_D N,
\label{depth of 0-state-preparation in ldpc codes}
\end{align} 
where $c_D$ is a positive constant independent of $N$.

Thus, the gadget $\Tilde{C}_{N}^{\ket{0}}=\Tilde{C}(C_{N}^{\ket{0}}, L)$ satisfies the fault-tolerance conditions in Def.~\ref{def: fault-tolerance conditions of the state preparation gadgets for quantum LDPC codes}, where $C_{N}^{\ket{0}}$ is the original circuit and $L$ is a concatenation level, and the compilation procedure is detailed in Appendix~\ref{appendix: fault-tolerant protocol for open quantum circuits}.
 
\subsubsection{Clifford-state preparation gadget \label{sec: Clifford-state preparation gadget}}

\begin{figure*}[t]
    \centering
    \includegraphics[width=.9\textwidth]{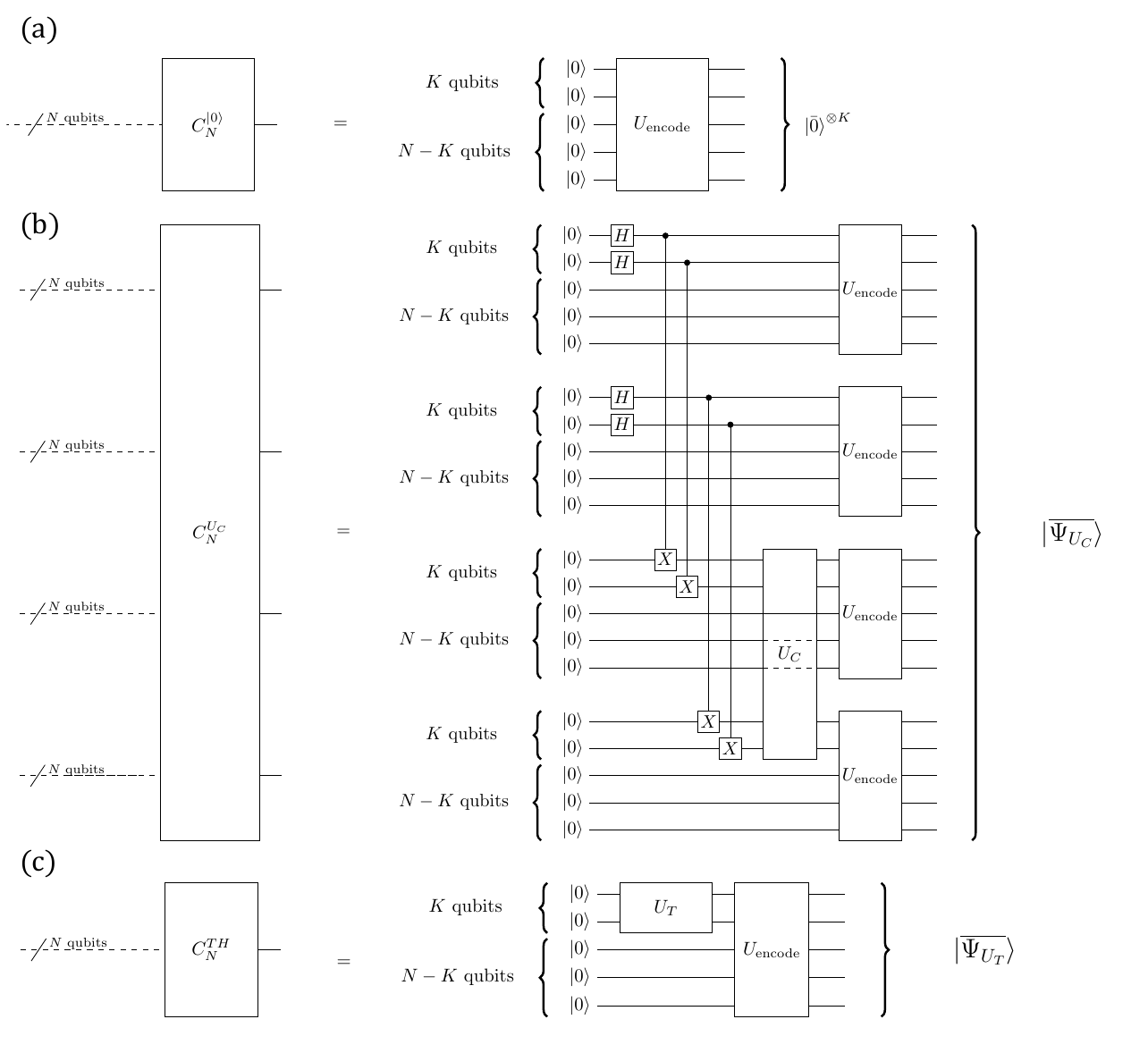}
    \caption{An original open circuit $C_N^{\ket{0}}$ is shown in (a).
    Here, $C_N^{\ket{0}}$ is equivalent to an open circuit to prepare logical $\ket{0}^{\otimes K}$ state of $\mathcal{Q}$ that consists of the preparation of $\ket{0}^{\otimes N}$, followed by an encoding circuit $U_{\mathrm{encode}}$ to encode an arbitrary $K$-qubit state $\ket{\psi}$ into a logical state $\ket{\Bar{\psi}}$ of $\mathcal{Q}$.
    An original open circuit $C_N^{U_C}$ is shown in (b).
    Here, $C_N^{U_C}$ is equivalent to an open circuit to prepare a logical state of $\ket{\Psi_{U_C}}=(I\otimes U_C)\ket{\Omega}$, where $\ket{\Omega}$ is the maximally entangled state in \eqref{eq: maximally entangled state between 4 registers}.
    An original open circuit $C_N^{TH}$ is shown in (c).
    Here, $C_N^{TH}$ is equivalent to an open circuit to prepare a logical state of $\ket{\Psi_{U_{T}}}=U_{TH}\ket{0}^{\otimes K}$ of $\mathcal{Q}$, where $U_{TH}$ is a tensor product of $TH$ and $I$ gate, and $TH$ is applied to the qubits to which $T$ gate or $T^{\dagger}$ gate will be applied using the $U_T$-gate abbreviation in \eqref{Fig: magic-abbreviation-ldpc.}.
    A physical circuit of the state-preparation gadgets are obtained from the protocol for simulating open circuits, as explained in Appendix~\ref{appendix: fault-tolerant protocol for open quantum circuits}.}
    \label{Fig: clifford_state_original.}
\end{figure*}

The Clifford-state preparation gadgets are designed to simulate an original open circuit ${C}_N^{U_C}$ to prepare a logical state $\ket{\overline{\Psi_{U_C}}}^{B_1 B_2 B_3 B_4}$ of four code blocks $B_1, B_2, B_3$, and $B_4$, each encoded in $\mathcal{Q}$, where
\begin{equation}
\begin{split}
\ket{\Psi_{U_C}} = (I\otimes U_{C})\ket{\Omega},
\end{split}
\label{eq: the state generated by clifford-state preparaion}
\end{equation}
$\ket{\Omega}$ is the maximally entangled state, and $U_C\in\Tilde{{\mathcal{C}}}_{2K}/\Tilde{{\mathcal{P}}}_{2K}$ is a Clifford unitary in \eqref{Fig: clifford-state-preparation-ldpc}.

To construct $C_{N}^{U_C}$, we start with an open circuit that prepares the state $\ket{\Psi_{U_C}}$ and then applies parallel encoding circuits $U_{\mathrm{encode}}$ in~\eqref{eq: encoding circuit for stabilizer codes} to prepare $\ket{\overline{\Psi_{U_C}}}$ across the code blocks, as shown in Fig.~\ref{Fig: clifford_state_original.}~(b). 
As with the $\ket{0}^{\otimes K}$-state preparation gadget in Sec.~\ref{sec: 0^K-state preparation gadget}, using the technique in Ref.~\cite{Maslov_2018}, we rewrite this open circuit in such a way that it can be described in terms of the elementary operations defined in Appendix~\ref{appendix: fault-tolerant protocol for open quantum circuits}, thereby yielding ${C}_N^{U_C}$.
Under this construction, the width and depth of $C_{N}^{U_C}$ are given by, respectively, 
\begin{equation} 
W({C}_N^{U_C}) = 4N, 
\label{width of Clifford-state-preparation in ldpc codes}
\end{equation} and 
\begin{equation} 
D({C}_N^{U_C}) \leq c_D N, 
\label{depth of Clifford-state-preparation in ldpc codes}
\end{equation} 
where $c_D\geq 1$ is a positive constant.

Thus, the gadget $\Tilde{C}_{N}^{U_C}=\Tilde{C}(C_{N}^{U_C}, L)$ satisfies the fault-tolerance conditions in Def.~\ref{def: fault-tolerance conditions of the state preparation gadgets for quantum LDPC codes}, where $C_{N}^{U_C}$ is the original circuit and $L$ is a concatenation level, and the compilation procedure is detailed in Appendix~\ref{appendix: fault-tolerant protocol for open quantum circuits}.
\begin{figure*}
    \centering
    \includegraphics[width=1.\textwidth]{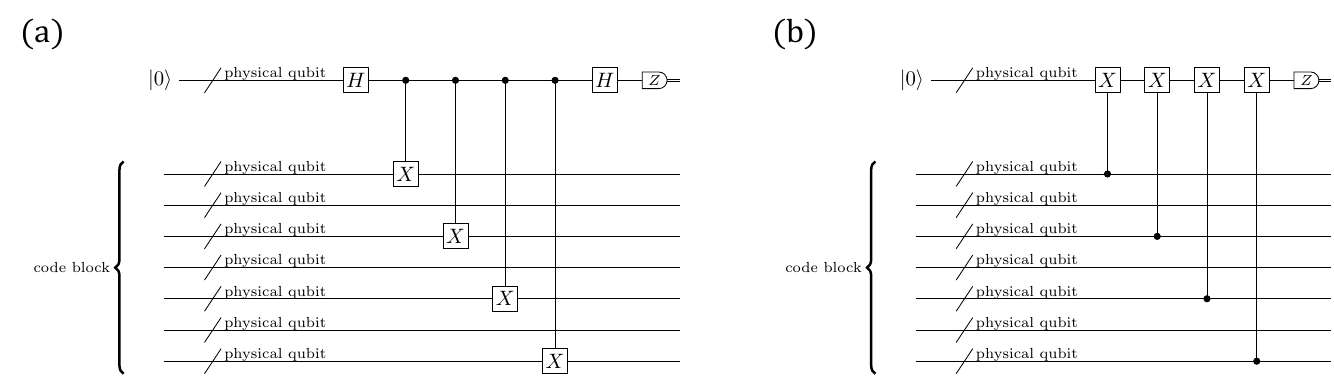}
    \caption{The circuit for measuring an $X$-type stabilizer generator is shown in (a), where the target qubits of the CNOT gates correspond to the support of the $X$-type stabilizer generator, and the circuit for measuring a $Z$-type stabilizer generator is shown in (b), where the control qubits of the CNOT gates correspond to the support of the $Z$-type stabilizer generator.
For each generator, a single auxiliary physical qubit is prepared, while the $N$ physical qubits correspond to the physical qubits in $\mathcal{Q}$~\eqref{eq: code parameters of quantum LDPC codes}.
The syndrome measurement circuit is a circuit for measuring all stabilizer generators, where each generator is measured using the corresponding circuit shown in this figure.}
    \label{fig: syndrome-measurement-circuit-ldpc}
\end{figure*}

\subsubsection{Magic-state preparation gadget\label{sec: Magic-state preparation gadget}}

The magic-state preparation gadgets are designed to simulate an original open circuit $C_N^{TH}$ to prepare a logical state $\ket{\overline{\Psi_{U_{TH}}}}$ of $\mathcal{Q}$, where
\begin{equation}
\begin{split}
    \ket{\Psi_{U_{TH}}}= U_{TH}\ket{0}^{\otimes K}.
\end{split}
\label{eq: state created by the magic state preparation gadget}
\end{equation}
Here, ${U_{TH}}$ is a tensor product of $TH$ gate and the $I$ gate, where $TH$ is applied to the logical qubits to which $T$ gate or $T^{\dagger}$ gate will be applied using the $U_T$-gate abbreviation explained in Sec.~\ref{sec: U_T-gate abbreviation}.

To construct $C_N^{TH}$, we start with an open circuit that prepares the state $\ket{\Psi_{U_{TH}}}$, followed by the encoding circuit $U_{\mathrm{encode}}$ to obtain the encoded state $\ket{\overline{\Psi_{U_{TH}}}}$, as shown in Fig.~\ref{Fig: clifford_state_original.}~(c).
By Applying the same approach to $U_{\mathrm{encode}}$ as used for the $\ket{0}^{\otimes K}$-state preparation gadget and the Clifford-state preparation gadget, we obtain $C_N^{TH}$, which is equivalent to the open circuit.
Since the original circuit $W$ consists of depth-$2$ circuit of $TH$, followed by stabilizer circuit $U_{\mathrm{encode}}$, the width and the depth of $C_N^{TH}$ are given by, respectively,
\begin{align}
    &W(C_{N}^{TH})\leq c_WN,
\label{width of magic-state-preparation in ldpc codes}\\
    &D(C_{N}^{TH})\leq c_D N,
\label{depth of magic-state-preparation in ldpc codes}
\end{align} 
where $c_W\geq 1$ and $c_D\geq 1$ are constants independent of $N$.

Thus, the gadget $\Tilde{C}_{N}^{TH}=\Tilde{C}(C_{N}^{TH}, L)$ satisfies the fault-tolerance conditions in Def.~\ref{def: fault-tolerance conditions of the state preparation gadgets for quantum LDPC codes}, where $C_{N}^{TH}$ is the original circuit and $L$ is a concatenation level, and the compilation procedure is detailed in Appendix~\ref{appendix: fault-tolerant protocol for open quantum circuits}.

\subsubsection{Error-correction gadget \label{sec: error-correction gadget}}

\begin{figure*}
    \centering
    \includegraphics[width=1.\textwidth]{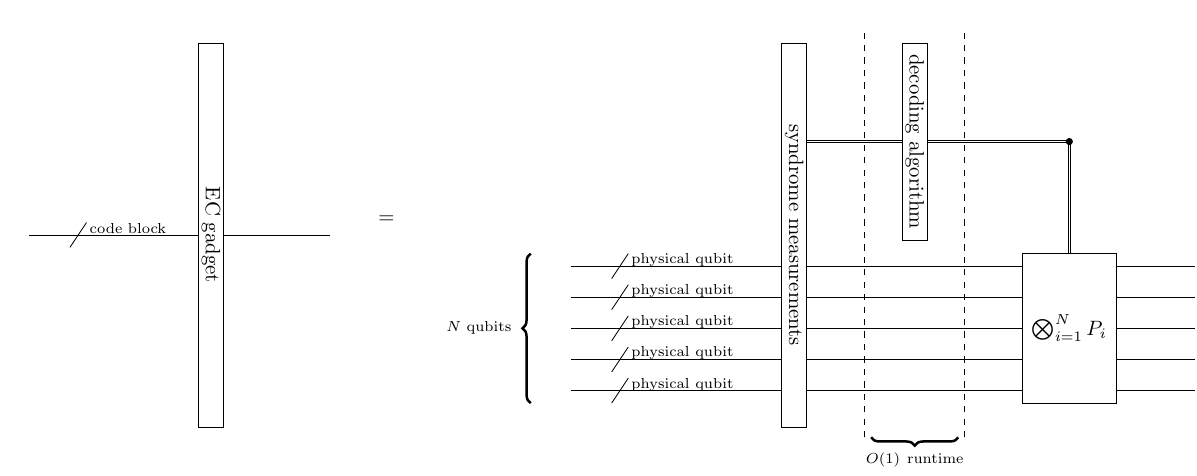}
    \caption{A physical circuit of the EC gadget is shown. 
    The gadget is implemented by the syndrome measurement circuit, followed by classical computation for deducing a recovery operation $\bigotimes_{i=1}^{N}P_i\in\Tilde{\mathcal{P}}_N$ using a single-shot decoding algorithm, as described in Def.~\ref{def: single-shot decoding algorithm with thresholds to suppress residual error}. Finally, the deduced recovery operation is applied.}
    \label{fig: ec-gadget-ldpc}
\end{figure*}
The EC gadget is designed to keep residual errors in a code block small using the single-shot decoding algorithm in Def.~\ref{def: single-shot decoding algorithm with thresholds to suppress residual error}.
The construction of an EC gadget is shown in Figs.~\ref{fig: syndrome-measurement-circuit-ldpc} and~\ref{fig: ec-gadget-ldpc}.
For each code block, $N-K=O(N)$ auxiliary physical qubits are used to perform the syndrome measurement, using the circuits in Fig.~\ref{fig: syndrome-measurement-circuit-ldpc} to measure a specific stabilizer generator.
Each auxiliary physical qubit is assigned to a corresponding stabilizer generator.
The measurements of all generators make up the syndrome measurement circuit of the EC gadget as shown in Fig.~\ref{fig: ec-gadget-ldpc}.

The measurements of the generators provide erroneous syndrome bits
\begin{equation}
    \Tilde{\sigma}=(\Tilde{\sigma}_X,\Tilde{\sigma}_Z)\in\mathbb{F}_2^{N-K},
\label{eq: recovered syndrome bits}
\end{equation}
as in~\eqref{eq: noisy syndrome}.
Note that Ref.~\cite{Gottesman2014Constant} uses Shor's error correction~\cite{shor1997faulttolerant}, which utilizes a cat state consisting of multiple qubits for each stabilizer generator, where the number of qubits in the cat state corresponds to the weights of the stabilizer generator.
However, the EC gadgets only need to satisfy our fault-tolerance conditions in Def.~\ref{def: fault-tolerance condition of gate, measurement for qLDPC codes}; therefore, a single auxiliary physical qubit per stabilizer generator is sufficient.
Using the erroneous syndrome bits $\Tilde{\sigma}$, the single-shot decoding algorithm in Def.~\ref{def: single-shot decoding algorithm with thresholds to suppress residual error} with $T$ internal loops deduces a recovery operation, and then the recovery operation is applied to physical qubits in the EC gadget to keep a residual error on the physical qubits small.

Let us divide the EC gadget into the syndrome measurement circuit part and the remaining part.
For the syndrome measurement circuit part, recall that each code block is encoded with a non-vanishing-rate CSS $(r,c)$ LDPC code $\mathcal{Q}$.
Measurements are performed in parallel for generators acting on disjoint physical qubits, and the depth of each quantum circuit to measure a single stabilizer generator is bounded by $(r+4)$; thus, the depth the syndrome measurement, denoted by $d^\synd$, is a constant bounded by
\begin{equation}
    d^\synd \leq (r+4)rc.
\label{eq: the depth of the quantum part of the EC gadget}
\end{equation}

For the remaining part, we choose an integer $T$ satisfying 
\begin{align}
\label{eq: T in EC gadgets}
    c(T)\geq \frac{\log \lambda}{\log p_\dec^\thre},
\end{align}
where $c(T)$ is the error suppression parameter, which is a monotonically increasing function with respect to $T$ and $p_\mathrm{dec}^\mathrm{th}\coloneqq \min\{p_1^\thre, p_2^\thre\}>0$ where $p_1^\thre$ and $p_2^\thre$ are the threshold constants of the decoding algorithm in Defs.~\ref{def: single-shot decoding algorithm with thresholds to suppress residual error} and \ref{def: Logarithmic-time decoding algorithm with threshold}.
With this choice, $T$ does not depend on $N$.
Here, $\lambda$ is a positive constant satisfying, with some constant $d_0\geq d^\synd$, 
\begin{align}
    \lambda<\qty(\frac{p_\dec^\thre}{ 2^{2d_0}})^{2^{2d_0}}.
\end{align}
The width of the gadget is
\begin{equation}
    N+M=N+(N-K)=2N-K,
\end{equation}
where $M=N-K$ is the number of auxiliary physical qubits for measuring stabilizer generators.
Therefore, the EC gadget is fault-tolerant in the sense of Def.~\ref{def: fault-tolerance conditions of the EC gadgets for quantum LDPC codes}.

\section{Threshold theorem and overhead analysis\label{sec: threshold theorem based on quantum LDPC codes}}
In this section, we prove the threshold theorem of our protocol.
\begin{theorem} (Threshold theorem for polylog-time- and constant-space-overhead fault-tolerant quantum computation).\label{theorem: threshold theorem for polylog-time- and constant-space-overhead fault-tolerant quantum computation}
Let $\qty{C^{\mathrm{org}}_n}_n$ be a sequence of original circuits, where each circuit $C_n^{\mathrm{org}}$ is specified by an integer $n$ and has width $W(n)\rightarrow \infty$ and depth $D(n)=O(\mathrm{poly}(W(n)))$ as $n\to\infty$.
Suppose that an original circuit $C^{\mathrm{org}}_n$ is compiled into a fault-tolerant circuit $C^{\mathrm{FT}}_n$ as described in Sec.~\ref{sec: Compilation of ideal quantum circuit into fault-tolerant quantum circuit}, and $C^{\mathrm{FT}}_n$ experiences a local stochastic Pauli error model with error parameters $p_i\leq p_\loc \text{ ~for all locations ~} i\in C_n^\FT$ satisfying Assumption~\ref{assumption: fault-tolerant circuit experiences the local stochastic Pauli error model}.

Suppose there exists a family of CSS LDPC codes $\{\mathcal{Q}_i\}$, indexed by an integer $i$, satisfying Assumption~\ref{assump: non-vanishing-rate quantum LDPC with efficient decoding algorithm}, i.e., the family has a single-shot decoding algorithm for reducing errors in Def.~\ref{def: single-shot decoding algorithm with thresholds to suppress residual error}, and a decoding algorithm for correcting errors in Def.~\ref{def: Logarithmic-time decoding algorithm with threshold}.
Furthermore, assume there exist constants $\beta\leq 1$ and $0<\gamma\leq 1$ such that $\mathcal{Q}_i$ has parameters 
\begin{equation}
    [[N_i,K_i=\Theta(N_i),D_i=\Theta(N_i^{\gamma})]],
\label{eq: parameters of quantum LDPC codes}
\end{equation}
where $N_i$ is monotonically increasing with $N_i\rightarrow\infty$ as $i\rightarrow\infty$, and 
\begin{equation}
\label{eq: beta in threshold theorem}
    0<N_i-N_{i-1}\leq N_{i-1}^\beta.
\end{equation}
The gadgets are constructed as described in Sec.~\ref{sec: Construction of fault-tolerant gadgets and abbreviations(qLDPC)}, i.e., state-preparation, gate, measurement, and EC gadgets satisfy the fault-tolerance conditions in Defs.~\ref{def: fault-tolerance condition of gate, measurement for qLDPC codes},~\ref{def: fault-tolerance conditions of the state preparation gadgets for quantum LDPC codes}, and \ref{def: fault-tolerance conditions of the EC gadgets for quantum LDPC codes}.

Then, there exists a threshold $q_{\mathrm{loc}}^{\mathrm{th}}>0$ such that if $p_\mathrm{loc}<q_{\mathrm{loc}}^{\mathrm{th}}$,
the following statement holds: for any target error $\varepsilon>0$,
there exists a sequence of fault-tolerant circuits $\qty{C^{\mathrm{FT}}_n}_n$ with width $W_{\mathrm{FT}}(n,\varepsilon)$ and depth $D_{\mathrm{FT}}(n,\varepsilon)$ achieving the constant space overhead and the polylogarithmic time overhead,
\begin{align}
    \frac{W_{\mathrm{FT}}(n,\varepsilon)}{W(n)}&=O(1),\\
    \frac{D_{\mathrm{FT}}(n,\varepsilon)}{D(n)}&=O\left(\mathrm{polylog}\left(\frac{|C^{\mathrm{org}}_n|}{\varepsilon}\right)\right),
\end{align}
as $n\rightarrow\infty$,
and $C^{\mathrm{FT}}_n$ outputs a bitstring sampled from the probability distribution that is close to that of $C^{\mathrm{org}}_n$ within $\varepsilon$ in total variation distance.
\end{theorem}

First, for a given $n$ and a fixed target error $\varepsilon$, we choose the parameters of our protocol as follows.
We choose a constant $\alpha$ such that
\begin{equation}
\label{eq: constant relation 1 in threshold theorem}
    \alpha\gamma>1.
\end{equation}

For any positive constant $c_N>0$, we define the lower bound $N_{\min}$ for the code distance as
\begin{align}
    &N_{{\min}}\coloneqq c_N\log^{\alpha}\qty(\frac{|C_n^{\org}|}{\varepsilon}).
\end{align}
We then choose the minimum integer satisfying 
\begin{equation}
    i_{\min}\coloneqq\min\{i\colon \mathcal{Q}_i \text{~with~} N_i\geq N_{\mathrm{min}}\}.
\end{equation}
From the family of codes $\{\mathcal{Q}_i\}_i$ with $[[N_i, K_i, D_i]]$, choose a code used to construct the gadgets as
\begin{equation}
\label{eq: code parameters in threshold theorem}
    [[N(n,\varepsilon)=N_{i_{\min}},K(n,\varepsilon)=K_{i_{\min}},D(n,\varepsilon)=D_{i_{\min}}]].
\end{equation}
Due to~\eqref{eq: beta in threshold theorem}, we have
\begin{align}
    N_{\min}\leq N_{i_{\min}}=N(n,\varepsilon)
    &\leq N_{i_{\min}-1}+\qty(N_{i_{\min}-1})^\beta \\
    &\leq 2N_{i_{\min}-1}\\
    &\leq 2c_N N_{\min},
    \label{eq: Nupper in threshold theorem}
\end{align}
which leads to
\begin{align}
    &N(n,\varepsilon)=\Theta\qty(\log^{\alpha}\qty(\frac{|C_n^{\org}|}{\varepsilon})).\label{eq: Ntheta in threshold theorem}
\end{align}
Due to \eqref{eq: code parameters in threshold theorem}, for sufficiently large $n$, we have
\begin{align}
    &K(n,\varepsilon)=\Theta\qty(\log^{\alpha}\qty(\frac{|C_n^{\org}|}{\varepsilon})),\label{eq: Ktheta in threshold theorem}\\
    &D(n,\varepsilon)=\Theta\qty(\log^{\alpha\gamma}\qty(\frac{|C_n^{\org}|}{\varepsilon})).\label{eq: Dtheta in threshold theorem}
\end{align}

For state preparation gadgets, we consider an open circuit $C_N$ used to prepare logical states of $\mathcal{Q}$.
From Lemma~\ref{lemma: fault-tolerance condition}, there exists positive constants $c_{W}$ and $c_{D}$, such that
\begin{align}
    W(C_N)&\leq c_{W}N,\label{eq: W(C_N) in threshold theorem}\\
    D(C_N)&\leq c_{D}N\label{eq: D(C_N) in threshold theorem},
\end{align}
which leads to 
\begin{align}
    |C_N|=W(C_N)D(C_N)\leq c_{W}c_{D}N\label{eq: |C_N| in threshold theorem}.
\end{align}
Let us choose an integer $L$ representing the concatenation level as 
\begin{align}
\label{eq: L in threshold theorem}
    L(n,\varepsilon)=\left\lceil \log_2 \log_{p^\thre_\loc/p_\loc}\qty(\frac{|C_n^\org|}{\varepsilon})^3\right\rceil,
\end{align}
where $p_\loc^\thre\in(0,1]$ is a constant independent of $n$ and $\varepsilon$, and $p_\loc\leq p_\loc^\thre$.
With this $L$, from Theorem~\ref{Theorem: level-reduction for the circuit that outputs a quantum state}, there exist positive constants $c_1, c_2, \gamma_1$, and $\gamma_2$, such that the width $W(\tilde{C}(C,L))$ and the depth $D(\tilde{C}(C,L))$ of the state preparation gadget $\tilde{C}(C,L)$ are bounded by 
\begin{align}
    \Tilde{W}^\prep(n,\varepsilon)\coloneqq W(\tilde{C}(C,L))&\leq c_1 W(C_N)2^{\gamma_1 L},\\
    \Tilde{D}^\prep(n,\varepsilon)\coloneqq D(\tilde{C}(C,L))&\leq c_2 D(C_N)2^{\gamma_2 L},
\end{align}
From \eqref{eq: Nupper in threshold theorem}, \eqref{eq: W(C_N) in threshold theorem}, \eqref{eq: D(C_N) in threshold theorem}, and \eqref{eq: L in threshold theorem}, there exists positive constants $c_{\Tilde{W}}$ and $c_{\Tilde{D}}$ such that
\begin{align}
     \Tilde{W}^\prep(n,\varepsilon)\leq c_{\tilde{W}}\log^{\alpha+\gamma_1}\qty(\frac{|C_n^\org|}{\varepsilon}),\label{eq: Wprep in threshold theorem}\\
     \Tilde{D}^\prep(n,\varepsilon)\leq c_{\tilde{D}}\log^{\alpha+\gamma_2}\qty(\frac{|C_n^\org|}{\varepsilon}).\label{eq: Dprep in threshold theorem}
\end{align}

Fix the parameters of the EC gadgets.
From Lemma~\ref{lemma: fault-tolerance condition}, there exists a positive constant $d_0$ such that the depth of each gate gadget, measurement gadget, and syndrome measurement circuit in the EC gadget is bounded by $d_0$.
Let $p_{\mathrm{dec}}^{\thre}\coloneqq \min\{p_1^\thre,p_2^\thre\}>0$, where $p_1^\thre$ and $p_2^\thre$ are the threshold values of the decoding algorithm in Defs.~\ref{def: single-shot decoding algorithm with thresholds to suppress residual error} and \ref{def: Logarithmic-time decoding algorithm with threshold}, respectively.
Choose an integer $T$ representing the runtime of classical computation of the single-shot decoding algorithm, such that it satisfies 
\begin{align}
\label{eq: c(T) in threshold theorem}
    c(T)\geq \frac{\log \lambda}{\log p_\dec^\thre},
\end{align}
where $c(T)$ is an error suppression parameter, and $\lambda$ is a positive constant satisfying
\begin{align}
\label{eq: lambda in threshold theorem}
    \lambda<\qty(\frac{p_\dec^\thre}{  2^{2d_0}})^{2^{2d_0}},
\end{align}
The choice of $T$ also determines the depth of the second part of the EC gadget (see Def.~\ref{def: fault-tolerance conditions of the EC gadgets for quantum LDPC codes}), which is independent of $n$ and $\varepsilon$. 
Thus, the entire depth 
\begin{equation}
\label{eq: d in threshold theorem}
d~(>d_0)
\end{equation}
of the EC gadget is also independent of $n$ and $\varepsilon$ as was assumed in the compiling procedure in Sec.~\ref{sec: Compilation of ideal quantum circuit into fault-tolerant quantum circuit}.

Finally, the number of registers $\kappa(n,\varepsilon)$ in \eqref{eq: kappa in compilation} is determined as
\begin{equation}
    \kappa(n,\varepsilon)\coloneqq \left\lceil\frac{W(n)}{K(n,\varepsilon)}\right\rceil.
\label{eq: kappa in threshold theorem}
\end{equation}
Choose the number of non-trivial intermediate operations $\tau(n,\varepsilon)$ in \eqref{eq: the maximum number of non-trivial intermediate operations L(n,varepsilon)} that we apply in a one-depth part of an intermediate circuit as
\begin{align}
\tau(n,\varepsilon)&\coloneqq\left\lceil \frac{W(n)}{\tilde{W}^\prep(n,\varepsilon )}\right\rceil.
\label{eq: tau in threshold theorem}
\end{align}

As a result, for a given sufficiently large $n$ and the target error $\varepsilon$, with this choice of the parameters of our fault-tolerant protocol, i.e., $[[N(n,\varepsilon),K(n,\varepsilon), D(n,\varepsilon)]]$, $L(n,\varepsilon)$, $T$, $\kappa(n,\varepsilon)$, $\tau(n,\varepsilon)$, we can obtain the fault-tolerant circuit $C_n^\FT$ of $C_n^\org$ as described in Sec.~\ref{sec: Compilation of ideal quantum circuit into fault-tolerant quantum circuit}.
In the rest of this section, we will show the threshold theorem and analyze the space and time overheads of our protocol under these parameters.

\subsection{Analysis of constant space overhead and polylogarithmic time overhead}

We show that a sequence of fault-tolerant circuit $\{C_n^\mathrm{FT}\}_n$ compiled from a corresponding sequence of original circuits $\{C_n^\org\}_n$ achieves a polylogarithmic time and a constant space overhead (recall Sec.~\ref{sec:setting} for the relevant definitions).

To this end, we begin by analyzing the compilation of an original circuit into an intermediate circuit as described in Sec.~\ref{sec:Compilation from original circuit to intermediate circuit}, and count the depth of the resulting intermediate circuit. 
Subsequently, we count the depth of a fault-tolerant circuit obtained from the intermediate circuit, following the compilation procedure presented in Sec.~\ref{subsubsec: Compilation from intermediate circuit to fault-tolerant circuit}. 
We then enumerate the number of rectangles in the fault-tolerant circuit and bound the error probability of the fault-tolerant simulation, specifically the total probability that accurate simulation of at least one of the rectangles fails.
Finally, we show that our protocol achieves a constant space overhead and a polylogarithmic time overhead.

In the compilation from the original circuit to the intermediate circuit, we replace operations in each one-depth part of the original circuit with intermediate operations.
The single-depth part of the $\ket{0}$-state preparations at the first time step of the original circuit is replaced with a single-depth part of $\ket{0}^{\otimes K}$-state preparations in the intermediate circuit, and the single-depth part of the $Z$-basis measurements before the final time step is replaced with a single-depth part of the $Z_K$ measurement operations.
The single-depth part of the middle of the original circuit is decomposed into two single-depth circuits: one composed of $T$, $T^{\dagger}$, and $I$ gates, and the other composed of Clifford gates.
For the single-depth circuit described by $T$, $T^{\dagger}$, and $I$ gates, we can replace it with a single-depth circuit of $U_T$-gate abbreviations.
However, the single-depth circuit composed of Clifford gates requires a different treatment. While a single use of the two-register Clifford-gate abbreviation allows the simultaneous application of multiple Clifford gates to logical qubits in a single pair of registers,
we need multiple two-register Clifford-gate abbreviations to implement Clifford gates on logical qubits in different pairs of registers.

To bound the depth of the intermediate circuit required for performing all combinations of the two-qubit Clifford gates in a one-depth part of Clifford gates of $C^{\mathrm{org}}_n$, we calculate combinations of two-register Clifford-gate abbreviations that can be executed simultaneously, using a solution of the following edge coloring problem.  
We consider an undirected graph $G = (V, E)$, where each vertex in $V$ corresponds to one of the registers, and each edge in $E$ connects vertices representing registers on which a two-qubit Clifford gate acts in a single-depth part of the Clifford layer (if the gate acts on the same register, the edge forms a self-loop). 
To simplify the analysis, we additionally define a graph $G'$ by removing all self-loops from $G$. 
In the graph $G'$, when each pair of adjacent edges is assigned a different color such that no two neighboring edges have the same color, the two-qubit Clifford gates assigned to each color are implemented as one-depth executable two-register Clifford-gate abbreviations in the intermediate circuit. 
Since each register has $K$ qubits, a qubit in a register can be connected to at most $K$ qubits in other registers by a two-qubit Clifford gate. 
In this case, $G'$ is a simple graph and the degree of $G'$ is at most $K$; hence, $G'$ is $(K+1)$-edge colorable~\cite{Vizing1964,Graph_Theory_Bondy}.
Note that the $(K+1)$-edge coloring of $G'$ can be found by classical computation within a polynomial time in terms of the number of vertices and edges~\cite{Arjomandi1982}.
When performing the two-register Clifford-gate abbreviation, if a vertex has both a self-loop and edges connecting to other vertices, color the edge corresponding to the self-loop with the same color as the first two-register Clifford gate abbreviation executed by the vertex with the self-loop. 
If a vertex has only a self-loop, color it with the color of the edge executed in the first time step.
Thus, any possible combination of Clifford gates in each one-depth part of the original circuit can be performed using two-register Clifford-gate abbreviations of depth up to $\left(K+1\right)$ in the intermediate circuit.
As a result, a single-depth part of the original circuit has the depth at most
\begin{equation}
    K(n,\varepsilon)+2.
\end{equation}
This depth is primarily dominated by the replacement of a single-depth part of the middle to $U_T$-gate abbreviations, followed by two-register Clifford abbreviations.

Next, we restrict the number of non-trivial intermediate operations $\tau(n,\varepsilon)$ that we apply in a one-depth part of the intermediate circuit to~\eqref{eq: tau in threshold theorem} in order to reduce the gate parallelism and achieve a constant space overhead.
As a result, upon our compilation procedure, the intermediate circuit compiled from a single-depth part of the original circuit has depth at most
\begin{align}
    &\qty(K(n,\varepsilon)+2)\times \frac{\kappa(n,\varepsilon)}{\tau(n,\varepsilon)}.
    \label{eq: the depth of intermediate circuit compiled from a one-depth part of the original circuit}
\end{align}

In the compilation from the intermediate circuit to the fault-tolerant circuit, we replace all the elementary operations contained in the intermediate operation with the corresponding rectangle, and for the sake of synchronization, the depth of each gadget can extend by up to a depth of $2d$, where $d$ in \eqref{eq: d in threshold theorem} is the depth of an EC gadget.
Since the state-preparation gadgets have the depth $\tilde{D}^\prep$ in \eqref{eq: Dprep in threshold theorem}, which is the largest depth among all the gadgets due to Lemma~\ref{lemma: fault-tolerance condition}, the depth of a physical circuit compiled from a one-depth part of the intermediate circuit is bounded by
\begin{align}
&\Tilde{D}^{\mathrm{prep}}(n,\varepsilon)+3d.
\label{eq: depth of a physical circuit compiled from a one-depth part of the intermediate circuit}
\end{align}
Thus, for sufficiently large $n$, the fault-tolerant circuit compiled from the original circuit has depth $D_{\mathrm{FT}}(n,\varepsilon)$ bounded by
\begin{align}
    &D_{\mathrm{FT}}(n,\varepsilon)\nonumber\\
    &=D(n)\times \left((K(n,\varepsilon)+2)\times \frac{\kappa(n,\varepsilon)}{\tau(n,\varepsilon)}\right)\times\nonumber\\
    &\qquad\qty(\Tilde{D}^{\mathrm{prep}}(n,\varepsilon)+3d)\\
    &\leq D(n)\times \qty(K(n,\varepsilon)+2)\times \qty(\frac{W(n)}{K(n,\varepsilon)}+1)\times\nonumber\\
    &\qquad \frac{\tilde{W}^\prep(n,\varepsilon)}{W(n)}\times \qty(\tilde{D}^\prep(n,\varepsilon)+3d),
\label{eq: depth of the fault-tolerant circuit}
\end{align}
where we used~\eqref{eq: kappa in threshold theorem} and~\eqref{eq: tau in threshold theorem}.

Next, we consider the width of the fault-tolerant circuit.
Since the width of the state-preparation gadget is given by \eqref{eq: Wprep in threshold theorem}, due to Lemma~\ref{lemma: fault-tolerance condition}, the width of state-preparation gadgets is larger than gate, measurement, and EC gadgets.
Thus, the width of a fault-tolerant circuit compiled from non-trivial intermediate operation is bounded by 
\begin{align}
\Tilde{W}^{\mathrm{prep}}(n,\varepsilon)\times \tau(n,\varepsilon).\label{eq: the width of a physical circuit compiled from a single non-trivial intermediate operation}
\end{align}
From Lemma~\ref{lemma: fault-tolerance condition}, the width of an EC gadget is $c_W N(n,\varepsilon )$.
Thus, the width of a fault-tolerant circuit compiled from a trivial intermediate operation is bounded by 
\begin{align}
\label{eq: the width of a physical circuit compiled from a single trivial intermediate operation}
     c_W N(n,\varepsilon )\times \qty(\kappa(n,\varepsilon)-\tau(n,\varepsilon)).
\end{align}
As a result, the width of the fault-tolerant circuit can be bounded by
\begin{align}
    &W_{\mathrm{FT}}(n,\varepsilon)\nonumber\\
    &=\Tilde{W}^{\mathrm{prep}}(n,\varepsilon)\times \tau(n,\varepsilon)+\nonumber\\
    &\hspace{1cm} c_W N(n,\varepsilon )\times \qty(\kappa(n,\varepsilon)-\tau(n,\varepsilon))\\
    &=\Tilde{W}^{\mathrm{prep}}(n,\varepsilon)\times\left\lceil \frac{W(n)}{\tilde{W}^\prep(n,\varepsilon )}\right\rceil+\nonumber\\
    &\hspace{1cm} c_W N(n,\varepsilon )\times \qty(\left\lceil\frac{W(n)}{K(n,\varepsilon)}\right\rceil-\left\lceil \frac{W(n)}{\tilde{W}^\prep(n,\varepsilon )}\right\rceil)\\
    &\leq\Tilde{W}^{\mathrm{prep}}(n,\varepsilon)\times\qty(\frac{W(n)}{\tilde{W}^\prep(n,\varepsilon )}+1)+\nonumber\\
    &\hspace{1cm} c_W N(n,\varepsilon )\times \qty(\frac{W(n)}{K(n,\varepsilon)}+1), 
\end{align}
where we used~\eqref{eq: kappa in threshold theorem} and~\eqref{eq: tau in threshold theorem}.
Therefore, the space overhead of our protocol is bounded by
\begin{align}
    &\frac{W_{\mathrm{FT}}(n,\varepsilon)}{W(n)}\nonumber\\
    &\leq 1+\frac{\tilde{W}^\prep(n,\varepsilon )}{W(n)}+c_W\frac{N(n,\varepsilon)}{K(n,\varepsilon)}+c_W\frac{N(n,\varepsilon)}{W(n)}\\
    &=O\qty(1+\frac{\log^{\alpha+\gamma_1}\qty(\frac{|C_n^\org|}{\varepsilon})}{W(n)})\\
    &=O(1),
\end{align}
where the second line follows from~\eqref{eq: Ntheta in threshold theorem},~\eqref{eq: Ktheta in threshold theorem},~\eqref{eq: Wprep in threshold theorem}, and the last line from~\eqref{eq:assumption_W} and~\eqref{eq:assumption_circuit_size}\footnote{To achieve the constant space overhead, it is important to guarantee that the required width of the circuit should grow sufficiently fast compared to the required depth and the target error since the size of code blocks needs to grow for error suppression.}.
Thus, our fault-tolerant protocol achieves the constant space overhead.

On the other hand, from \eqref{eq: depth of the fault-tolerant circuit}, the time overhead of our protocol is bounded by
\begin{align}
&\frac{D_{\mathrm{FT}}(n,\varepsilon)}{D(n)}\nonumber\\
&\leq \qty(K(n,\varepsilon)+2)\times \qty(\frac{W(n)}{K(n,\varepsilon)}+1)\times\nonumber\\
&\qquad \frac{\tilde{W}^\prep(n,\varepsilon)}{W(n)}\times \qty(\tilde{D}^\prep(n,\varepsilon)+3d)\\
\label{eq: DFT in threshold theorem without order}
&=O\left(\tilde{W}^\prep(n,\varepsilon)\tilde{D}^\prep(n,\varepsilon)\times\right.\nonumber\\
&\quad\left.\qty(1+\frac{1}{K(n,\varepsilon)}+\frac{K(n,\varepsilon)}{W(n)}+\frac{1}{\tilde{D}^\prep(n,\varepsilon)})\right)\\
&=O\qty(\log^{2\alpha+\gamma_1+\gamma_2}\qty(\frac{|C_n^\org|}{\varepsilon})),
\label{eq: DFT in threshold theorem}
\end{align}
where we used~\eqref{eq:assumption_W},~\eqref{eq:assumption_circuit_size},~\eqref{eq: Ktheta in threshold theorem},~\eqref{eq: Wprep in threshold theorem}, and~\eqref{eq: Dprep in threshold theorem}.
Thus, our fault-tolerant protocol achieves the polylogarithmic time overhead.

\subsection{Partial circuit reduction}
\label{sec: parial circuit reduction}
Our error analysis of fault-tolerant simulations is carried out by propagating Pauli errors in each part of a faulty physical circuit toward the end of the part so that the faulty part is replaced with a non-faulty circuit. 
To perform this procedure consistently, it is crucial to define a local stochastic Pauli error model for the partially replaced fault-tolerant circuit so that its output probability distribution matches that of the entire fault-tolerant circuit. 
In this paper, we refer to the replacement procedure of a part of the fault-tolerant circuit while ensuring consistency with the rest as \textit{partial circuit reduction}.
The argument on this consistency was missing in the existing analyses of the threshold theorem of the constant-space-overhead protocol with quantum LDPC codes~\cite{Gottesman2014Constant,Fawzi_2018, grospellier:tel-03364419}, but with the techniques introduced here, we will complete the proof of the threshold theorem.

Before describing the partial circuit reduction, we clarify the problems that were not addressed in the existing analyses of the threshold theorem of the constant-space-overhead protocol with quantum LDPC codes~\cite{Gottesman2014Constant,Fawzi_2018, grospellier:tel-03364419}. 
In particular, in the following, we will explain two issues that motivate our development of circuit reducibility through partial circuit reduction.

The first problem arises from the assumption on decoding algorithms of the quantum LDPC codes.
The decoding algorithms are required to reduce and correct errors only for each single code block, as presented by Assumption~\ref{assump: non-vanishing-rate quantum LDPC with efficient decoding algorithm} in Sec.~\ref{sec: Conditions of fault-tolerant gadgets on quantum LDPC codes}. 
However, the assumption of FTQC is that the entire fault-tolerant circuit, rather than each single code block, experiences the local stochastic error model, as shown by Assumption~\ref{assumption: fault-tolerant circuit experiences the local stochastic Pauli error model} in Sec.~\ref{sec:setting}. 
To complete the proof of the threshold theorem, it is necessary to show that even if the decoding algorithm is applied locally to the code block, the set of error parameters in the local stochastic Pauli error model of the entire fault-tolerant circuit should be updated to another one so that the error rate of the part of the circuit in the code block remains consistent with the rest of the circuit.
To address this issue, we will show Lemmas~\ref{lemma: transitive rule} and \ref{lemma: extension rule} below, which make it possible to use the results on the decoding algorithms in~\cite{Fawzi_2018, grospellier:tel-03364419} as a black box in the proof of the threshold theorem. 
Specifically, these lemmas show that if a local stochastic Pauli error model defined on a part of the physical circuit is replaced with another local stochastic Pauli error model defined on this part with different error parameters.
These lemmas allow us to proceed with partial circuit reduction by considering and replacing a local stochastic Pauli error model defined for each rectangle of interest, without needing to analyze the entire fault-tolerant circuit at once.

The second challenge is that the decoding algorithm for the quantum LDPC codes may fail to reduce and correct errors with a nonzero probability~\cite{Fawzi_2018, grospellier:tel-03364419}.
In Refs.~\cite{Fawzi_2018, grospellier:tel-03364419}, it has been shown that in a low-error regime, errors are clustered in a correcatable way with a high probability, and then the decoding algorithm can successfully reduce and correct such errors; in this case, the error rate after the decoding is bounded with a theoretical guarantee.
However, with a nonzero probability, the errors may not be clustered in this way, and then, it is no longer possible to provide this guarantee. 
In the procedure of partial circuit reduction, we need to take into account the possibility of this failure, which complicates the replacement procedure of error parameters of the local stochastic Pauli error model in the partial circuit reduction.
To address this issue, we will consider partial circuit reduction with failure probability.

For our analysis, we define the conditions under which a physical circuit $C$ that follows a certain local stochastic Pauli error model can be transformed into another physical circuit $C'$ that follows a different local stochastic Pauli error model while maintaining the same output distribution.

First, we introduce the terminology for a part of the entire fault-tolerant circuit to be replaced in a partial circuit reduction.
\begin{definition}[Incomplete circuits]
\label{def: incomplete circuits}
An incomplete circuit is defined as a circuit that may not start with state preparations and may not end with measurements (but may accept quantum inputs and outputs, and also classical inputs and outputs representing measurement outcomes).
The set of all locations of an incomplete circuit $C$ is also denoted by $C$.
To define the equivalence between $C$ and $C'$,
we represent classical inputs and outputs for $C$ and $C'$ as diagonal quantum states.
This representation allows us to treat both circuits as implementing completely positive and trace-preserving (CPTP) maps.
We say that $C$ and $C'$ are equivalent if $C$ and $C'$ implement the same CPTP map.
\end{definition}

Next, we introduce the definition of circuit reducibility that allows us to apply the partial circuit reduction to an incomplete circuit, i.e., part of the entire fault-tolerant circuit.
For a set $C$, we write the power set of $C$ as $2^C$, which is the set of all possible subsets of $C$.
As explained above, the decoding algorithm for quantum LDPC codes may fail to guarantee the reduction to a local stochastic Pauli error model with a small probability.
To deal with this failure, we introduce a flag $m$, i.e., a binary-output function of the set of (faulty) locations indicating the case where the local stochastic Pauli error model cannot be defined;
in the case of $m=1$, the circuit may output an arbitrary probability distribution.

\begin{definition}[Circuit reducibility]
\label{def: delta-reducibility}
We say that a pair $\qty(C, \{p_i\}_{i \in C})$ of an incomplete circuit $C$ and a set $\{p_i\}_{i \in C}$ of error parameters for $C$ is $\delta$-reducible to another pair $\qty(C',\{q_j\}_{j \in C'})$ of an incomplete circuit $C'$ and a set $\{q_j\}_{j \in C'}$ of error parameters for $C'$ 
if there exist maps $\Gamma \colon 2^{C} \to 2^{C'}$ and $m\colon 2^C\rightarrow \{0,1\}$ satisfying the following conditions.
\begin{enumerate}
    \item \label{cond: delta-reducibility 1} For any probability distribution of $F\in 2^C$ satisfying
        \begin{align}
        \label{eq:circuit_reduction_2_1}
           \forall A\subseteq C,~~ \mathbb{P}[F\supseteq A]&\leq\prod_{i \in A} p_i,
        \end{align}
        it holds that
        \begin{align}
            \mathbb{P}[m(F)=0]\geq 1-\delta.
        \label{eq:circuit_reduction_2_3}
     \end{align}
     and
        \begin{align}
        \hspace{0.8cm}
        \label{eq:circuit_reduction_2_2}
           &\forall A'\subseteq C',\nonumber\\
            &\mathbb{P}[m(F)=0\text{~and~}\Gamma(F)\supseteq A']&\leq\prod_{j \in A'} q_j.
        \end{align}
    \item \label{cond: delta-reducibility 2} 
    For any set $F \subseteq C$ of faulty locations in $C$ with 
    $m(F)=0$,
     and any assignment of Pauli errors in $F$, there exists an assignment of Pauli errors in the set $\Gamma(F) \subseteq C'$ of faulty locations in $C'$ such that $C$ with the errors in $F$ and $C'$ with the errors in $\Gamma(F)$ are equivalent.
\end{enumerate}
\end{definition}

Note that the use of the single-shot decoding algorithm for reducing errors in Def.~\ref{def: single-shot decoding algorithm with thresholds to suppress residual error} for error suppression in EC gadgets implies the circuit reducibility, which will be used in the proof of the existence of a threshold in Sec.~\ref{sec: Proof of existence of a threshold}.
Here, the case of $m = 1$ corresponds to the failure of the decoding algorithm, and $\delta$ corresponds to the failure probability of the decoding algorithm in Def.~\ref{def: delta-reducibility}.
The circuit reducibility is also used for analyzing the failure probability of the state-preparation gadget in Def.~\ref{def: fault-tolerance conditions of the state preparation gadgets for quantum LDPC codes}.

Here, we introduce two lemmas, which are convenient for finding implications regarding circuit reducibility. 
\begin{lemma}[Transitive rule]
\label{lemma: transitive rule}
   Suppose that  $(C, \{p_i\}_{i \in C})$  is $\delta$-reducible to $(C',\{q_j\}_{j \in C'})$
   and 
     $(C',\{q_j\}_{j \in C'})$  is $\delta'$-reducible to $(C'',\{r_k\}_{k \in C''})$,
     where $\delta, \delta'\ge 0$ and $\delta+\delta'\le 1$.
     Then, 
      $(C, \{p_i\}_{i \in C})$  is $(\delta+\delta')$-reducible to $(C'',\{r_k\}_{k \in C''})$.
\end{lemma}
\begin{proof}
If $\delta+\delta'=1$, the lemma is trivially true. Assume
\begin{align}
\label{eq:assumption_delta}
    \delta+\delta'<1.
\end{align}

Suppose that a random variable $F\in 2^C$ satisfies
\begin{align}
\forall A\subseteq C,\quad \mathbb{P}[F\supseteq A]&\leq\prod_{i \in A} p_i.
\label{eq: trans1}
\end{align}
        By definition of circuit reducibility in Def.~\ref{def: delta-reducibility}, there exist maps $\Gamma_1: 2^C \to 2^{C'}$ and
        $m_1: 2^C \to \{0,1\}$ satisfying 
        \begin{align}
            \mathbb{P}[m_1(F)=0]&\geq 1-\delta
        \label{eq: trans2}
        \end{align}
        and 
        \begin{align}
            \forall A'\subseteq C',\quad 
            \mathbb{P}[m_1(F)=0\text{~and~}\Gamma_1(F)\supseteq A']&\leq\prod_{j \in A'} q_j.
        \end{align} 
        Under~\eqref{eq:assumption_delta}, i.e., under $1-\delta>0$, the condition~\eqref{eq:circuit_reduction_2_3} in the definition of circuit reducibility implies that, for a trivial probability distribution $\mathbb{P}[F=\varnothing]=1$ in~\eqref{eq:circuit_reduction_2_1}, it necessarily holds that
        \begin{align}
            0<\mathbb{P}[m_1(F)=0]=1.
        \end{align}
        Thus, under~\eqref{eq:assumption_delta}, the map $m_1$ must satisfy
        \begin{align}
            m_1(\varnothing)=0.
        \end{align}

        Define a random variable $F'\in 2^{C'}$ by
        \begin{align}
            F'\coloneqq\begin{cases}
                \Gamma_1(F)&\text{if $m_1(F)=0$},\\
                \varnothing&\text{if $m_1(F)=1$}.
            \end{cases}
        \end{align}
        Then, for any $A'\neq \varnothing$, we have
        \begin{align}
        \mathbb{P}[F'\supseteq A']
        =\mathbb{P}[m_1(F)=0\text{~and~}\Gamma_1(F)\supseteq A'].
        \end{align}
        On the other hand, it trivially holds that
        \begin{align}
         \mathbb{P}[F'\supseteq \varnothing]=\prod_{j \in \varnothing} q_j=1.   
        \end{align}
        Hence, we have
        \begin{align}
            \forall A'\subseteq C',~~ 
            \mathbb{P}[F'\supseteq A']&\leq\prod_{j \in A'} q_j.
        \end{align}
        Then, by definition of circuit reducibility in Def.~\ref{def: delta-reducibility}, there exist maps $\Gamma_2: 2^{C'} \to 2^{C''}$ and
        $m_2: 2^{C'} \to \{0,1\}$ satisfying 
          \begin{align}
            &\mathbb{P}[m_2(F')=0]\geq 1-\delta',\\
            &\forall A''\subseteq C'',\quad
            \mathbb{P}[m_2(F')=0\text{~and~}\Gamma_2(F')\supseteq A'']\leq\prod_{k \in A''} r_k.
        \end{align}

        We construct a map
        \begin{align}
            \Gamma\coloneqq \Gamma_2 \circ \Gamma_1,
        \end{align}
        and also define a map $m: 2^C \to \{0,1\}$
        by
        \begin{align}
            m(F)=\begin{cases}
                0&\text{if $m_1(F)=0$ and $m_2(\Gamma_1(F))=0$},\\
                1&\text{otherwise}.
            \end{cases}
        \end{align}
        Then, we see that, due to the union bound,
        \begin{align}
            &\mathbb{P}[m(F)=0] \nonumber\\
            &=\mathbb{P}[m_1(F)=0\text{~and~}m_2(\Gamma_1(F))=0] \nonumber\\
            &=\mathbb{P}[m_1(F)=0\text{~and~}m_2(F')=0]\\
            &\geq 1-\delta-\delta',
        \end{align}
        For any $A''\subseteq C''$, we also have
        \begin{align}
            &\mathbb{P}[m(F)=0\text{~and~}\Gamma(F)\supseteq A''] \nonumber\\
            &=\mathbb{P}[m_1(F)=0\text{~and~}m_2(\Gamma_1(F))=0\text{~and~}\nonumber\\
            &\qquad\Gamma_2 \circ \Gamma_1(F)\supseteq A''] \\
            &=\mathbb{P}[m_1(F)=0\text{~and~}m_2(F')=0\text{~and~}\Gamma_2(F')\supseteq A''] \\
            &\leq \mathbb{P}[m_2(F')=0\text{~and~}\Gamma_2(F')\supseteq A''].
        \end{align}
        We thus have 
        \begin{align}
            \forall A''\subseteq C'',~~ 
            \mathbb{P}[m(F)=0\text{~and~}\Gamma(F)\supseteq A'']&\leq\prod_{k \in A''} r_k,
        \end{align}
        which shows that $\Gamma$ and  $m$ satisfy condition 1 of Def.~\ref{def: delta-reducibility}. It is easy to confirm that they satisfy condition 2 of Def.~\ref{def: delta-reducibility}. 
\end{proof}
\begin{lemma}[Extension rule]
\label{lemma: extension rule}
   If $(C, \{p_i\}_{i \in C})$  is $\delta$-reducible to $(C',\{q_j\}_{j \in C'})$, then for any $(C_0,\{r_k\}_{k \in C_0})$,
   $(C\cup C_0, \{p_i\}_{i \in C}\cup \{r_k\}_{k \in C_0})$ is $\delta$-reducible to $(C'\cup C_0,\{q_j\}_{j \in C'}\cup \{r_k\}_{k \in C_0})$.
\end{lemma}
\begin{proof}
Since the case $\delta=1$ is trivially true, in the following, we assume
\begin{align}
\label{eq:assumption_delta_extension}
    \delta<1.
\end{align}
For convenience, we write
\begin{align}
    f(A)&\coloneqq\prod_{i \in A} p_i\quad\text{for $A\in C$},\\
    f'(A')&\coloneqq\prod_{j \in A'} q_j\quad\text{for $A'\in C'$},\\
    f_0(A_0)&\coloneqq\prod_{k\in A_0} r_k\quad\text{for $A_0\in C_0$}.
\end{align} Note that 
$f(\varnothing)=f'(\varnothing)=f_0(\varnothing)=1$.
By definition of circuit reducibility in Def.~\ref{def: delta-reducibility}, there exist maps $\Gamma: 2^C \to 2^{C'}$ and $m: 2^C \to \{0,1\}$ such that for any random variable $\tilde{F}$ satisfying
        \begin{align}
         \label{eq:extension11}
        \forall A\subseteq C,\quad \mathbb{P}[\tilde{F}\supseteq A]&\leq f(A),
        \end{align}
        it holds that
        \begin{align}
        \label{eq:m_F_delta_condition}
            \mathbb{P}[m(\tilde{F})=0]&\geq 1-\delta
        \end{align}
        and
        \begin{align}
         \label{eq:extension6}
            \forall A'\subseteq C',\quad 
            \mathbb{P}[m(\tilde{F})=0\text{~and~}\Gamma(\tilde{F})\supseteq A']&\leq f'(A').
        \end{align}
        Under~\eqref{eq:assumption_delta_extension}, i.e., under $1-\delta>0$,~\eqref{eq:m_F_delta_condition} implies that, for a trivial probability distribution $\mathbb{P}[\tilde{F}=\varnothing]=1$ in~\eqref{eq:extension11}, it necessarily holds that
        \begin{align}
            0<\mathbb{P}[m(\tilde{F})=0]=1.
        \end{align}
        Thus, under~\eqref{eq:assumption_delta_extension}, the map $m$ must satisfy
        \begin{align}
        \label{eq:tilde_m_0}
            m(\varnothing)=0,
        \end{align}
        and~\eqref{eq:extension6} implies that
\begin{align}
\label{eq:extension7}
     f'(A')=1  \text{~for~}  A'  \subseteq \Gamma(\varnothing) 
\end{align}

Suppose that random variables $F\in 2^C$ and $F_0\in 2^{C_0}$ satisfy
\begin{align}
 \label{eq:extension3}
 &\forall A\subseteq C, \forall A_0\subseteq C_0,\nonumber\\
 &\mathbb{P}[F\supseteq A\text{~and~}F_0\supseteq A_0]
 \leq f(A) f_0(A_0).
\end{align}
The case with $A_0=\varnothing$ implies that the random variable $F$ satisfies~\eqref{eq:extension11} and thus
 \begin{align}
 \label{eq:extension4}
\mathbb{P}[m(F)=0]
 \geq 1-\delta.
\end{align}
Therefore, if we can further show that 
\begin{align}
\label{eq:extension1}
&\mathbb{P}[m(F)=0\text{~and~}\Gamma(F)\supseteq A'\text{~and~}F_0\supseteq A_0]\nonumber\\
&\leq f'(A') f_0(A_0)
\end{align}
holds for all $A'\subseteq C'$ and $A_0\subseteq C_0$, then it will imply that condition 1 of Def.~\ref{def: delta-reducibility} is met.
Since condition 2 is obviously met, the lemma will be proved.
In the following, we will prove~\eqref{eq:extension1}.

Since $\mathbb{P}[F_0\supseteq A_0]=0$ trivially implies~\eqref{eq:extension1}, we henceforth assume that $\mathbb{P}[F_0\supseteq A_0]>0$. 
The case with $A=\varnothing$ in~\eqref{eq:extension3}, i.e.,
\begin{align}
 \label{eq:extension5}
 \mathbb{P}[F_0\supseteq A_0]
 \leq f_0(A_0),
\end{align}
 then implies $f_0(A_0)>0$ and allows us to define
 \begin{align}
           \lambda\coloneqq \frac{\mathbb{P}[F_0\supseteq A_0]}{f_0(A_0)},
 \end{align}
 which satisfies $0<\lambda \leq 1$. 
If $A' \subseteq \Gamma(\varnothing) $, then~\eqref{eq:extension7} and~\eqref{eq:extension5} trivially leads to~\eqref{eq:extension1} as
\begin{align}
&\mathbb{P}[m(F)=0\text{~and~}\Gamma(F)\supseteq A'\text{~and~}F_0\supseteq A_0]\nonumber\\
&\leq\mathbb{P}[F_0\supseteq A_0]\\
&\leq f_0(A_0)\\
&=f'(A') f_0(A_0);
\end{align}
thus, we henceforth assume that
\begin{align}
\label{eq:extension8}
A'  \nsubseteq \Gamma(\varnothing).
\end{align}
 
Let us introduce an auxiliary random variable $a\in \{0,1\}$ that is independent of $(F, F_0)$ and 
satisfies
\begin{align}
    \mathbb{P}[a=1]=\lambda.
\end{align}
Define a random variable $\tilde{F}\in 2^{C}$ by
\begin{align}
    \tilde{F}\coloneqq\begin{cases}
        F &\text{if  $a=1$},\\
        \varnothing&\text{if $a=0$}.
    \end{cases}
\end{align}
Due to~\eqref{eq:tilde_m_0}, we have
\begin{align}
\label{eq:tilde_M_F_tilde_F}
    \mathbb{P}[m(F)=0]\leq\mathbb{P}[m(\tilde{F})=0].
\end{align}
For any $A\subseteq C$ with $A\neq \varnothing$, we have
\begin{align}
&\mathbb{P}[\tilde{F}\supseteq A|F_0\supseteq A_0] \nonumber\\
&= \lambda\mathbb{P}[F\supseteq A|F_0\supseteq A_0] \\
&=\frac{\mathbb{P}[F\supseteq A\text{~and~}F_0\supseteq A_0]}{f_0(A_0)}. 
\end{align}
Using~\eqref{eq:extension3}, we have
\begin{align}
\mathbb{P}[\tilde{F}\supseteq A|F_0\supseteq A_0] \le f(A)
\end{align}
which is also true for $A= \varnothing$. This implies that conditioned on $F_0\supseteq A_0$, 
the random variable $\tilde{F}$ satisfies  (\ref{eq:extension11}) and thus
        \begin{align}
         \label{eq:extension9}
            &\forall A'\subseteq C',\nonumber\\
            &\mathbb{P}[m(\tilde{F})=0\text{~and~}\Gamma(\tilde{F})\supseteq A'|F_0\supseteq A_0]\leq f'(A').
        \end{align}
Combined with~\eqref{eq:extension8},~\eqref{eq:tilde_M_F_tilde_F}, and~\eqref{eq:extension9}, we have
\begin{align}
& \mathbb{P}[m(F)=0\text{~and~}\Gamma(F)\supseteq A'\text{~and~}F_0\supseteq A_0] \nonumber \\
&\leq
\mathbb{P}[m(\tilde{F})=0\text{~and~}\Gamma(F)\supseteq A'\text{~and~}F_0\supseteq A_0] \\
&=\lambda^{-1} \mathbb{P}[m(\tilde{F})=0\text{~and~}\Gamma(\tilde{F})\supseteq A'\text{~and~}F_0\supseteq A_0]
 \\
&= f_0(A_0) \mathbb{P}[m(\tilde{F})=0\text{~and~}\Gamma(\tilde{F})\supseteq A' | F_0\supseteq A_0] 
 \\
&\leq  f'(A') f_0(A_0),
\end{align}
which proves~\eqref{eq:extension1}.
\end{proof}

In addition, we introduce the following lemma, which provides more explicit conditions for the circuit reducibility in Def.~\ref{def: delta-reducibility}.
This is useful for confirming the reducibility by examining error propagation in the relevant circuit.

\begin{lemma}[Condition for reducibility]
\label{lemma: condition for reducibility}
     Consider a pair $\qty(C, \{p_i\}_{i \in C})$ for which $C$ is partitioned into two disjoint subsets $C_1$ and $C_2$, and
\begin{align}
    \forall i \in C_1, \quad p_i &\le p^{(1)}\in[0,1], \\
    \forall i \in C_2, \quad p_i &\le p^{(2)}\in(0,1]
\end{align} holds. The pair $\qty(C, \{p_i\}_{i \in C})$ is reducible (i.e., $\delta$-reducible with $\delta=0$) to another pair $\qty(C',\{q_j\}_{j \in C'})$ if there exists a map $\Gamma \colon 2^{C} \to 2^{C'}$ and positive integers $c_{\rm{fw}}$, $c_{\rm{bw}}^{(1)}$, and $c_{\rm{bw}}^{(2)}$ satisfying the following conditions.
    \begin{enumerate}
        \item \label{cond1 reducibility} It holds that
        \begin{align}
            \Gamma(A)&= \bigcup_{i\in A} \Gamma(\{i\}),\\
            \Gamma(\varnothing)&=\varnothing.
        \end{align}

        \item \label{cond2 reducibility}For all $i \in C$,
        \begin{align}
            |\Gamma(\{i\})|\le c_{\rm{fw}}.
        \end{align}

        \item \label{cond3 reducibility}For all $j \in C'$,
        \begin{align}
            |\{ i\in C_1 \colon  j\in\Gamma(\{i\})  \}| &\le c_{\rm{bw}}^{(1)},\\
            |\{ i\in C_2 \colon  j\in\Gamma(\{i\})  \}| &\le c_{\rm{bw}}^{(2)}.
        \end{align}
        
        \item \label{cond4 reducibility}Define
        \begin{align}
            C_\mathrm{faulty}^{\prime}\coloneqq\bigcup_{i\in C:p_i\neq 0}\Gamma(\{i\}).
        \end{align}
        For all $j \in C'_\mathrm{faulty}$,
        \begin{align}
        \label{eq:q_j_condition}
            q_j \ge  2^{c_{\bw}^{(2)}}\left(1 + \frac{p^{(1)}}{p^{(2)}}\right)^{c_\bw^{(1)}} \qty(p^{(2)})^{1/c_{\fw}}\nonumber\\+2^{c_{\bw}^{(1)}}\qty(p^{(1)})^{1/c_{\fw}};
        \end{align}
        and for all $j \in C'\setminus C'_\mathrm{faulty}$,
        \begin{align}
            q_j=0.
        \end{align}
        
        \item \label{cond5 reducibility}\label{cond: cond: equivalence of incomplete circuits in assumption}
        For any set $A \subseteq C$ of faulty locations in $C$ and any assignment of Pauli errors in $A$, there exists an assignment of Pauli errors in the set $\Gamma(A) \subseteq C'$ such that $C$ with the errors in $A$ and $C'$ with the errors in $\Gamma(A)$ are equivalent.
    \end{enumerate}
\end{lemma}

\begin{proof}
Define maps $\Lambda_1: 2^{C'} \to 2^{C_1}$ and $\Lambda_2: 2^{C'} \to 2^{C_2}$ by
\begin{align}
  &\Lambda_1(A')\coloneqq \{ i\in C_1 \colon  \Gamma(\{i\}) \cap A' \neq \varnothing \},\\
  &\Lambda_2(A')\coloneqq \{ i\in C_2 \colon  \Gamma(\{i\}) \cap A' \neq \varnothing \}.
\end{align}
Using the condition~\ref{cond3 reducibility}, we have 
\begin{align}
\label{eq:lambda_1_A_bound}
  &|\Lambda_1(A')|= \left| \bigcup_{j \in A'} \{ i\in C_1 \colon  \Gamma(\{i\}) \cap \{j\} \neq \varnothing \} \right|
   \le c_{\rm{bw}} ^{(1)}|A'|,\\
\label{eq:lambda_2_A_bound}
  &  |\Lambda_2(A')|= \left| \bigcup_{j \in A'} \{ i\in C_2 \colon  \Gamma(\{i\}) \cap \{j\} \neq \varnothing \} \right|
   \le c_{\rm{bw}}^{(2)} |A'|.
\end{align}
Suppose that a random variable $F\in 2^C$ satisfies
\begin{align}
    &\forall A_1\subseteq C_1\an A_2\subseteq C_2,\nonumber\\
    &\mathbb{P}[F\cap C_1\supseteq A_1\an F\cap C_2\supseteq A_2]\leq\qty(p^{(1)})^{|A_1|} \qty(p^{(2)})^{|A_2|},
    \label{eq:assumption_on_F}
\end{align}
where $0^0=1$ when $p^{(1)}=0$ and $|A_1|=0$.
Our goal is to show that 
\begin{align}
    \text{$\forall A'\subseteq C'$,\quad}\mathbb{P}[\Gamma(F)\supseteq A']\leq \prod_{j\in A'}q_j
\end{align}
for $|A'|\geq 1$.
In particular, if $A'\not\subseteq C_\mathrm{faulty}'$, then $\mathbb{P}[\Gamma(F)\supseteq A']=0$; thus, the goal reduces to showing
\begin{align}
\label{eq:goal_of_reducibility}
    \text{$\forall A'\subseteq C_\mathrm{faulty}'$,\quad}\mathbb{P}[\Gamma(F)\supseteq A']\leq \prod_{j\in A'}q_j
\end{align}
for $|A'|\geq 1$ and $q_j$ satisfying~\eqref{eq:q_j_condition}.
In the following analysis, we consider $A'\subseteq C_\mathrm{faulty}'$.

If $\Gamma(F)\supseteq A'$, defining $\Lambda(A') \coloneqq \Lambda_1(A') \cup \Lambda_2(A')$, we obtain from condition~\ref{cond1 reducibility}
\begin{align}
A'&=\Gamma(F)\cap A'\\
&=\bigcup_{i\in F } \Gamma(\{i\})  \cap A'\\
&= \bigcup_{i\in F\cap \Lambda(A') } \Gamma(\{i\})  \cap A'.
\end{align}
Thus, using condition~\ref{cond2 reducibility}, we have
\begin{align}
&|A'|\le \sum_{i\in F\cap \Lambda(A') } |\Gamma(\{i\})| \le c_{\fw} |F\cap \Lambda(A') |.
\end{align}
We write
\begin{align}
\label{eq:k_0_bound}
    k_0\coloneqq\left\lceil \frac{|A'|}{c_\fw}\right\rceil\geq\frac{|A'|}{c_\fw}.
\end{align}
Then, it holds that
\begin{align}
    \mathbb{P}[\Gamma(F) \supseteq A']
    &\leq 
\mathbb{P}\left[ c_{\fw}|F\cap \Lambda(A') | \ge |A'|\right]\\
&=
\mathbb{P}\left[ |F\cap \Lambda(A') | \ge k_0\right]
\end{align}
To bound this probability, we consider two cases based on the size of $|F\cap \Lambda_1(A')|$, i.e.,
\begin{align}
&\mathbb{P}\left[ |F\cap \Lambda(A') | \ge k_0\right]\nonumber\\
&\leq \mathbb{P}[|F\cap \Lambda_1(A')|\geq k_0]+\nonumber\\
&~~~~ \mathbb{P}[|F\cap \Lambda(A')|\geq k_0 \an |F\cap \Lambda_1(A')|\leq k_0].
\end{align}
In the following, we will bound each term on the right-hand side.

For the first term, it hold that 
\begin{align}
&\mathbb{P}[|F \cap \Lambda_1(A')| \ge k_0]\nonumber\\
&=\sum_{\substack{H\subseteq \Lambda_1(A') : \\  |H| \ge k_0}} 
 \mathbb{P}[F\cap \Lambda_1(A')=H]  \\
 &\le \sum_{\substack{H\subseteq \Lambda_1(A') : \\  |H| \ge k_0}} 
\qty(p^{(1)})^{k_0}\\  
&\le \sum_{H\subseteq \Lambda_1(A') } 
\qty(p^{(1)})^{k_0}  
\\
&= 2^{|\Lambda_1(A')|} \qty(p^{(1)})^{k_0},
\end{align}
where we used~\eqref{eq:assumption_on_F}. 
Due to~\eqref{eq:lambda_1_A_bound}, this becomes
\begin{align}
\mathbb{P}[|F \cap \Lambda_1(A')| \geq k_0] \leq \left(2^{c_\bw^{(1)}}\qty(p^{(1)})^{1/c_{\fw}}\right)^{|A'|},
\end{align}
where we used~\eqref{eq:k_0_bound} and $p^{(1)} \leq 1$.

For the second term, we have 
\begin{align}
&\mathbb{P}[|F \cap \Lambda(A')| \ge k_0  \text{~and~} |F \cap \Lambda_1(A')| \le k_0]\nonumber\\
&=\sum_{\substack{H\subseteq \Lambda(A') : \\  |H| \ge k_0 \\  |H \cap \Lambda_1(A')| \le k_0 }} 
 \mathbb{P}[F\cap \Lambda(A')=H].
\end{align}
We further decompose $H$ on the right-hand side into disjoint sets $H_1$ and $H_2$, i.e., $H_1 = H \cap \Lambda_1(A')$ and $H_2 = H \cap \Lambda_2(A')$, so $H = H_1 \cup H_2$. Then we have
\begin{align}
&\mathbb{P}[|F \cap \Lambda(A')| \ge k_0  \text{~and~} |F \cap \Lambda_1(A')| \le k_0]\nonumber\\
&=\sum_{\substack{H_1\subseteq \Lambda_1(A') : \\   |H_1| \le k_0 }} 
\sum_{\substack{H_2\subseteq \Lambda_2(A') : \\  |H_1| +  |H_2| \ge k_0  }} 
  \mathbb{P}[F\cap \Lambda(A')=H_1\cup H_2] \\
  &=\sum_{r=0}^{k_0} \sum_{\substack{H_1\subseteq \Lambda_1(A') : \\   |H_1| = r }} 
\sum_{\substack{H_2\subseteq \Lambda_2(A') : \\  |H_2| \ge k_0-r  }} 
  \mathbb{P}[F\cap \Lambda(A')=H_1\cup H_2] \\
   &\le \sum_{r=0}^{k_0} \sum_{\substack{H_1\subseteq \Lambda_1(A') : \\   |H_1| = r }} 
\sum_{\substack{H_2\subseteq \Lambda_2(A') : \\  |H_2| \ge k_0-r  }} 
   \qty(p^{(1)})^{r} \qty(p^{(2)})^{k_0-r}
   \label{eq: koas111301}
    \\
    &\le \sum_{r=0}^{k_0} \sum_{\substack{H_1\subseteq \Lambda_1(A') : \\   |H_1| = r }} 
\sum_{H_2\subseteq \Lambda_2(A')  }
   \qty(p^{(1)})^{r} \qty(p^{(2)})^{k_0-r} \\
   &= \sum_{r=0}^{k_0} \binom{|\Lambda_1(A')|}{r} 2^{|\Lambda_2(A')|}   \qty(p^{(1)})^{r} \qty(p^{(2)})^{k_0-r} \\
  &\le 2^{|\Lambda_2(A')|}   \sum_{r=0}^{|\Lambda_1(A')|} \binom{|\Lambda_1(A')|}{r} \qty(p^{(1)})^{r} \qty(p^{(2)})^{k_0-r} \\
  &= 2^{|\Lambda_2(A')|} \qty(p^{(2)})^{k_0-|\Lambda_1(A')|} \qty(p^{(1)}+ p^{(2)})^{|\Lambda_1(A')|} \\
  \label{eq:division_p_2}
  &= 2^{|\Lambda_2(A')|} \qty(p^{(2)})^{k_0} \left(1 + \frac{p^{(1)}}{p^{(2)}}\right)^{|\Lambda_1(A')|}.
\end{align}
where we used~\eqref{eq:assumption_on_F} for the inequality~\eqref{eq: koas111301}, and~\eqref{eq:division_p_2} follows from $p^{(2)}>0$.
Using the bounds~\eqref{eq:lambda_1_A_bound},~\eqref{eq:lambda_2_A_bound}, and~\eqref{eq:k_0_bound}, and $p^{(2)} \leq 1$, we have
\begin{align}
&\mathbb{P}[|F \cap \Lambda(A')| \ge k_0 \text{ and } |F \cap \Lambda_1(A')| \le k_0]\\
&\leq \left(2^{c_{\bw}^{(2)}}\left(1 + \frac{p^{(1)}}{p^{(2)}}\right)^{c_{\text{bw}}^{(1)}} \qty(p^{(2)})^{1/c_{\text{fw}}}\right)^{|A'|}.
\end{align}

Adding the bounds from the two terms, and noting that $x^n + y^n \leq (x + y)^n$ for $x, y \geq 0$ and $n \geq 1$, we obtain
\begin{align}
    &\mathbb{P}[\Gamma(F) \supseteq A'] \\
    &\leq \mathbb{P}[|F\cap \Lambda_1(A')|\geq k_0]+\nonumber\\
    &~~~~ \mathbb{P}[|F\cap \Lambda(A')|\geq k_0 \an |F\cap \Lambda_1(A')|\leq k_0]\\
    &\leq \left(2^{c_{\bw}^{(1)}}\qty(p^{(1)})^{1/c_{\fw}} + 2^{c_{\bw}^{(2)}}\left(1 + \frac{p^{(1)}}{p^{(2)}}\right)^{c_\bw^{(1)}} \qty(p^{(2)})^{1/c_{\fw}}\right)^{|A'|} \\
    & \le \prod_{j \in A'} q_j
\end{align}
from condition~\ref{cond4 reducibility}
with $A'\subseteq C_\mathrm{faulty}'$ and $|A'|\ge 1$, which proves~\eqref{eq:goal_of_reducibility}.

Combined with condition~\ref{cond5 reducibility}, all the requisites for the $0$-reducibility in Def.~\ref{def: delta-reducibility} are satisfied.
\end{proof}

These lemmas show that to analyze the error bounds after the partial circuit reduction, it suffices to consider the error parameters of the local stochastic Pauli error model defined on the part of the circuit to be replaced, rather than that defined on the entire circuit.
Our error analysis for the fault-tolerant simulation is conducted by propagating all Pauli errors that occur in each part of the circuit to the end of the part so that the fault-tolerant circuit is iteratively rewritten into another equivalent circuit where a new local stochastic Pauli error model is defined by replacing each part.
After completing this rewriting process for all parts of the fault-tolerant circuit, we will obtain a circuit with the same output distribution as the original circuit, as will be shown in the next section.

\subsection{Proof of existence of threshold}
\label{sec: Proof of existence of a threshold}

In this section, we prove the existence of a positive threshold $q_{\mathrm{loc}}^{\mathrm{th}}$ such that if $p_{\mathrm{loc}}$ is below this threshold, the propagation of errors during the execution of rectangles is sufficiently suppressed. 
To this end, we establish a circuit reducibility rule for each rectangle involving a gadget for an elementary operation and EC gadgets, such that it is reduced to a non-faulty operation followed by a constant-depth circuit of faulty wait operations that are responsible for all the Pauli errors in the rectangle.
Based on these rules of circuit reducibility of rectangles, as shown in the diagrams and lemmas below, we can perform a partial circuit reduction of each rectangle, from the start to the end of a fault-tolerant circuit in an inductive manner, and find an upper bound on the failure probability of a fault-tolerant simulation.

We first introduce constants and functions that we use in the subsequent proofs. 
We define the threshold $q_{\mathrm{loc}}^{\mathrm{th}}$ and related constants explicitly as follows.
We define a constant $q^\thre_\prep>0$ as a unique solution of the equation
\begin{align}
\label{eq:q_prep_th}
  2^{(d_0+1) 2^{d_0}}\qty((2d+M)q_{\mathrm{prep}}^{\mathrm{th}})^{1/2^{d_0}} = p^\mathrm{dec}_\mathrm{th}.
\end{align}
Also, due to~\eqref{eq: lambda in threshold theorem}, the equations
\begin{align} 
&2^{2^{2d_0}}\qty(1+\frac{d q_\gat^\thre}{\lambda})^{2^{2d_0+1}(d_0+1)}{\lambda}^{1/2^{2d_0}}\nonumber \\
        &~~~~~~~~~+2^{2^{2d_0+1}(d_0+1)}(dq_\gat^\thre)^{1/2^{2d_0}}=p_\dec^\thre,\label{eq:q_gate_th}\\
&2^{2^{d_0}}\qty(1+\frac{d q_\meas^\thre}{\lambda})^{2^{d_0+1}(d_0+1)}{\lambda}^{1/2^{d_0}}\nonumber \\
        &~~~~~~~~~+2^{2^{d_0+1}(d_0+1)}(dq_\meas^\thre)^{1/2^{d_0}}=p_\mathrm{dec}^{\mathrm{th}}\label{eq:q_meas_th}
\end{align}
have unique solutions, respectively, which define  $q_\gat^\thre>0$ and $q_\meas^\thre>0$.
Furthermore, let
\begin{align}
    p_\loc^\thre>0
\end{align}
denote the threshold constant of the protocol for an open quantum circuit, as shown in Appendix~\ref{Theorem: level-reduction for the circuit that outputs a quantum state} of Appendix~\ref{appendix: fault-tolerant protocol for open quantum circuits}.
We then define the threshold of our protocol as 
\begin{equation}
\label{eq: q_th in threshold theorem}
    q_{\mathrm{loc}}^{\mathrm{th}}\coloneqq\min\qty{q_{\mathrm{prep}}^{\mathrm{th}}, q_{\mathrm{gate}}^{\mathrm{th}},q_{\mathrm{meas}}^{\mathrm{th}}, p_\loc^\thre}.
\end{equation}

To use a $\delta$-reducible relation for each rectangle of the circuit, we also define $\delta$ as a function of 
$(n, \varepsilon)$ as
\begin{equation}
\label{eq: delta in threshold theorem}
    \delta(n, \varepsilon)
    \coloneqq
    \qty(\frac{\varepsilon}{|C_n^\org|})^2.
\end{equation}
We will use several properties for $\delta(n, \varepsilon)$ in the subsequent proofs. 
Let $c_D$ and $c_W$ be the constants given by Lemma~\ref{lemma: fault-tolerance condition}.
Recall that we have chosen a family of parameters $[[N(n,\varepsilon),K(n,\varepsilon), D(n,\varepsilon)]]$, $L(n,\varepsilon)$, and $T$ as in~\eqref{eq: code parameters in threshold theorem},~\eqref{eq: L in threshold theorem}, and~\eqref{eq: c(T) in threshold theorem}.
From Defs.~\ref{def: single-shot decoding algorithm with thresholds to suppress residual error} and~\ref{def: Logarithmic-time decoding algorithm with threshold} with~\eqref{eq: Dtheta in threshold theorem},  
we see that, for sufficiently large $n$, there exists a positive constant $c_\mathrm{dec}$
such that $D(n, \varepsilon)$ satisfies
\begin{align}
    &\delta_1(D)\leq \exp(-c_\mathrm{dec}D), \\
    &\delta_2(D)\leq \exp(-c_\mathrm{dec}D).
\end{align}
Combined with \eqref{eq: constant relation 1 in threshold theorem}, \eqref{eq: Dtheta in threshold theorem} and (\ref{eq: delta in threshold theorem}), we see that for sufficiently large $n$, 
the quantities $\delta(n, \varepsilon)$ and $D(n, \varepsilon)$ 
satisfy 
\begin{align} 
\label{eq: delta1}
 \delta_1(D) \le \frac{\delta}{c_W+1}, \\
 \label{eq: delta2}
 \delta_2(D) \le \frac{\delta}{c_W+1}.
\end{align}
From \eqref{eq: Ntheta in threshold theorem}, \eqref{eq: L in threshold theorem}, and (\ref{eq: delta in threshold theorem}), we also see that for sufficiently large $n$, 
the quantities $\delta(n, \varepsilon)$, $N(n, \varepsilon)$, and $L(n, \varepsilon)$ 
satisfy 
\begin{align} 
\label{eq: delta3}
\delta_3(N, L)=c_3 c_D c_W N^2 {p}^{\mathrm{th}}_{\mathrm{loc}}\left(\frac{p_{\mathrm{loc}}}{{p}^{\mathrm{th}}_{\mathrm{loc}}}\right)^{2^L}
\le \frac{\delta}{c_W+1},
\end{align}
where $c_3$ is given by Theorem~\ref{Theorem: level-reduction for the circuit that outputs a quantum state} in Appendix~\ref{appendix: fault-tolerant protocol for open quantum circuits}.

For convenience, we present the circuit reducibility of rectangles with diagrams as follows.
We use thin and thick lines to distinguish unencoded qubits and qubits that form a code block. 
The two types of lines are often connected by an ideal encoder that maps an input state of a $K$-qubit register to the corresponding logical state encoded on $K$ logical qubits of a code block of $\mathcal{Q}$, which is depicted by
    \begin{align}
    \includegraphics[width=0.3\textwidth]{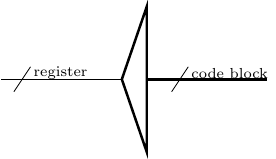}\label{fig: ideal encoder}
    \end{align}
An operation on unencoded qubits is represented by a box surrounded by thin lines and is always assumed to be ideal with no faults (the error parameters for all the locations are zero).
    On the other hand, a gadget corresponding to an elementary operation is represented by a box surrounded by thick lines and with a number underneath, which indicates that the error parameters for all the locations are the same and bounded by this number.
    A sequence of boxes shaded in gray represents a rectangle, while white boxes represent wait operations to represent Pauli errors.
    Note that, for simplicity, multiple-depth wait operations (such as those that appear in gadget synchronization) are represented by single-depth wait operations with increased error parameters.
    This is valid since $s$-depth wait operations with error parameter $p_\loc$ are obviously reducible to a single depth wait with error parameter $dp_\loc$ when $s\le d$.
    A single depth of faulty wait operations to represent Pauli errors in a rectangle is shown with error parameters bounded by $p$ as follows.
    \begin{align}
    \includegraphics[width=0.3\textwidth]{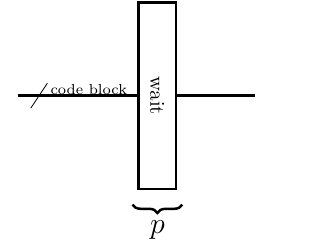}\label{fig: wait operation in threshold theorem}.
    \end{align}
    In this way, a diagram composed of lines, boxes, and encoders represents a pair $(C, \{p_i\}_{i\in C})$ of a circuit and the set of error parameters.
    
    We now provide all the reducibility rules used in the proof as diagrams.
    Each statement of circuit reducibility is represented by two such depictions of circuit-parameter pairs connected by an arrow with $\delta$ on top of it, as shown in~\eqref{eq: 0-state-preparation in threshold thorem} to~\eqref{eq: Bell measurement in threshold thorem} below.
    The proofs and the requisites and relations of the parameters will follow immediately.

\begin{lemma}[Circuit reducibility of rectangles]
\label{lemma: circuit-reducibility of rectangles}
Suppose that the parameter $\lambda$ is chosen to satisfy \eqref{eq: lambda in threshold theorem}.
Consider gadgets satisfying Lemma~\ref{lemma: fault-tolerance condition} and Defs.~\ref{def: fault-tolerance condition of gate, measurement for qLDPC codes}--\ref{def: fault-tolerance conditions of the EC gadgets for quantum LDPC codes} with parameters $N$, $D$, $L$, $T$, $\delta_1$, $\delta_2$, and $\delta_3$ satisfying \eqref{eq: c(T) in threshold theorem}, \eqref {eq: delta1}, \eqref {eq: delta2}, and \eqref {eq: delta3}.
If $p_\loc\in \qty(0, q_{\mathrm{loc}}^{\mathrm{th}})$ for $q_{\mathrm{loc}}^{\mathrm{th}}$ in~\eqref{eq: q_th in threshold theorem}, the reducibility relations shown by the 
diagrams~\eqref{eq: 0-state-preparation in threshold thorem}--\eqref{eq: Bell measurement in threshold thorem} hold.
\end{lemma}
\begin{widetext}
    \begin{enumerate}
        \item $\delta$-reduciblities of state-preparation rectangles
        \begin{itemize}
        \item  $\ket{0}^{\otimes K}$-state preparation rectangle\\
        \begin{align}
        \includegraphics[width=.9\textwidth]{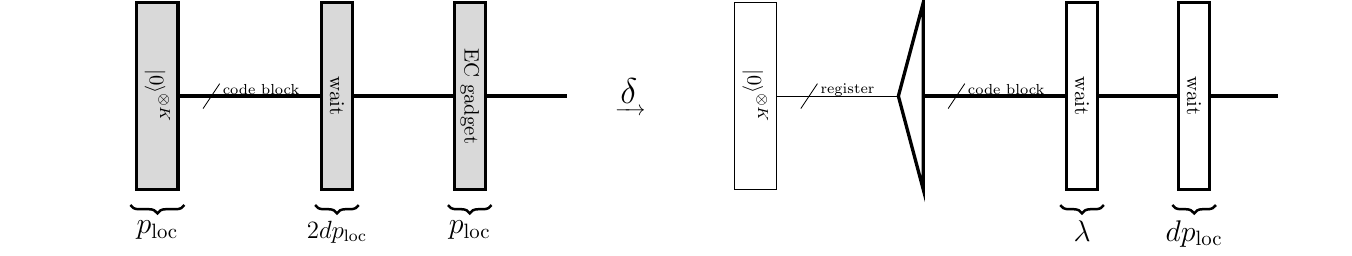}\label{eq: 0-state-preparation in threshold thorem},
        \end{align}
        \item  Clifford-state preparation rectangles\\
        \begin{align}
        \includegraphics[width=.9\textwidth]{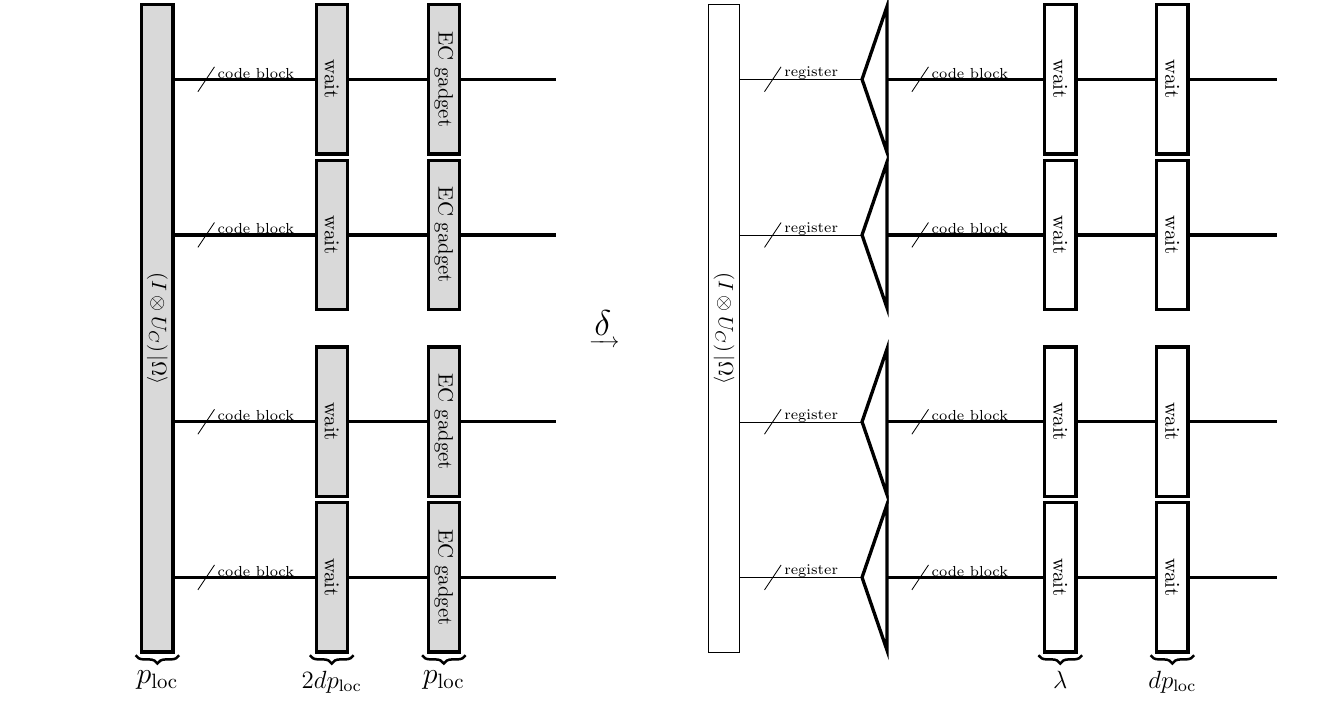}\label{eq: clifford-state-preparation in threshold thorem},
        \end{align}
        \item  magic-state preparation rectangles\\
        \begin{align}
        \includegraphics[width=.9\textwidth]{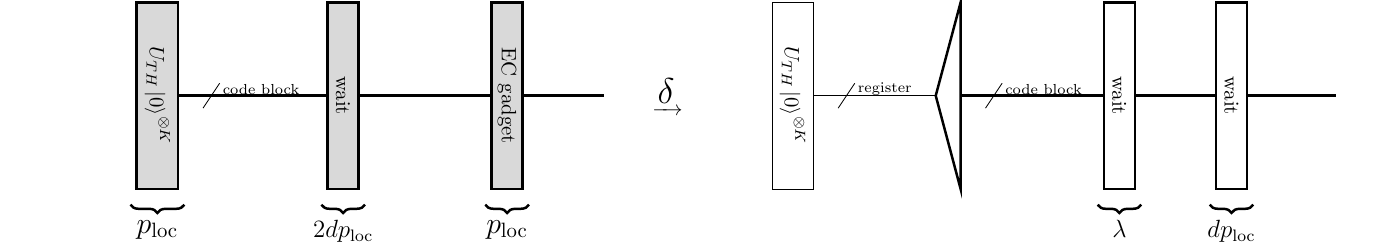}\label{eq: magic-state-preparation in threshold thorem},
        \end{align}
        \end{itemize}

        \item $\delta$-reduciblities of gate rectangles
        \begin{itemize}
        \item CNOT-gate rectangle\\
        \begin{align}
        \includegraphics[width=.85\textwidth]{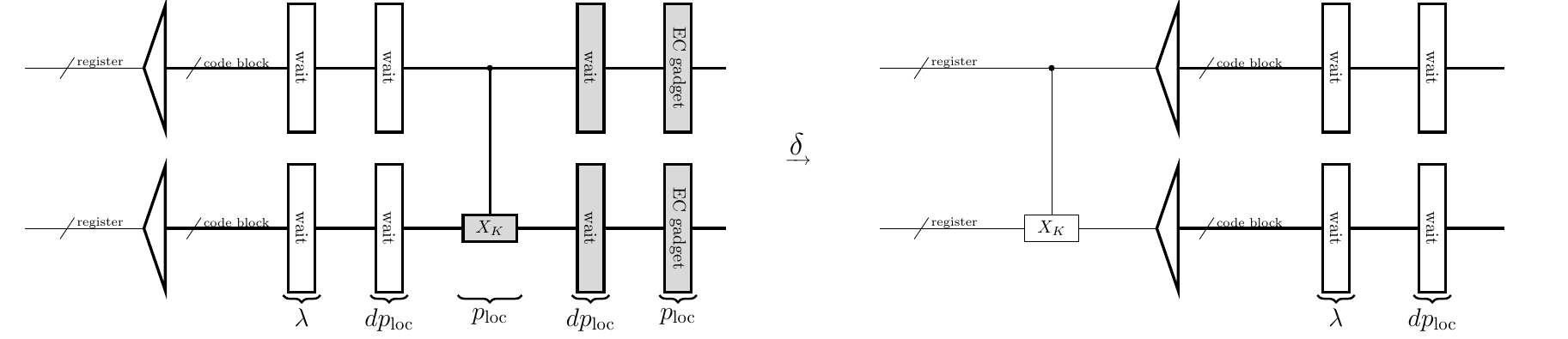}\label{eq: cnot-gate in threshold thorem},
        \end{align}
                \item Pauli-gate rectangles\\
        \begin{align}
        \includegraphics[width=.85\textwidth]{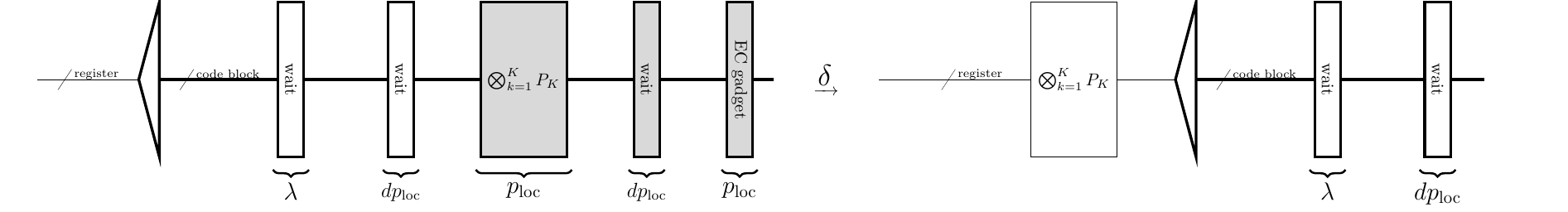}\label{eq: Pauli gate in threshold thorem},
        \end{align}
        \item wait rectangle\\
        \begin{align}
        \includegraphics[width=.85\textwidth]{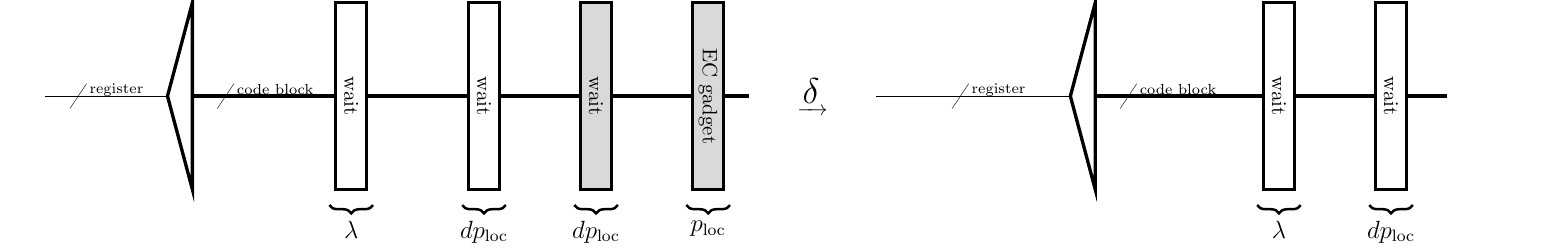}\label{eq: wait in threshold thorem},
        \end{align}        
        \end{itemize}

        \item  $\delta$-reduciblities of measurement rectangles
        \begin{itemize}
        \item $Z_K$-measurement rectangle\\ 
        \begin{align}
        \includegraphics[width=.8\textwidth]{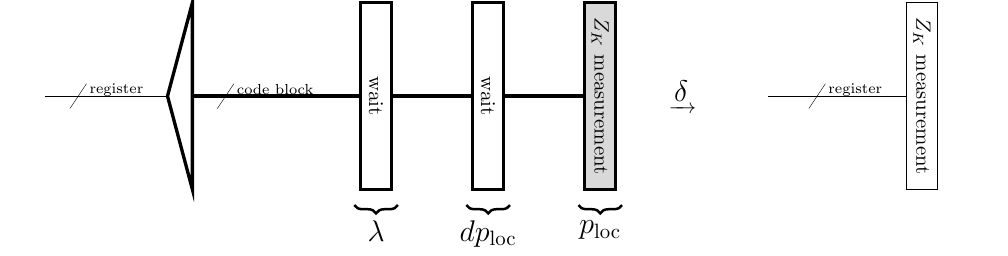}\label{eq: measurement in threshold thorem},
        \end{align}

        \item Bell-measurement rectangle\\ 
        \begin{align}
        \includegraphics[width=.8\textwidth]{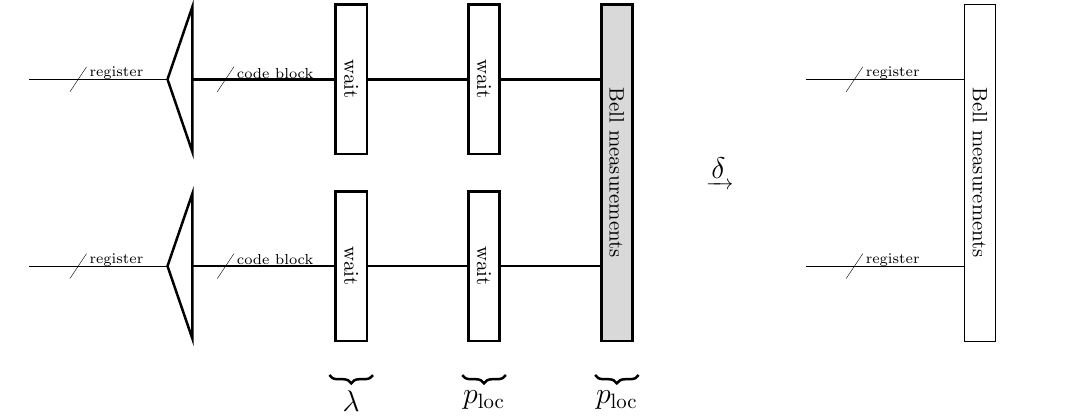}\label{eq: Bell measurement in threshold thorem}.
        \end{align}
    \end{itemize}
    \end{enumerate}
    \end{widetext}

First, we show that the circuit reducibilities of the state-preparation rectangles, i.e., the $\ket{0}^{\otimes K}$-state preparation rectangle in \eqref{eq: 0-state-preparation in threshold thorem}, Clifford-state preparation rectangles in \eqref{eq: clifford-state-preparation in threshold thorem}, and magic-state preparation rectangles in \eqref{eq: magic-state-preparation in threshold thorem}.

\begin{proof}[Proof of~\eqref{eq: 0-state-preparation in threshold thorem},~\eqref{eq: clifford-state-preparation in threshold thorem}, and~\eqref{eq: magic-state-preparation in threshold thorem}]
    Let $C^{\EC}$ denote an EC gadget with depth $d$.
    The EC gadget is an incomplete circuit composed of the syndrome measurement circuit $C^\synd$, followed by a circuit of wait locations during classical computation, and finally a Pauli recovery operation deduced by the decoding algorithm, i.e.,
     \begin{align}
        C^{\mathrm{EC}}\coloneqq C^{\mathrm{synd}}\cup C^{\mathrm{wait}}\cup C^{\mathrm{rec}},
    \end{align}
    where the depth of ${C}^{\mathrm{synd}}$ is bounded by a constant $d_0$ from Lemma~\ref{lemma: fault-tolerance condition}.
    Let $C$ denote an incomplete circuit composed of a state-preparation gadget $\tilde{C}^{\mathrm{prep}}$, followed by a single-depth circuit of wait operations $C^\wait_\sync$, an EC gadget $C^{\mathrm{EC}}$, i.e.,
    \begin{align}
        C\coloneqq \tilde{C}^{\mathrm{prep}}\cup C^{\mathrm{wait}}_\sync\cup C^{\mathrm{EC}},
    \end{align}
    where ${C}^{\mathrm{wait}}_\sync$ is inserted to synchronize the gadgets during the compilation process.
    Here $C$ has a set of error parameters such that
    \begin{align}
    p_i&=p_\loc&&\text{if $j \in C^\mathrm{prep}$},&&&\\
    p_i&\leq 2dp_\mathrm{loc}&&\text{if $j \in C^{\mathrm{wait}}_\sync$},&&&\\
    p_i&\leq p_\mathrm{loc}&&\text{if $j \in {C}^{\mathrm{EC}}$}.&&&
    \end{align}
    Let $C'$ denote an incomplete circuit composed of the original open circuit ${C}^{\mathrm{prep}}$ of $\tilde{C}^{\mathrm{prep}}$, followed by a single-depth circuit $C^\wait_\re$ of wait operations, and a single-depth circuit $C^\wait$ of wait operations, i.e.,
    \begin{align}
         C'\coloneqq {C}^{\mathrm{prep}}\cup {C}^{\mathrm{wait}}_\re \cup {C}^{\mathrm{wait}}.
    \end{align}
    Here, $C'$ has a set of error parameters such that 
    \begin{align}
    q_j&= 0&&\text{if $j \in {C}^{\mathrm{prep}}$},\label{eq:reduced error parameter q_j 2}\\
    q_j&\leq \lambda&&\text{if $j \in {C}^{\mathrm{wait}}_\re$},\label{eq:reduced error parameter q_j 1}\\
    q_j&\leq d p_{\mathrm{loc}}&&\text{if $j \in {C}^{\mathrm{wait}}$}.\label{eq:reduced error parameter q_j 1}
    \end{align}
    As shown in the diagrams, we will show the circuit reducibility from $\{C,\{p_i\}_{i\in C}\}$ to $\{C',\{q_j\}_{j\in C'}\}$.
    To this end, we consider circuit reducibility for $\tilde{C}^{\mathrm{prep}}$, $C^{\mathrm{synd}}$, and the decoding algorithm, individually.
    
    Using Theorem~\ref{Theorem: level-reduction for the circuit that outputs a quantum state} in Appendix~\ref{appendix: fault-tolerant protocol for open quantum circuits}, under the condition that
    \begin{align}
    \label{eq:threshold_1}
        p_{\mathrm{loc}} < p_{\mathrm{loc}}^{\mathrm{th}},
    \end{align}
    a pair $\qty(\tilde{C}^{\mathrm{prep}}, \{p_i\}_{i\in \tilde{C}^{\mathrm{prep}}})$, where the set of error parameters $\{p_i\}_i$ are given by
    \begin{align}
        p_i&\leq p_\mathrm{loc}&&\text{if $i \in  \tilde{C}^{\mathrm{prep}}$},
    \end{align}
    is $\delta_3$-reducible to $\qty(C^{\mathrm{prep}} \cup {C''}^{\mathrm{wait}}, \{p'_j\}_{j\in C^{\mathrm{prep}} \cup {C''}^{\mathrm{wait}}})$, where ${C''}^{\mathrm{wait}}$ is a single-depth circuit of wait locations, and $\{p'_j\}$ is the set of error parameters given by
    \begin{align}
    p'_j&=0&&\text{if $j \in C^\mathrm{prep}$},&&&\\
    p'_j&\leq Mp_\mathrm{loc}&&\text{if $j \in {C''}^{\mathrm{wait}}$}.&&&
    \end{align}
    Here, $M\geq 1$ is a constant.
    Since ${C''}^{\wait}\cup C_\sync^\wait$ consists only of wait operations, it is obviously reducible to ${C''}^{\wait}\cup C_\sync^\wait$ with the error parameter at each wait location in $C_\sync^\wait$ bounded by $(2d+M)p_\loc$, while the error parameter of each wait location in ${C''}^{\wait}$ is zero.
    Thus, from Lemmas~\ref{lemma: transitive rule} and \ref{lemma: extension rule}, $\{C,\{p_i\}_{i\in C}\}$ is $\delta_3$-reducible to an incomplete circuit
    \begin{equation}
        C''\coloneqq C^{\mathrm{prep}}\cup {C''}^{\mathrm{wait}}\cup {C}^{\mathrm{wait}}_\sync\cup C^{\mathrm{EC}},
    \end{equation}
    and the set of error parameters $\{p'_j\}_{j\in C''}$ given by
    \begin{align}
    p'_j&=0&&\text{if $j \in C^\mathrm{prep}\cup{C''}^{\mathrm{wait}}$},&&&\\
    p'_j&\leq (2d+M)p_\mathrm{loc}&&\text{if $j \in {C}^{\mathrm{wait}}_\sync$},&&&\\
    p'_j&\leq p_\mathrm{loc}&&\text{if $j \in {C}^{\mathrm{EC}}$},&&&
    \end{align}
    where
    
    Next, let us consider circuit reduction of an incomplete circuit of
    \begin{equation}
        {C'}^{\mathrm{synd}}\coloneqq  {C}^{\mathrm{wait}}_\sync\cup{C''}^{\mathrm{wait}}\cup C^{\mathrm{synd}},
    \end{equation}
    and the set of error parameters $\{p'_i\}_{i\in {C'}^{\synd}}$, where
    \begin{align}
    p'_j&=0 &&\text{if $i \in  {C}^{\mathrm{wait}}_\sync$}\\
    p'_i&\leq (2d+M)p_\mathrm{loc}&&\text{if $i \in{C''}^{\mathrm{wait}}\cup C^{\mathrm{synd}}$}.
    \end{align}
    Here, the error parameter at each location of $C^{\synd}$ is upper bounded by $p_\mathrm{loc}$, which is also upper bounded by $(2d+M)p_\mathrm{loc}$.

    To apply Lemma~\ref{lemma: condition for reducibility} for the circuit reduction of $C^{\prime \synd}$, we consider a partition of the set $C'^\synd$ into disjoint sets $C_1$ and $C_2$.
    Here, $C_1$ corresponds to ${C}^{\mathrm{wait}}_\sync$, and $C_2$ corresponds to ${C''}^{\mathrm{wait}}\cup C^{\mathrm{synd}}$.
    We consider error propagation from a given set of faulty locations in ${C'}^{\synd}$ to the set of faulty locations by propagating Pauli errors until the time step after the $Z$-basis measurements or the final time step of the wait locations in ${C'}^{\mathrm{synd}}$.
    Let us denote the non-faulty locations in the middle of ${C}^{\prime\synd}$ after this error propagation by ${C}^{\prime\synd}_{\mathrm{nonfaulty}}$, and the faulty locations by ${C}^{\prime\synd}_{\mathrm{nonfaulty}}$, i.e.,
    \begin{align}
        {{C}^{\prime\synd}}={C}^{\prime\synd}_\mathrm{nonfaulty}\cup {C}^{\prime\synd}_\mathrm{faulty}.
    \end{align}
    This propagation defines the map $\Gamma^{\synd}: 2^{ {C'}^{\synd}} \rightarrow 2^{{C}^{\prime\synd}_\mathrm{faulty}}$. 
    A single-qubit Pauli error in the middle of a physical circuit may propagate, after a single-depth part of the physical circuit, to at most two-qubit Pauli errors as operations in the physical circuit act on at most two qubits.
    A Pauli error occurring at any location in ${C}^{\prime\mathrm{synd}}$ may propagate to at most $2^{d_0}$ locations. 
    In contrast, a Pauli error in ${C''}^\wait$ does not propagate; thus, we have 
    \begin{align}
        |\Gamma^\synd(\{i\})|\leq c_\fw =2^{d_0} \text{~for all~} i\in {C'}^{\mathrm{synd}}.
    \end{align}
    Moreover, in the error propagation, each location in ${C}_\mathrm{faulty}^{\mathrm{synd}}$ may be affected by errors occurring in at most $(d_0+1)\times 2^{d_0}$ possible locations in ${C'}^{\mathrm{synd}}$.
    Thus we have for all $j \in {C}^{\prime\synd}_\mathrm{faulty}$
    \begin{align}
        &|\{ i\in C_\sync^\wait\colon  j\in\Gamma^\synd(\{i\})  \}| \le c_\bw^{(1)}=2^{d_0},&&\\
        &|\{ i\in {C''}^{\mathrm{wait}}\cup C^{\mathrm{synd}} \colon  j\in\Gamma^\synd(\{i\})  \}|\nonumber \\
        & \le c_\bw^{(2)}=(d_0+1) 2^{d_0}.&&
    \end{align}
Furthermore, under the partition, we have $p_i=p^{(1)}=0$ for all $i\in C^\wait_\sync$ and $p_i\leq p^{(2)} = (2d+M)p_\mathrm{loc}$ for all $i\in C''^\wait\cup C^\synd$.
    As a result, if we choose $p''_j$ for all $j\in {C}^{\prime\synd}_{\mathrm{faulty}}$ as
    \begin{align}
        p''_j &= 2^{c_{\bw}^{(2)}}\left(1 + \frac{p^{(1)}}{p^{(2)}}\right)^{c_\bw^{(1)}} \qty(p^{(2)})^{1/c_{\fw}}+2^{c_{\bw}^{(1)}}\qty(p^{(1)})^{1/c_{\fw}}\\
        &= 2^{(d_0+1) 2^{d_0}}\qty((2d+M)p_\mathrm{loc})^{1/2^{d_0}},
    \end{align}
    the map $\Gamma^{\synd}$ and the sets of error parameters $\{p'_i\}_{i\in {C'}^{\synd}}$ and $\{p''_j\}_{j\in {C}^{\prime\synd}_{\mathrm{faulty}}}$ satisfy all conditions in Lemma~\ref{lemma: condition for reducibility}.
    Thus, $\qty({C'}^{\synd}, \{p'_i\}_{i\in {C'}^{\synd}})$ is reducible to $\qty({C'}^{\synd}_\mathrm{faulty}, \{p''_j\}_{j\in {C'}^{\synd}_\mathrm{faulty}})$.
    
    Finally, we consider circuit reducibility of $C^{\prime\synd}_{\mathrm{faulty}}$ through the decoding algorithm.
    Let $V$ denote the set of physical qubits of a single code block in the EC gadget, and $W_X$ and $W_Z$ be the sets of syndrome bits for $X$-type and $Z$-type stabilizer generators, respectively, as in Def.~\ref{def: Local stochastic Pauli error model on physical qubits and syndrome bits}.
    The error parameters of the local stochastic Pauli error model on physical qubits and syndrome bits in Def.~\ref{def: Local stochastic Pauli error model on physical qubits and syndrome bits} are bounded by 
    \begin{align}
    p_{\mathrm{phys}}&\leq   2^{(d_0+1) 2^{d_0}}\qty((2d+M)p_\mathrm{loc})^{1/2^{d_0}},\\
    p_{\mathrm{synd}}&\leq 2^{(d_0+1) 2^{d_0}}\qty((2d+M)p_\mathrm{loc})^{1/2^{d_0}}.
    \end{align}
    In this case, if we have
    \begin{align}
    \label{eq: threshold ineq in state preparation}
    p_\mathrm{loc}<\frac{1}{2d+M}\qty(\frac{p_\mathrm{dec}^\mathrm{th}}{  2^{(d_0+1) 2^{d_0}}})^{2^{d_0}}= q^\thre_\prep
    \end{align}
    for $q^\thre_\prep$ in~\eqref{eq:q_prep_th}, then we obtain
    \begin{align} \label{eq:threshold_condition_1} 
      2^{(d_0+1) 2^{d_0}}\qty(  (2d+M)p_\mathrm{loc})^{1/2^{d_0}} < p^\mathrm{dec}_\mathrm{th},
    \end{align}
    so we can use the single-shot decoding algorithm in Def.~\ref{def: single-shot decoding algorithm with thresholds to suppress residual error} with a threshold of $p^{\thre}_{\mathrm{dec}}>0$.
    In our protocol, we choose $T$ in \eqref{eq: c(T) in threshold theorem} to satisfy
    \begin{align}
        c(T)=\frac{\log \lambda}{\log p_\mathrm{dec}^\mathrm{th}},
    \end{align}
    which guarantees that the error parameter on $V$ after a recovery operation is applied as
    \begin{align}
        p_{\mathrm{phys}}\leq \qty(p_\mathrm{dec}^\mathrm{th}) ^{c(T)}\leq \lambda.
    \end{align}
    
    If we associate $V$ with the set of wait locations and $W_X$ and $W_Z$ with $Z_K$-measurement locations, the map $\Gamma^{\mathrm{dec}}$ induced by the decoding algorithm in Def.~\ref{def: single-shot decoding algorithm with thresholds to suppress residual error} satisfies the conditions in Def.~\ref{def: delta-reducibility}.
    Thus, if $p_{\loc}<q_\mathrm{prep}^\mathrm{th}$, the incomplete circuit corresponding to each code block in $C^{\prime\synd}_{\mathrm{faulty}}$, where the error parameter of each location is bounded by $2^{(d_0+1) 2^{d_0}}\qty((2d+M)p_\mathrm{loc})^{1/2^{d_0}}$, is $\delta_1$-reducible to $C^\wait_\re$ with further reduced error parameters $q_j$ at each wait location $j$ is bounded by $\lambda$.
    Note that this argument holds even if the recovery operation is applied after $C^\mathrm{wait}$ in $C^{\mathrm{EC}}$.
    This is because the Pauli operator in $C^{\mathrm{rec}}$ can be back-propagated to $C^{\prime\synd}_{\mathrm{faulty}}$ through $C^{\mathrm{wait}}$, allowing us to rewrite the EC gadget as if $C^{\mathrm{rec}}$ were applied immediately after $C^{\prime\synd}_{\mathrm{faulty}}$.

    Since, by Lemma~\ref{lemma: fault-tolerance condition}, the width of the original circuit that the state-preparation gadget aims to simulate is bounded by $c_WN$ (with a constant $c_W$), the decoding algorithm is performed for at most $c_W$ code blocks.
    Therefore, if $p_{\loc}<q_\mathrm{prep}^\mathrm{th}$, then ${C}_\mathrm{faulty}^{\prime\mathrm{synd}}$ is $c_W\delta_1$-reducible to $C_\re^\wait$ with the error parameter at each wait location on a qubit upper bounded by $\lambda$.
    
    As a result, from Lemmas~\ref{lemma: transitive rule} and \ref{lemma: extension rule}, $\{C,\{p_i\}_{i\in C}\}$ is $\delta$-reducible to $\{C',\{q_j\}_{j\in C'}\}$ by removing the non-faulty part of ${C''}^{\mathrm{wait}}\cup C^{\prime\synd}_{\mathrm{nonfaulty}}$ from $C''$, i.e., the $\delta$-reducibilities of the diagrams in \eqref{eq: 0-state-preparation in threshold thorem}, \eqref{eq: clifford-state-preparation in threshold thorem}, and \eqref{eq: magic-state-preparation in threshold thorem} hold.
    Here, due to the additive bound from Lemma~\ref{lemma: transitive rule}, the parameter $\delta$ follows from
    \begin{align}
        c_W\delta_1(D)+\delta_3(N,L)&\leq (c_W+1)\times \frac{\delta}{c_W+1}\\
        &=\delta,
    \end{align}
    where we used \eqref{eq: delta1} and \eqref{eq: delta3}.
    Thus, we conclude the proof.
\end{proof}

Next, we show that the circuit reducibilities of the CNOT-gate rectangle in \eqref{eq: cnot-gate in threshold thorem}, Pauli rectangles in \eqref{eq: Pauli gate in threshold thorem}, and the wait rectangle in \eqref{eq: wait in threshold thorem}.

\begin{proof}[Proof of~\eqref{eq: cnot-gate in threshold thorem},~\eqref{eq: Pauli gate in threshold thorem}, and~\eqref{eq: wait in threshold thorem}]

    Let $C$ denote an incomplete circuit as
    \begin{align}
        C\coloneqq {C}^{\mathrm{wait}}_\re\cup{C}^{\mathrm{wait}}_\prece\cup C^{\mathrm{gate}}\cup {C}^{\mathrm{wait}}_\sync \cup C^{\mathrm{EC}},
    \end{align}
    where $C$ is composed of a single-depth circuit $C_\re^\wait$ of wait operations, followed by a single-depth circuit $C^{\mathrm{wait}}_\prece$ of wait operations in the preceding EC gadget, followed by a sinlge-depth CNOT-gate, Pauli-gate or wait gadget $C^{\mathrm{gate}}$, further followed by a single-depth circuit of wait locations $C_\sync^\wait$ for the synchronization, an depth-$d$ EC gadget $C^{\mathrm{EC}}$.
    Here, $C$ has a set of error parameters $\{p_i\}_{i\in C}$ such that
        \begin{align}
    p_i&\leq \lambda&&\text{if $i \in {C}^{\mathrm{wait}}_\re$},\label{eq:reduced error parameter q_j 2}\\
   p_i&\leq p_{\mathrm{loc}}&&\text{if $i \in C^\gat$ or $C^\EC $},\\
   p_i&\leq dp_{\mathrm{loc}}&&\text{if $i \in C^\wait_\prece$ or $ C^\wait_\sync $}.
    \end{align}
    Let $C'$ denote an incomplete circuit composed of ${C}^{\mathrm{gate}}$, followed by a single-depth circuit ${C}^{\prime\wait}_\re$, and a single depth circuit ${C}^{\mathrm{wait}}$, i.e.,
    \begin{align}
         C'\coloneqq {C}^{\mathrm{gate}}\cup {C}^{\prime\wait}_\re\cup {C}^{\mathrm{wait}}.
    \end{align}
    Here, $C'$ has a set of error parameters $\{q_j\}_{j\in C'}$ such that
        \begin{align}
    q_j&= 0&&\text{if $j \in {C}^{\mathrm{gate}}$},\label{eq: reduced error parameter in gate gadget 1}\\
    q_j&\leq \lambda&&\text{if $j \in {C}^{\prime\mathrm{wait}}_\re$},\label{eq: error parameter in gate gadget 2}\\
    q_j&\leq dp_{\mathrm{loc}}&&\text{if $j \in {C}^{\mathrm{wait}}$}.\label{eq: error parameter in gate gadget 3}\\
    \end{align}
    As shown in the diagrams, we will show the circuit reducibility from $\{C,\{p_i\}_{i\in C}\}$ to $\{C',\{q_j\}_{j\in C'}\}$.
    To this end, we first consider circuit reduction of an incomplete circuit of
    \begin{equation}
        C''\coloneqq {C}^{\mathrm{wait}}_\re\cup{C}^{\mathrm{wait}}_\prece\cup C^{\mathrm{gate}}\cup {C}^{\mathrm{wait}}_\sync\cup C^{\mathrm{synd}},
    \end{equation}
    and the set of error parameters $\{p_i\}_{i\in C''}$ 

    To apply Lemma~\ref{lemma: condition for reducibility} for the reduction of $C''$, we consider a partition of the set $C''$ into disjoint sets $C_1$ and $C_2$.
    Here, $C_1$ corresponds to $C''\setminus C_\re^\wait$, and $C_2$ corresponds to  $C_\re^\wait$.
    We consider error propagation from a given set of faulty locations in ${C''}$ to the set of locations until the time step after the $Z$-basis measurements or the final time step of the wait locations in ${C^\synd}$.
    Let us denote the non-faulty locations in the middle of $C''$ after this error propagation by $C''_{\mathrm{nonfaulty}}$, and the faulty locations by $C''_{\mathrm{nonfaulty}}$, i.e.,
    \begin{align}
        {C''}=C''_\mathrm{nonfaulty}\cup C''_\mathrm{faulty}.
    \end{align}
    This propagation define the map $\Gamma: 2^{ {C''}}\rightarrow 2^{C''_{\mathrm{nonfaulty}}}$.
    Similarly to the above proof for preparation gadgets, we have
    \begin{align}
        |\Gamma(\{i\})|\leq c'_\fw= 2^{2d_0} \text{~for all~} i\in {C'}^{\mathrm{synd}}.
    \end{align}
    Moreover, in the error propagation, each location in ${C}_\mathrm{faulty}^{\mathrm{synd}}$ may be affected by errors occurring in at most $2^{2d_0}$ possible locations from $C_\re^\wait$ and at most $2^{2d_0}\times 2(d_0+1)=2^{2d_0+1}(d_0+1)$ possible locations from ${C''}\setminus C_\re^\wait$; thus we have for all $j \in C''_{\mathrm{faulty}}$,
    \begin{align}
        &|\{ i\in {C''}\setminus{C}^\wait_\re \colon  j\in\Gamma^\synd(\{i\})  \}| \nonumber\\
        &\le c_\bw^{\prime (1)}= 2^{2d_0+1}(d_0+1),\\
        &|\{ i\in {C}^\wait_\re \colon  j\in\Gamma^\synd(\{i\})  \}| \le c_\bw^{\prime (2)}= 2^{2d_0}.
    \end{align}  
Furthermore, under the partition, we have $p_i=p^{(1)}= dp_\loc$ for all $i\in C''\setminus C^\wait_\re $ and $p_i\leq p^{(2)} =\lambda$ for all $i\in C^\wait_\re$.
    As a result, if we choose $p''_j$ for all $j\in C''_{\mathrm{faulty}}$ as
    \begin{align}
        p''_j&= 2^{c_{\bw}^{\prime (2)}}\left(1 + \frac{p^{(1)}}{p^{(2)}}\right)^{c_\bw^{\prime(1)}} \qty(p^{(2)})^{1/c'_{\fw}}+2^{c_{\bw}^{\prime(1)}}\qty(p^{(1)})^{1/c'_{\fw}}\\
        &=2^{2^{2d_0}}\qty(1+\frac{d p_\loc}{\lambda})^{2^{2d_0+1}(d_0+1)}{\lambda}^{1/2^{2d_0}}\nonumber \\
        &~~~~~~~~~+2^{2^{2d_0+1}(d_0+1)}(dp_\loc)^{1/2^{2d_0}},
    \end{align}
    the map $\Gamma$ and the sets of error parameters $\{p_i\}_{i\in {C''}}$ and $\{p''_j\}_{j\in C''_{\mathrm{faulty}}}$ satisfy all conditions in Lemma~\ref{lemma: condition for reducibility}.
    Thus, a pair $\qty(C'',\{p_i\}_{i\in C''})$ is reducible to $\{C''_{\mathrm{faulty}},\{p''_j\}_{j\in C''_{\mathrm{faulty}}}\}$.

    Next, we consider the circuit reducibility with the decoding algorithm.
    We can bound the error parameters of the local stochastic Pauli error model on physical qubits and syndrome bits in Def.~\ref{def: Local stochastic Pauli error model on physical qubits and syndrome bits} as
    \begin{align}
        p_{\mathrm{phys}}&\leq   2^{2^{2d_0}}\qty(1+\frac{d p_\loc}{\lambda})^{2^{2d_0+1}(d_0+1)}{\lambda}^{1/2^{2d_0}}\nonumber \\
        &~~~~~~~~~+2^{2^{2d_0+1}(d_0+1)}(dp_\loc)^{1/2^{2d_0}},\\
        p_{\mathrm{synd}}&\leq   2^{2^{2d_0}}\qty(1+\frac{d p_\loc}{\lambda})^{2^{2d_0+1}(d_0+1)}{\lambda}^{1/2^{2d_0}}\nonumber \\
        &~~~~~~~~~+2^{2^{2d_0+1}(d_0+1)}(dp_\loc)^{1/2^{2d_0}}.
    \end{align}
    In this case, if we have
    \begin{align}
        p_\mathrm{loc}<q_\gat^\thre
    \end{align}
    for $q_\gat^\thre$ in~\eqref{eq:q_gate_th},
    then we have
\begin{align}
\label{eq: lambda condition in gate rectangle}
      &2^{2^{2d_0}}\qty(1+\frac{d p_\loc}{\lambda})^{2^{2d_0+1}(d_0+1)}{\lambda}^{1/2^{2d_0}}\nonumber \\
        &~~~~~~~~~+2^{2^{2d_0+1}(d_0+1)}(dp_\loc)^{1/2^{2d_0}} < p_\dec^\thre,
\end{align}
so we can use the single-shot decoding algorithm in Def.~\ref{def: single-shot decoding algorithm with thresholds to suppress residual error} with a threshold of $p^{\thre}_{\mathrm{dec}}>0$.
Thus, if $p_\loc<p_\gat^\thre$, due to \eqref{eq: c(T) in threshold theorem}, the single-shot decoding algorithm in Def.~\ref{def: single-shot decoding algorithm with thresholds to suppress residual error} guarantees that the error parameter on physical qubits after a recovery operation is applied is bounded by
\begin{align}
\label{eq:threshold_condition_2}
    p_{\mathrm{phys}}< \qty(p_\mathrm{dec}^\mathrm{th}) ^{c(T)}\leq \lambda.
\end{align}

Since, by Lemma~\ref{lemma: fault-tolerance condition}, the width of the gate gadget is bounded by $c_WN$, the decoding algorithm is performed for at most $c_W$ code blocks.
Then, using the decoding algorithm in Def.~\ref{def: single-shot decoding algorithm with thresholds to suppress residual error}, $\qty(C''_{\mathrm{faulty}},\{p''_i\}_i)$
is $c_W\delta_1$-reducible to $\qty(C_\re^{\prime\wait},\{q_j\}_j)$, where $q_j\leq \lambda$ if $j$ in $C_\re^{\prime\wait}$, and the prefactor of $c_W\delta_1$ is due to the gadgets having at most $c_W$ code blocks.

As a result, from Lemmas~\ref{lemma: transitive rule} and \ref{lemma: extension rule} and by following the same argument as the above proof for the preparation gadgets, we see that $\qty(C,\{p_i\})_{i\in C}$ is $\delta$-reducible to $\qty(C',\{q_j\})_{j\in C'}$, i.e., the $\delta$-reducibilities of the diagrams in \eqref{eq: cnot-gate in threshold thorem}, \eqref{eq: Pauli gate in threshold thorem}, and \eqref{eq: wait in threshold thorem} hold due to 
\begin{equation}
    c_W\delta_1\leq \delta,
\end{equation}
where we used~\eqref{eq: delta1}.
\end{proof}

Finally, we show that the circuit reducibilities of the $Z_K$-measurement rectangle in \eqref{eq: measurement in threshold thorem} and the Bell measurement rectangle in \eqref{eq: Bell measurement in threshold thorem}.

\begin{proof}[Proof of~\eqref{eq: measurement in threshold thorem} and~\eqref{eq: Bell measurement in threshold thorem}]

Let $C$ denote an incomplete circuit composed of an at single-depth circuit ${C}^{\mathrm{wait}}_\re$, followed by a single-depth circuit $C^\wait$ of wait operations, followed by an at most depth-$d_0$ circuit in $Z_K$- or Bell-measurement gadget $C^{\mathrm{meas}}$, i.e.,
    \begin{align}
        C=C^\wait_\re\cup {C}^{\mathrm{wait}}\cup C^{\mathrm{meas}}.
    \end{align}
    Here, $C$ has the set of error parameters
        \begin{align}
    p_i&\leq \lambda&&\text{if $i \in {C}^{\mathrm{wait}}_\re$},\label{eq:reduced error parameter q_j 2}\\
   p_i&\leq p_\loc&&\text{if $i \in C^\meas$},\\
   p_i&\leq dp_{\mathrm{loc}}&&\text{if $i \in  C^\wait $}.
    \end{align}
    Let $C'$ denote an incomplete circuit composed only of $C^{\mathrm{meas}}$, i.e.,
    \begin{align}
        C'=C^{\mathrm{meas}}.
    \end{align}
    Here, $C'$ has the set of error parameters
        \begin{align}
    q_j&= 0&&\text{if $j \in C'$}.\label{eq:reduced error parameter q_j 2}
    \end{align}

To apply Lemma~\ref{lemma: condition for reducibility} for the reduction of $C$, we consider a partition of the set $C$ into disjoint sets $C_1$ and $C_2$.
    Here, $C_1$ corresponds to $C\setminus C_\re^\wait$, and $C_2$ corresponds to  $C_\re^\wait$.
    We consider the propagation of Pauli errors at the faulty locations in $C$ to the time step after the $Z$-basis measurements in ${C}^{\mathrm{meas}}$. 
    Let us denote the non-faulty locations in the middle of $C$ after this error propagation by $C_{\mathrm{nonfaulty}}$, and the faulty locations by $C_{\mathrm{faulty}}$, i.e.,
    \begin{align}
        C=C_\mathrm{nonfaulty}\cup C_\mathrm{faulty}.
    \end{align}
    By considering a map $\Gamma\colon 2^C\rightarrow 2^{C_\mathrm{faulty}}$, we have
    \begin{align}
        |\Gamma(\{i\})|\leq c''_\fw= 2^{d_0} \text{~for all~} i\in {C'}^{\mathrm{synd}}.
    \end{align}

    Moreover, in the error propagation, each location in ${C}_\mathrm{faulty}$ may be affected by errors occurring in at most $2^{d_0}$ possible locations from $C^\wait_\re$ and at most $2^{d_0+1}\times (d_0+1)$ possible locations from $C^\wait\cup C^\meas$; thus we have for all $j \in C_{\mathrm{faulty}}$,
    \begin{align}
        &|\{ i\in C^\wait\cup C^\meas \colon  j\in\Gamma(\{i\})  \}| \le c_\bw^{\prime\prime(1)} = 2^{d_0+1} (d_0+1),\\
        &|\{ i\in {C}^\wait_\re \colon  j\in\Gamma(\{i\})  \}| \le c_\bw^{\prime\prime(2)} =2^{d_0}.
    \end{align}    
Furthermore, under the partition, we have $p_i=p^{(1)}= dp_\loc$ for all $i\in C\setminus C^\wait_\re $ and $p_i\leq p^{(2)} =\lambda$ for all $i\in C^\wait_\re$.
    As a result, if we choose $p'_j$ for all $j\in C_{\mathrm{faulty}}$ as
    \begin{align}
        p'_j&= 2^{c_{\bw}^{\prime\prime(2)}}\left(1 + \frac{p^{(1)}}{p^{(2)}}\right)^{c_\bw^{\prime\prime(1)}} \qty(p^{(2)})^{1/c''_{\fw}}+2^{c_{\bw}^{\prime\prime(1)}}\qty(p^{(1)})^{1/c''_{\fw}}\\
        &=2^{2^{d_0}}\qty(1+\frac{d p_\loc}{\lambda})^{2^{d_0+1}(d_0+1)}{\lambda}^{1/2^{d_0}}\nonumber \\
        &~~~~~~~~~~~~~~~~~~~~~~~~+2^{2^{d_0+1}(d_0+1)}(dp_\loc)^{1/2^{d_0}},
    \end{align}
    the map $\Gamma$ and the sets of error parameters $\{p_i\}_{i\in {C}}$ and $\{p'_j\}_{j\in C_{\mathrm{faulty}}}$ satisfy all conditions in Lemma~\ref{lemma: condition for reducibility}.
    Thus, from Lemma~\ref{lemma: condition for reducibility}, a pair $\qty({C, \{p_i\}_i)})$ is reducible to $\qty( C_\mathrm{faulty}, \{p'_j\}_j)$, where 
    \begin{align}
    &p'_j\leq  2^{2^{d_0}}\qty(1+\frac{d p_\loc}{\lambda})^{2^{d_0+1}(d_0+1)}{\lambda}^{1/2^{d_0}}\nonumber \\
        &~~~~~~~~~~~~~~~~~~~~~~~~+2^{2^{d_0+1}(d_0+1)}(dp_\loc)^{1/2^{d_0}}\\
    &\text{for all $j\in C_{\mathrm{faulty}}$}
    \end{align}

    Since there are no syndrome bit errors occurred, we can bound the error parameters of the local stochastic Pauli error model on physical qubits in Def.~\ref{def: Local stochastic Pauli error model on physical qubits and syndrome bits} as
    \begin{align}
        p_{\mathrm{phys}}&\leq 2^{2^{d_0}}\qty(1+\frac{d p_\loc}{\lambda})^{2^{d_0+1}(d_0+1)}{\lambda}^{1/2^{d_0}}\nonumber \\
        &~~~~~~~~~~~~~~~~~~~~~~~~+2^{2^{d_0+1}(d_0+1)}(dp_\loc)^{1/2^{d_0}}.
    \end{align}
In this case, if we have
\begin{align}
    p_\mathrm{loc}<q_\meas^\thre
\end{align}
for $q_\meas^\thre$ in~\eqref{eq:q_meas_th},
then we obtain
\begin{align}
\label{eq: threshold ineq in meas}
      &2^{2^{d_0}}\qty(1+\frac{d p_\loc}{\lambda})^{2^{d_0+1}(d_0+1)}{\lambda}^{1/2^{d_0}}\nonumber \\
        &~~~~~~~~~+2^{2^{d_0+1}(d_0+1)}(dp_\loc)^{1/2^{d_0}}<p_\mathrm{dec}^{\mathrm{th}},
\end{align}
so we can use the decoding algorithm for correcting errors in Def.~\ref{def: Logarithmic-time decoding algorithm with threshold}.

Since, by Lemma~\ref{lemma: fault-tolerance condition}, the width of measurement gadget is bounded by $c_WN$, the decoding algorithm is performed for at most $c_W$ code blocks.
Then, if $p_\loc<q_\meas^\thre$, by using the decoding algorithm in Def.~\ref{def: Logarithmic-time decoding algorithm with threshold}, $\qty(C_{\mathrm{faulty}},\{p'_i\}_i)$ is $c_W\delta_2$-reducible to $\qty(C_{\mathrm{faulty}},\{q_j\}_j)$, where $q_j=0$, and the prefactor of $c_W\delta_2$ is due to the gadgets having at most $c_W$ code blocks.

Thus, from Lemmas~\ref{lemma: transitive rule} and \ref{lemma: extension rule},
we see that $(C,\{p_i\}_{i\in C})$ is $\delta$-reducible to $(C',\{q_j\}_{j\in C'})$, i.e.,
the $\delta$-reducibilities of the diagrams in \eqref{eq: measurement in threshold thorem} and \eqref{eq: Bell measurement in threshold thorem} hold due to 
\begin{equation}
    c_W\delta_2\leq \delta,
\end{equation}
where we used~\eqref{eq: delta2}.
\end{proof}

By sequentially applying the circuit reduction procedure to each rectangle in the fault-tolerant circuit, using Lemma~\ref{lemma: circuit-reducibility of rectangles}, along with Lemmas~\ref{lemma: transitive rule} and \ref{lemma: extension rule}, we replace each rectangle with its non-faulty version from the start to the end.
After all rectangles are replaced, a non-faulty version of the fault-tolerant circuit is obtained.
In the following, we show that the overall failure probability of these applications of the circuit reduction procedure is bounded by $\varepsilon$, which completes the proof of the threshold theorem for our protocol with the quantum LDPC codes.

Suppose that $p_\mathrm{loc}<q_\mathrm{loc}^\mathrm{th}$ for the threshold $q_\mathrm{loc}^\mathrm{th}$ in~\eqref{eq: q_th in threshold theorem}.
Since our protocol allocates five auxiliary registers for each register, the number of rectangles in the fault-tolerant circuit is bounded by
\begin{align}
    &6\kappa(n,\varepsilon)\times D_\FT(n,\varepsilon).
\label{eq: number of gadgets in the physical circuit}
\end{align}
where $\kappa(n,\varepsilon)$ and $D_\FT(n,\varepsilon)$ are bounded by~\eqref{eq: kappa in threshold theorem} and~\eqref{eq: depth of the fault-tolerant circuit}.
Then, from Lemma~\ref{lemma: circuit-reducibility of rectangles}, each rectangle is $\delta(n,\varepsilon)$-reducible to a non-faulty rectangle with $\delta(n,\varepsilon)$ given by \eqref{eq: delta in threshold theorem}.
Therefore, from Lemmas~\ref{lemma: transitive rule} and \ref{lemma: extension rule}, by performing the partial circuit reductions to all rectangles in ${C}_n^{\mathrm{FT}}$ from state-preparation gadgets to measurement gadgets sequentially, we obtain a parameter $\delta^{\FT}(n,\varepsilon)$ such that a pair $\qty({C}_n^{\mathrm{FT}}, \{p_i\}_{i\in C_n^{\FT}})$ with $p_i\leq p_{\loc}$ for all $i\in p_{\loc}$ is $\delta^{\FT}(n,\varepsilon)$-reducible to another pair $\qty({C}_n^{\prime\mathrm{FT}}, \{q_i\}_{i\in {C}_n^{\prime\FT}})$ with $q_j=0$ for all $j\in {C}_n^{\prime\FT}$.
By counting the number of rectangles in $C^{\mathrm{FT}}_n$, the additive bound in Lemma~\ref{lemma: transitive rule} provides $\delta^{\FT}(n,\varepsilon)$ as
\begin{align}
    &\delta^\FT(n,\varepsilon)\nonumber\\
    &\leq \delta(n,\varepsilon)\times \qty(6\kappa(n,\varepsilon)\times D_\FT(n,\varepsilon))\\
    &=O\qty(\qty(\frac{\varepsilon}{|C_n^\org|})^2 \frac{W(n)}{K(n)}D(n)\log^{2\alpha+\gamma_1+\gamma_2}\qty(\frac{|C_n^\org|}{\varepsilon}))\\
    &=O\qty(\varepsilon\times\qty(\frac{\varepsilon}{|C_n^\org|})\log^{\alpha+\gamma_1+\gamma_2}\qty(\frac{|C_n^\org|}{\varepsilon})),
\end{align}
where the second line follows from~\eqref{eq: kappa in threshold theorem},~\eqref{eq: depth of the fault-tolerant circuit},~\eqref{eq: DFT in threshold theorem}, and~\eqref{eq: delta in threshold theorem}, the third line from~\eqref{eq:assumption_circuit_size} and~\eqref{eq: Ktheta in threshold theorem}.
Thus, for every fixed $\varepsilon$, it holds for sufficiently large $n$ that
\begin{align}
    &\delta^\FT(n,\varepsilon)\leq \frac{\varepsilon}{2}.
\label{eq: failure probability of circuit reduction}
\end{align}

Let $q^{\mathrm{org}}$ be the output probability distribution of the original circuit and $q^{\mathrm{FT}}$ be the output probability distribution of the fault-tolerant circuit.
Then, we have
\begin{equation}
    q^{\mathrm{FT}}=\qty(1-\delta^\FT(n,\varepsilon))q^{\mathrm{org}}+\delta^\FT(n,\varepsilon) q^{\mathrm{fail}}.
\end{equation}
Therefore, the total variation distance between $q^{\mathrm{org}}$ and $q^{\mathrm{FT}}$ can be bounded as
\begin{align}
    &\sum_{x\in \mathbb{F}_2^{W(n)}}\left|q^{\mathrm{org}}(x)-q^{\mathrm{FT}}(x)\right|_1\nonumber\\
    &=\delta^\FT(n,\varepsilon) \sum_{x\in \mathbb{F}_2^{W(n)}}\left|q^{\mathrm{org}}(x)-q^{\mathrm{fail}}(x)\right|\\
    &\leq \varepsilon,
\label{eq:  failure probability total variation}
\end{align}
where we used~\eqref{eq: failure probability of circuit reduction} and $\sum_{x}\left|q^{\mathrm{org}}(x)-q^{\mathrm{fail}}(x)\right|\leq 2$.
As a result, $q_\loc^\thre$ in \eqref{eq: q_th in threshold theorem} defines the threshold of our protocol such that if $p_\loc\leq q_\loc^\thre$, the failure probability of a fault-tolerant simulation can be bounded by $\varepsilon$, which concludes the proof of Theorem~\ref{theorem: threshold theorem for polylog-time- and constant-space-overhead fault-tolerant quantum computation}.

\section{Conclusion}
\label{sec:conclusion}
In this work, we present a hybrid fault-tolerant protocol that combines concatenated codes for gate operations with a non-vanishing-rate quantum LDPC code (in particular, the quantum expander codes~\cite{Leverrier2015QuantumExpander, Fawzi_2018_eff, Fawzi_2018}) for quantum memory. 
Our protocol achieves polylogarithmic time overhead while maintaining constant space overhead.
This approach improves upon the time overhead of existing constant-space-overhead protocols, which exhibit either polynomial time overhead with non-vanishing-rate quantum LDPC codes~\cite{Gottesman2014Constant, Fawzi_2018, grospellier:tel-03364419}, as well as the quasi-polylogarithmic time overhead of the protocol with concatenated quantum Hamming codes~\cite{Yamasaki_2024}. 
This improvement is achieved by increasing the parallelism for executing logical gates from $O(W(n)/\mathrm{poly}(n))$ to $O(W(n)/\mathrm{polylog}(n))$, improving over Ref.~\cite{Gottesman2014Constant}, by incorporating recent developments in the analysis of decoding algorithms~\cite{Fawzi_2018, grospellier:tel-03364419}.
These results eliminate redundant spacetime tradeoffs present in the analyses of the existing constant-space-overhead FTQC protocols~\cite{Gottesman2014Constant, Fawzi_2018, grospellier:tel-03364419,Yamasaki_2024}; in particular, compared to the conventional fault-tolerant protocols with polylogarithmic space and time overheads, our protocol makes it possible to reduce the space overhead to constant without increasing the time overhead, up to the degree of the polylogarithmic factors.

Moreover, from the perspective of technical contriubution,
we introduce a new technique called partial circuit reduction to analyze errors in fault-tolerant circuits by locally analyzing the errors of each rectangle; with this technique, the analysis completes the proof of the threshold theorem for the constant-space-overhead protocol with quantum LDPC codes, closing a logical gap overlooked in the existing analyses~\cite{Gottesman2014Constant, Fawzi_2018, grospellier:tel-03364419} that did not explicitly argue the error model of the entire fault-tolerant circuit when analyzing each gadget. 
In addition, our proof of the threshold theorem works even when accounting for the runtime of classical computation during executing the fault-tolerant circuit, e.g., for decoding and gate teleportation, whereas existing analyses of protocols for quantum LDPC codes assume that such classical computation is instantaneous~\cite{Gottesman2014Constant, Fawzi_2018, grospellier:tel-03364419}.

These results highlight that the quantum LDPC code approach can achieve a low-time overhead FTQC while maintaining a constant space overhead, as well as the code-concatenation approaches~\cite{Yamasaki_2024, yoshida2024concatenate}, contributing to the fundamental understanding of low-overhead FTQC. 
Considering recent advances in experiments~\cite{Bluvstein_2023,xu_constant-overhead_2024,PhysRevX.13.041052,sunami2024scalablenetworkingneutralatomqubits,PhysRevLett.129.050504} and implementation methods of LDPC codes~\cite{bravyi2024high, delfosse2021bounds,Tremblay_2022,berthusen20242d, pattison2023hierarchical}, a comprehensive investigation is needed to compare the quantum LDPC codes approach and the code-concatenation approach for low-overhead FTQC\@.
Our results are fundamental for the future investigation of this issue to reveal which of the two approaches holds more promise for the future physical realization of FTQC.

\paragraph*{Note added}:
After the completion of our manuscript, a paper presenting a fault-tolerant protocol with constant space overhead and polylogarithmic time overhead~\cite{Nguyen2024} was uploaded, relying on a substantially different approach from ours. 
The primary contribution of our work, distinct from theirs, lies in the proof of the threshold theorem in Theorem~\ref{theorem: threshold theorem for polylog-time- and constant-space-overhead fault-tolerant quantum computation}, which takes into account the non-zero runtime of classical computation such as that of decoding algorithms as presented in Assumption~\ref{assump: non-vanishing-rate quantum LDPC with efficient decoding algorithm}. 
Furthermore, their work depends on unpublished assertions required for the existence of positive-rate quantum locally testable codes (QLTCs) with inverse-polylogarithmic relative distance, and inverse-polylogarithmic soundness (see private-communication citations in Theorems~6.9 and 6.11 (footnote~40) of Ref.~\cite{Nguyen2024}). 
In contrast, our work completes the proof of the threshold theorem with constant space overhead and polylogarithmic time overhead without relying on conjectures.

\begin{acknowledgements}
This work was supported by JST [Moonshot R\&D][Grant Number JPMJMS2061].
S.T.\ was supported by JSPS KAKENHI Grant Number 23KJ0521, and FoPM, WINGS Program, the University of Tokyo.
H.Y.\ acknowledges JST PRESTO Grant Number JPMJPR201A, JPMJPR23FC, JSPS KAKENHI Grant Number JP23K19970, and MEXT Quantum Leap Flagship Program (MEXT QLEAP) JPMXS0118069605, JPMXS0120351339\@.
The quantum circuits shown in this paper are drawn using \textsc{qpic} \cite{qpic}.

\end{acknowledgements}

\appendix

\section{Fault-tolerant protocol with concatenated Steane codes for open quantum circuits \label{appendix: fault-tolerant protocol for open quantum circuits}}
In this section, we explain a protocol based on concatenated Steane codes to simulate original open circuits that output a quantum state rather than measurement outcomes.
The protocol will be used to simulate the original open circuits to prepare logical states of a quantum LDPC code, and the encoded states are used to perform logical Clifford gates and logical $T$ and $T^{\dagger}$ gates acting on logical qubits in a quantum LDPC code via gate teleportation~\cite{gottesman2010introduction,knill2005scalable,Knill_2005}.
As discussed later in Sec.~\ref{sec: Construction of fault-tolerant gadgets and abbreviations(qLDPC)}, this protocol is not used in isolation.
It is integrated into the protocol with quantum LDPC codes in such a way that we use the concatenated Steane code to prepare encoded states of quantum LDPC codes in a fault-tolerant way.
The existing analysis~\cite{Gottesman2014Constant,Fawzi_2018,grospellier:tel-03364419} of the fault-tolerant protocol with quantum LDPC codes does not explicitly bound the error rate in preparing the encoded states, and one of our contributions here is to provide a thorough analysis including this part of the protocol to present a complete threshold theorem of the overall protocol.

\subsection{Compilation from original open circuit to fault-tolerant circuit \label{sec: Compilation from original open circuit to fault-tolerant circuit}}

We describe the procedure of our protocol to compile an original open circuit $C$ into a physical circuit that achieves the target failure probability for state preparation, $\delta_3 > 0$, in Def.~\ref{def: fault-tolerance conditions of the state preparation gadgets for quantum LDPC codes}.
To simulate original open circuits, we use a concatenated Steane code as described in Sec.~\ref{sec:Concatenated Steane codes} to simulate original open circuits, with code parameters
\begin{equation}
    [[N=7^L, K=1, D=3^L]],
\end{equation}
where $L$ represents the concatenation level.
For each original circuit $C$, the protocol selects a suitable value of $L$ to achieve the target failure probability $\delta_3$.
The protocol initially represents the original circuit as a level-$L$ circuit, which consists of elementary operations on level-$L$ qubits. 
Next, for $l\in\{L,\ldots,1\}$, the protocol recursively constructs a level-$(l-1)$ circuit from the level-$l$ circuit by decomposing it into a sequence of elementary operations acting on level-$(l-1)$ qubits.
At the base level, the level-$0$ circuit corresponds to a physical circuit.
Each elementary operation in a level-$l$ circuit is referred to as a location within that circuit.

For a level-$l$ circuit, we require that at most one elementary operation can act on a single level-$l$ qubit in a single time step.
The depth of a level-$l$ circuit is determined by the number of time steps.
A set of elementary operations includes $\ket{0}$-state preparation, $\ket{T}$-state preparation, $H$-, $S$-, $S^{\dagger}$-, CNOT-, and Pauli-gate operations, $Z$-basis measurement operation, and a wait operation, where $\ket{T}\coloneqq TH\ket{0}$.
For each elementary operation in the level-$l$ circuit, we define the level-$l$ gadget for the elementary operation, which is a level-$(l-1)$ circuit to carry out the logical elementary operation on logical qubits of the code blocks of the Steane code.
We also define a level-$l$ error-correction (EC) gadget, which is a level-$(l-1)$ circuit that performs error correction on a set of seven level-$(l-1)$ qubits forming the Steane code.

Moreover, we introduce a level-$l$ decoding gadget, which is a level-$(l-1)$ circuit to transform a logical state encoded in the Steane code consisting of the set of seven level-$(l-1)$ qubits to the same state of an unencoded level-$(l-1)$ qubit.

The elementary operations are represented in a diagram.
In the diagram, a change from a dashed input line to a solid output line represents the allocation of a level-$l$ qubit, while a change from a solid line to a dashed line represents deallocation. 
If both the input and output lines are dashed, the corresponding level-$l$ qubits are used as a workspace to perform an elementary operation.
A double output line represents a bit allocated by an elementary operation.

The elementary operations are represented as follows.
\begin{itemize}
    \item $\ket{0}$-state preparation
\begin{align}
    \includegraphics{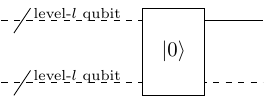}\label{Fig: 0-state-preparation}.
\end{align}
The $\ket{0}$-state preparation allocates a single level-$l$ qubit that is prepared in the state $\ket{0}$.

    \item $\ket{T}$-state preparation
\begin{align}
    \includegraphics{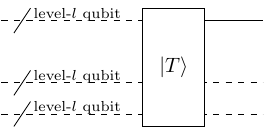}\label{Fig: magic-state-preparation}.
\end{align}
The $\ket{T}$-state preparation allocates a single level-$l$ qubit that is prepared in the state $\ket{T}$.

    \item $H$-gate operation
\begin{align}
    \includegraphics{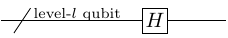}\label{Fig: h-gate}.
\end{align}
The $H$-gate operation applies the $H$ gate to a level-$l$ qubit.
    \item $S$-gate operation
\begin{align}
    \includegraphics{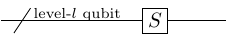}\label{Fig: s-gate}.
\end{align}
The $S$-gate operation applies the $S$ gate to a level-$l$ qubit.

    \item $S^{\dagger}$-gate operation
\begin{align}
    \includegraphics{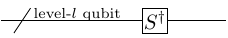}\label{Fig: s-dag-gate}.
\end{align}
The $S^{\dagger}$-gate operation applies the $S^{\dagger}$ gate to a level-$l$ qubit.

    \item CNOT-gate operation
\begin{align}
    \includegraphics{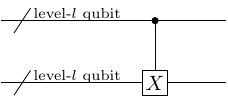}\label{Fig: cnot-gate}.
\end{align}
The CNOT-gate operation applies the CNOT gate between two level-$l$ qubits.

    \item Pauli-gate operation
\begin{align}
    \includegraphics{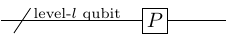}\label{Fig: pauli-gate}.
\end{align}
A Pauli-gate operation applies a Pauli gate $P\in\Tilde{\mathcal{P}}_1$ to a level-$l$ qubit.
\end{itemize}
\begin{itemize}
    \item $Z$-basis measurement operation
\begin{align}
    \includegraphics{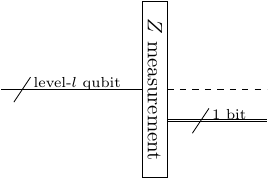}\label{Fig: z-basis-measurement}.
\end{align}
The $Z$-basis measurement operation performs the $Z$-basis measurement on a level-$l$ qubit.
It deallocates the level-$l$ qubit and outputs a $1$-bit measurement outcome.
    \item Wait operation
    \begin{align}
    \includegraphics{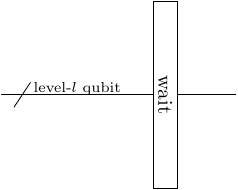}
    \label{Fig: wait operation box},
\end{align}
or simply
\begin{align}
    \includegraphics{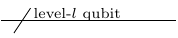}
    \label{Fig: wait operation}.
\end{align}
The wait operation applies the $I$ operation to a level-$l$ qubit, which is regarded as a special case of Pauli gates.
\end{itemize}

For the on-demand Clifford-state preparation in \eqref{eq: on-demand state preparation}, we need a on-demand Clifford operation.
\begin{itemize}
    \item On-demand Clifford operation
\begin{align}
    \includegraphics{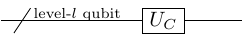}\label{Fig: uc-gate}.
\end{align}
    A on-demand Clifford operation is used for applying a Clifford unitary $U_C\in \{I, S, S^\dagger\}$ on demand and receives a $2$ bit string as input to specify a $U_C$.
\end{itemize}

In addition to these elementary operations and their corresponding gadgets, the EC gadget and the decoding gadget are represented as follows.
\begin{itemize}
    \item EC gadget
\begin{align}
    \includegraphics[width=0.35\textwidth]{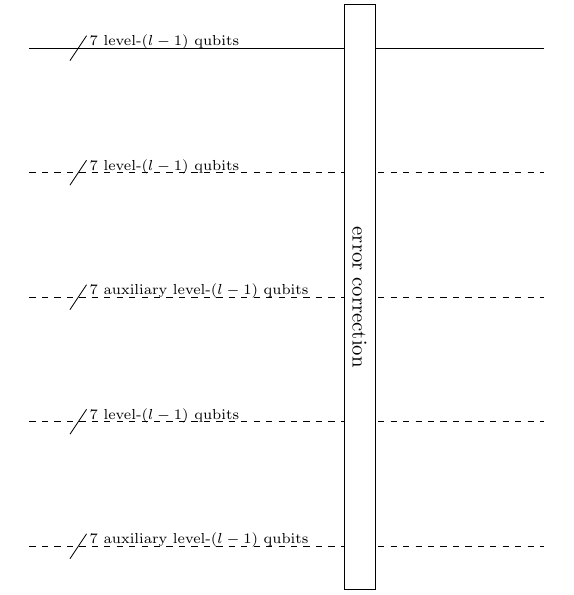}\label{Fig: EC gadget for concatenated code}.
\end{align}
An EC gadget performs quantum error correction on a Steane code consisting of 7 level-$(l-1)$ qubits, temporarily utilizing the other four sets of 7 level-$(l-1)$ qubits.

    \item 
    Decoding gadget
    \begin{align}
        \includegraphics[width=0.4\textwidth]{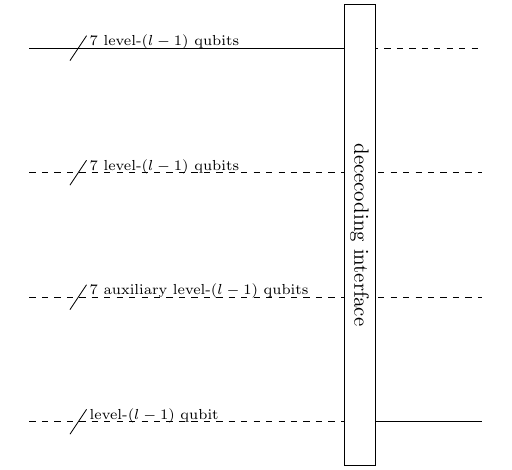}\label{Fig: dec_gadget},
    \end{align}
    or simply
     \begin{align}
        \includegraphics[width=0.4\textwidth]{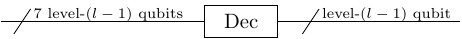}\label{Fig: dec_gadget_simple}.
    \end{align}
    A decoding gadget performs the decoding operation from an encoded level-$l$ qubit in a code block of the Steane code to an unencoded level-$(l-1)$ qubit.
\end{itemize}

\begin{figure*}[t]
    \includegraphics[width=\textwidth]{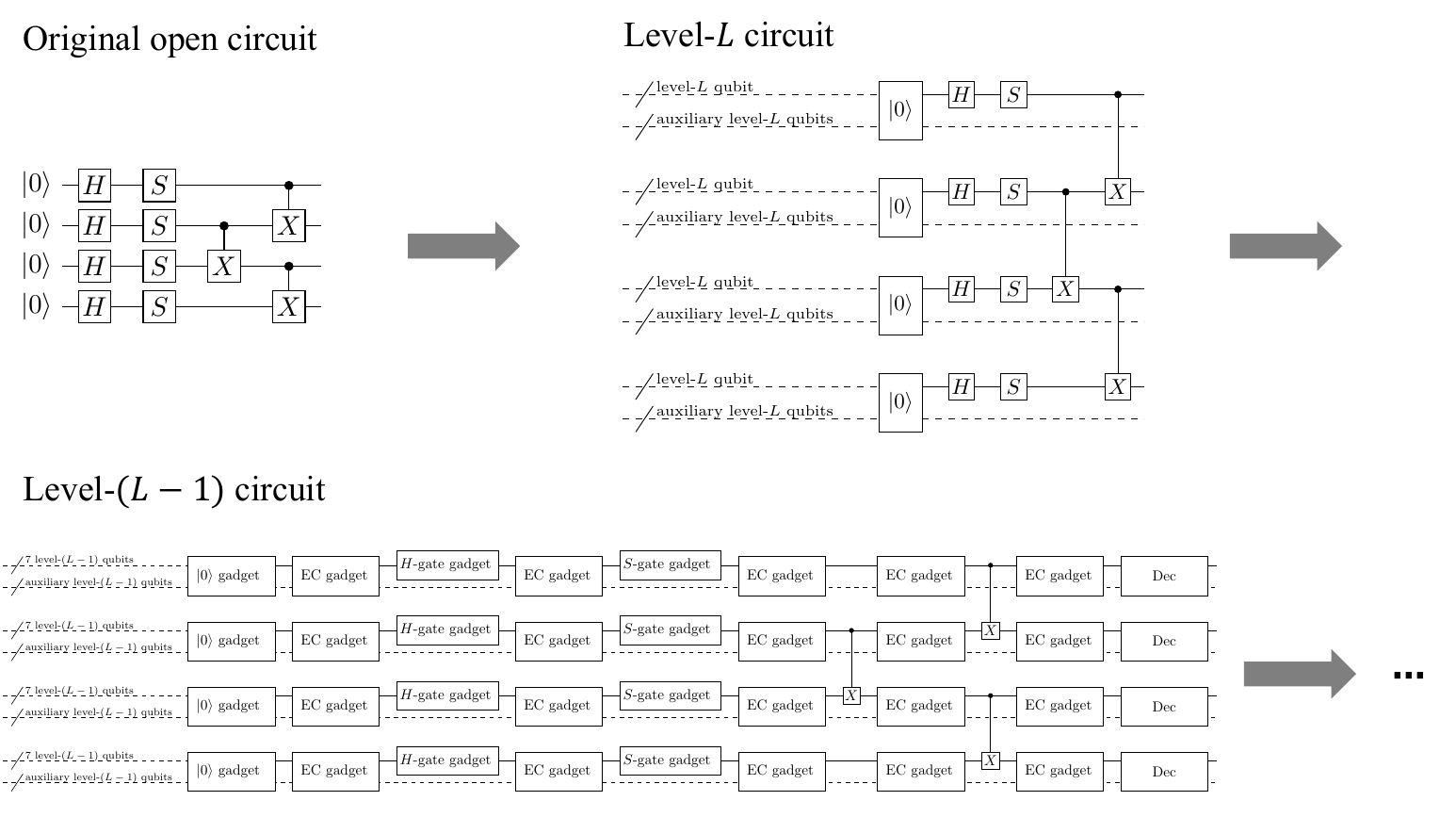}
    \caption{Compilation procedure of our protocol for original open circuits. First, we compile an original open circuit into a level-$L$ circuit that consists only of elementary operations acting on level-$L$ qubits. For each level $l\in\{L,\ldots,1\}$, we recursively compile the level-$l$ circuit into the level-$(l-1)$ circuit by replacing each elementary operation with the corresponding gadget, adding decoding gadgets in the end (denoted by Dec), and inserting EC gadgets between each pair of adjacent gadgets for operations and decoding. 
    Here, each output solid wire of the decoding gadget represents a single level-$(l-1)$ qubit, while each input solid wire represents seven level-$(l-1)$ qubits.
    In this way, we finally obtain a level-$0$ circuit, i.e., a physical circuit to simulate an original open circuit.}
    \label{fig: compilation procedure for ideal open circuits}
\end{figure*}
\noindent The gadgets are carefully designed to satisfy the fault-tolerance conditions, which are presented in Appendix~\ref{sec: Condition of fault tolerance}, and their constructions are shown in Appendix~\ref{sec: Construction of gadgets and abbreviation}.

The procedure for compiling an original open circuit to prepare a quantum state $\ket{\psi}$ circuit into a physical circuit is as follows, which we also show in Fig.~\ref{fig: compilation procedure for ideal open circuits}, is as follows.
Given an original open circuit $C$, we first compile the original circuit into a level-$L$ circuit, denoted by $C^{(L)}$, by replacing each operation in the original circuit with the corresponding elementary operation that acts on level-$L$ qubits.
This compilation can be performed without prior knowledge of the concatenation level $L$.
As for $T$ and $T^{\dagger}$ gates in the original circuit, we replace a $T$ gate with a sequence of elementary operations to perform $T$ gate on level-$L$ qubit via gate teleportation as shown in Fig.~\ref{Fig: t-steane-abbreviation-elementary.} and a $T^{\dagger}$ gate with an $S^{\dagger}$-gate operation, followed by the sequence, respectively.
\begin{figure}[t]
    \centering
    \includegraphics[width=0.45\textwidth]{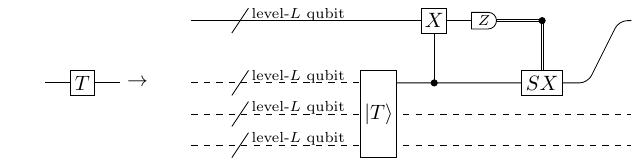}
    \caption{The replacement of the $T$ gate in an original open circuit into the sequence of elementary operations in the level-$L$ circuit for performing the $T$ gate through gate teleportation.}
    \label{Fig: t-steane-abbreviation-elementary.}
\end{figure}
In addition to each level-$L$ qubit, we allocate a constant number of auxiliary level-$L$ qubits.
Of these seven, three auxiliary level-$L$ qubits are used for elementary operations and $T$-gate abbreviations.
The other four registers are never explicitly used in the level-$L$ circuit and not shown in Fig.~\ref{fig: compilation procedure for ideal open circuits}, but these level-$L$ qubits are used to provide a workspace for EC gadgets that will appear in the level-$(L-1)$ circuit.
Next, for a given target failure probability $\delta_3$ and the number of locations in $C^{(L)}$, we determine the concatenation level $L$ such that the failure probability of the fault-tolerant simulation can be bounded by $\delta_3$.

For each level $l \in \{L, L-1,\ldots, 1\}$, we recursively compile the level-$l$ circuit into the corresponding level-$(l-1)$ circuit. 
For each level-$l$ qubit in the level-$l$ circuit, a set of seven level-$(l-1)$ qubits is allocated in the level-$(l-1)$ circuit, along with a constant number of auxiliary level-$(l-1)$ qubits per level-$l$ qubit.
As in the level-$L$ case, three of these seven auxiliary level-$(l-1)$ qubits are used for abbreviations and elementary operations in the level-$(l-1)$ circuit. 
The remaining four level-$(l-1)$ qubits are never explicitly used in the level-$(l-1)$ circuit, but are used as a workspace for EC gadgets.
Before compiling each level-$l$ circuit into the level-$(l-1)$ circuit, we expand the $T$-gate abbreviation into the corresponding elementary operations. 
Then, we replace each elementary operation with a combination of the corresponding gadget followed by an EC gadget acting on each of the seven level-$(l-1)$ qubits. This combination is referred to as a \textit{rectangle}.
Specifically, gate rectangles for the $H$-gate operation in \eqref{Fig: h-gate}, the $S$-gate operation in \eqref{Fig: s-gate}, the $S^{\dagger}$-gate operation in \eqref{Fig: s-dag-gate}, the CNOT-gate operation in \eqref{Fig: cnot-gate}, and the Pauli-gate operations in \eqref{Fig: pauli-gate} are defined as the corresponding gate gadget followed by an EC gadget on each of the seven level-$(l-1)$ qubits. State-preparation rectangles for the $\ket{0}$-state preparation in \eqref{Fig: 0-state-preparation} and the $\ket{T}$-state preparation in \eqref{Fig: magic-state-preparation} are similarly defined. The measurement rectangle for the $Z$-basis measurement operation in \eqref{Fig: z-basis-measurement} is defined as only the corresponding measurement gadget, without an EC gadget. 
For convenience in later analysis of threshold, we referred to a combination of a rectangle and an EC gadget preceding the rectangle as \textit{extended rectangle} or \textit{ExRec}.
For simplicity of analysis, EC gadgets are inserted synchronously, i.e., executed after all gadgets for elementary operations belonging to the same depth in the level-$l$ circuit have completed. 
Until the EC gadgets are completed, gadgets for the next depth in the level-$l$ circuit are not executed. Wait operations are incorporated after previously completed gadgets to ensure this synchronization.
Thus, under synchronization, rectangles may include wait operations between the gadget corresponding to the elementary operation and the EC gadget.

The level-$(l-1)$ circuit obtained at this stage outputs an encoded version of the quantum state $\ket{\psi}$, where each physical qubit of $\ket{\psi}$ is replaced by the 7 level-$(l-1)$ qubits forming the logical qubit of the Steane code. 
To obtain the unencoded state $\ket{\psi}$, we add the level-$l$ decoding gadget to each logical qubit at the end of the level-$(l-1)$ circuit $C^{(l-1)}$, as shown in Fig.~\ref{fig: compilation procedure for ideal open circuits}.
Thus, a level-$l$ circuit for $l\in\{L,\ldots, 1\}$ is compiled into a level-$(l-1)$ circuit composed of rectangles and the decoding gadgets instead of measurements at the end of closed circuits.
We write this level-$(l-1)$ circuit as $C^{(l-1)}$.
By applying this procedure recursively, we obtain the level-$0$ circuit $C^{(0)}$ on physical qubits.

\subsection{Conditions of fault-tolerant gadgets on concatenated Steane codes\label{sec: Condition of fault tolerance}}
In this section, we provide the fault-tolerance conditions for concatenated Steane codes~\cite{gottesman2010introduction, aliferis2005quantum}.
For the argument, we introduce the notion of an ideal decoder and $r$-filter~\cite{gottesman2010introduction, aliferis2005quantum}.
The definition of the ideal decoder is as follows.
\begin{definition}[Ideal decoder]
\label{def: ideal decoder}
An ideal decoder is a non-faulty operation that performs syndrome measurement, applies a recovery operation based on the decoding algorithm, and finally maps an encoded state of the Steane code to the corresponding unencoded state.
\end{definition}
The ideal decoder is represented in a diagram as
\begin{itemize}
\item Ideal decoder
\begin{align}
    \includegraphics{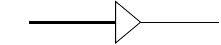}\label{Fig: ideal decoding operation},
\end{align}
\end{itemize}
where the bold line represents a logical qubit and the thin line represents an unencoded physical qubit.

The $r$-filter is a mathematical object used for analysis and does not correspond to any physical operation during quantum computation.
The $r$-filter is used to ensure that the weight of an error occurring on a codeword at a given time is not too large.
It is defined as follows:
\begin{definition}[$r$-filter]
\label{def: r-filter}
    An $r$-filter is a projector onto the subspace spanned by all states of the form $E\ket{\Bar{\psi}}$, where $E\in \mathcal{P}_N$ with a weight of at most $r$ and $\ket{\Bar{\psi}}\in\mathcal{Q}$ is a codeword. If the ideal syndrome bits of an encoded state were hypothetically obtained and the syndrome bits indicate a Pauli error with a weight of at most $r$, then the $r$-filter leaves the encoded state unchanged.
    Conversely, if the indicated error exceeds the weight $r$, the $r$-filter rejects the encoded state, terminating the computation.
\end{definition}
Reference~\cite{gottesman2010introduction} provides the definition for $[[N,1,2t+1]]$ codes where $0\leq r \leq t$.
However, for the Steane code with $t=1$, we can easily interpret the $r$-filter for $r=0,1$.
The $0$-filter is defined as a projector onto the code space, i.e.,
\begin{equation}
    \Pi_{\mathcal{Q}}\coloneqq \prod_i \frac{I+g_i}{2},
\end{equation}
where $\{g_i\}_i$ is a set of stabilizer generators as given in~\eqref{eq: stabilizer generator for Steane code.}.
The $0$-filter is illustrated in the diagram as,
\begin{itemize}
    \item $0$-filter
\begin{align}
  \nonumber\\&\quad\includegraphics{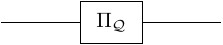}.
\end{align}
\end{itemize}
The $1$-filter for the Steane code becomes the identity operator $I$ acting on the whole Hilbert space because all the $X(Z)$-type errors that appear in the Steane code can be decomposed into logical $X(Z)$ errors plus weight-$1$ $X(Z)$ errors.
The $1$-filter is not shown in the diagram because it is trivial.

Then, the fault-tolerance conditions for preparation, gate, and measurement gadgets are presented below with diagrams using the ideal decoder in Def.~\ref{def: ideal decoder} and the $0$-filter in Def.~\ref{def: r-filter}.
In the diagrams, an operation surrounded by thick lines represents a faulty operation acting on a level-$l$ qubit, while an operation surrounded by thin lines represents a non-faulty operation acting on a level-$(l-1)$ qubit. The variable $s$ shown in the upper right of a gadget represents the number of faults in the gadget.
\begin{widetext}
The $\ket{0}$-state preparation gadget is fault-tolerant if it satisfies
\begin{align}
    &\text{Prep A: when $s=0$}\nonumber\\&\quad\includegraphics{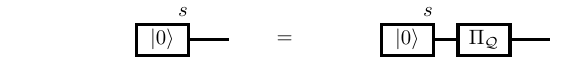}\label{eq: Prep A},\\
    &\text{Prep B: when $s=0,1$}\nonumber\\&\quad\includegraphics{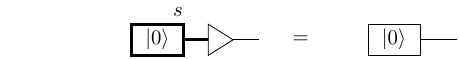}\label{eq: Prep B}.
\end{align}
Also, the $\ket{T}$-state preparation gadget is fault-tolerant if it satisfies
\begin{align}
    &\text{Prep A: when $s=0$}\nonumber\\&\quad\includegraphics{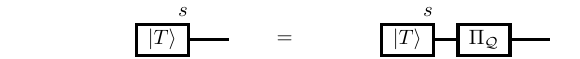}\label{eq: Prep A T},\\
    &\text{Prep B: when $s=0,1$}\nonumber\\&\quad\includegraphics{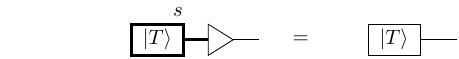}\label{eq: Prep B T}.
\end{align}
The Pauli-, $H$-, and $S$-gate gadgets are fault-tolerant if they satisfy the following conditions, with each gate denoted by $U$:
\begin{align}
    &\text{Gate A: when $s=0$}\nonumber\\&\quad\includegraphics{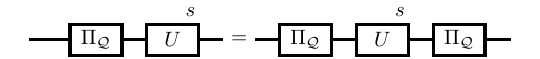},
    \label{eq: single-qubit gate A}\\
    &\text{Gate B: when $s=0$}\nonumber\\&\quad\includegraphics{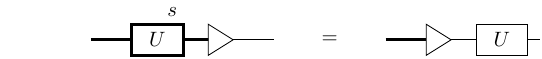}\label{eq: single-qubit gate B1},\\
    &\text{and when $s=0,1$}\nonumber\\&\quad\includegraphics{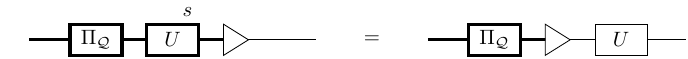}
    \label{eq: single-qubit gate B2}.
\end{align}
The CNOT-gate gadget is fault-tolerant if it satisfies, with the CNOT-gate denoted by $U$,

\begin{align}
    &\text{Gate A: when $s=0$}\nonumber\\&\quad\includegraphics{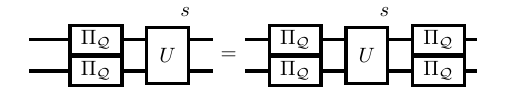},
    \label{eq: two-qubit gate 2A}\\
    &\text{Gate B: when $s=0$}\nonumber\\&\quad\includegraphics{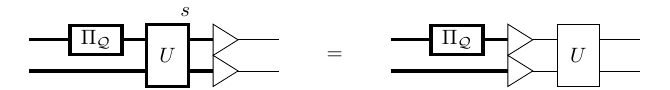}\label{eq: two-qubit gate 2B1},\\
    \nonumber\\&\quad\includegraphics{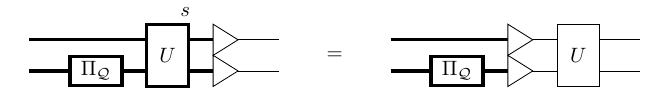}\label{eq: two-qubit gate 2B2},\\
    &\text{and when $s=0,1$}\nonumber\\&\quad\includegraphics{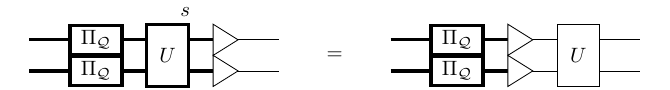}
    \label{eq: two-qubit gate 2B3}.
\end{align}

The $Z$-basis measurement gadget is fault-tolerant if it satisfies
\begin{align}
    &\text{Meas A: when $s=0$}\nonumber\\&\quad\includegraphics{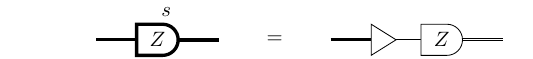}
    \label{eq: measA},\\
    &\text{Meas B: when $s=0,1$}\nonumber\\&\quad\includegraphics{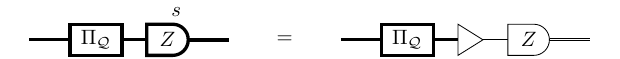}
    \label{eq: measB}.
\end{align}

In addition, we present the fault-tolerance conditions for EC gadgets.
In the diagrams, an operation surrounded by thick lines represents a faulty operation acting on seven level-$l$ qubits, while a thin line represents a single level-$l$ qubit.

The EC gadget is fault-tolerant if it satisfies 
\begin{align}
    &\text{EC A: when $s=0$}\nonumber\\&\quad\includegraphics{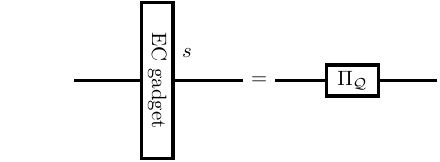},
    \label{eq: ec A}\\
    &\text{EC B: when $s=0$}\nonumber\\&\quad\includegraphics{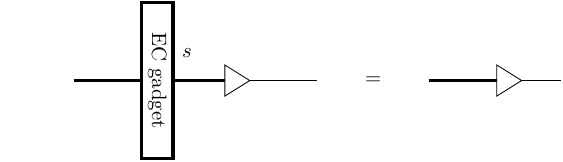},\label{eq: ecB1}\\
    &\text{and when $s=0,1$}\nonumber\\&\quad\includegraphics{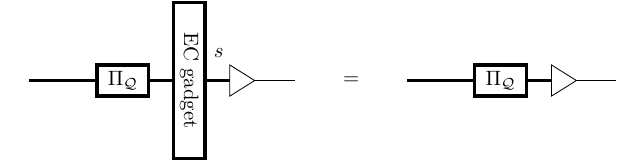}.\label{eq: ecB2}
\end{align}

Finally, we give the validity condition of decoding gadgets.
Again, in the diagrams, an operation surrounded by thick lines represents a faulty operation acting on seven level-$l$ qubits, while a thin line represents a single level-$l$ qubit.
We say that the decoding gadget is valid if it satisfies
\begin{align}
    &\text{Dec A: when $s=0$}\nonumber\\&\quad
    \includegraphics{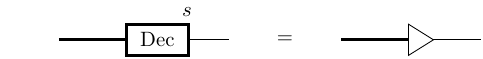}
    \label{eq: decA}.
\end{align}

\end{widetext}

\subsection{Constructions of abbreviations and gadgets \label{sec: Construction of gadgets and abbreviation}}
In this section, we present the explicit construction of an abbreviation and gadgets that satisfy the fault-tolerance condition presented from~\eqref{eq: Prep A} to~\eqref{eq: ecB2}.

\subsubsection{$Z$-basis measurement gadget\label{sec: Z-basis measurement gadget steane}}
The $Z$-basis measurement gadget is constructed as shown in Fig.~\ref{Fig: z-basis-measurement-steane.}.
The gadget is implemented by transversal $Z$-basis measurements, followed by the execution of the decoding algorithm by classical computation.

We present the explicit procedure for performing the $Z$-basis measurement.
Specifically, we give the procedure for calculating the measurement outcome $\bar{z}\in\{0,1\}$ of the Pauli operator $Z$ on a level-$l$ qubit from the noisy measurement outcomes $z_i\in \mathbb{F}_2$ of $Z_i$ on the $i$-th level-$(l-1)$ qubits for $i\in\{1,\ldots, 7\}$.
The procedure is as follows:
\begin{enumerate}
    \item We measure the $7$ level-$(l-1)$ qubits on the $Z$-basis. The outcomes are $z\in\mathbb{F}_2^{7}$, where each bit $z_i\in\mathbb{F}_2$ is the noisy measurement outcome on the $i$-th level-$(l-1)$ qubit.
    
    \item We calculate the syndrome bits $\sigma_X\in\mathbb{F}_2^{3}$ of the $Z$-type generators from the $7$-bit outcomes of the $Z$-basis measurement of the level-$(l-1)$ qubits by using the relation,

    \begin{equation}
    (\sigma_X)_i=\bigoplus_{j\in\mathrm{supp}(g_i^{Z})}z_{j},
    \end{equation}
    where $\mathrm{supp}(g_i^{Z})$ denotes the set of qubit indices where $g_i^{Z}$ has non-trivial support.
    
    \item We execute the decoding algorithm~\eqref{eq: decoding algorithm for Steane code} which takes the syndrome bits $\sigma_X \in \mathbb{F}_2^{3}$ as input and outputs the recovery operation of an $X$-type Pauli operator, denoted by $\hat{F} \in \mathcal{P}_{N}$.
    Then, the corrected bits $\tilde{z}\in\mathbb{F}_2^{7}$ are obtained for $i\in\{1,\ldots,7\}$ as
    \begin{equation}
        \tilde{z}_i\coloneqq z_i\oplus (\phi(\hat{F}))_{i},
    \end{equation}
    where $(\phi(\hat{F}))_{j}$ represents the $j$-th element of the row vector $\phi(\hat{F})$.
    \item  We calculate the noiseless measurement outcome $\bar{z}\in\mathbb{F}_2$ of $Z$ on the level-$l$ qubit as
    \begin{equation}
        \bar{z}\coloneqq \bigoplus_{i\in\mathrm{supp}(\Bar{Z})}\Tilde{z}_i,
    \end{equation}
    where $\Bar{Z}$ is the logical-$Z$ operator of the Steane code as given in~\eqref{eq: logical opeartor for Steane code}, and $\mathrm{supp}(\Bar{Z})$ denotes the set of indices of the physical qubits on which the logical-$Z$ operator acts nontrivially.
\end{enumerate}

\begin{figure*}[t]
    \centering
    \includegraphics[width=\textwidth]{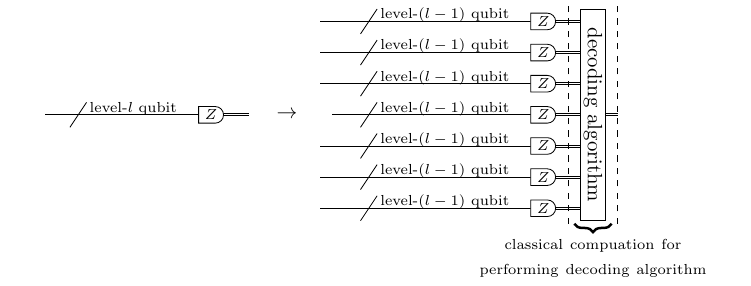}
    \caption{The level-$(l-1)$ circuit of the $Z$-basis measurement gadget. 
    The gadget is implemented by transversal $Z$-basis measurements, followed by the decoding algorithm in \eqref{eq: decoding algorithm for Steane code}.}
    \label{Fig: z-basis-measurement-steane.}
\end{figure*}
The fault tolerance of the $Z$-basis measurement gadget defined as in \eqref{eq: measA} and \eqref{eq: measB} is ensured by transversality.

\subsubsection{Gate gadgets}
The constructions of the Pauli-, $H$-, $S$-, $S^{\dagger}$ and CNOT-gate gadgets are shown in Fig.~\ref{Fig: gate gadgets.}.
The Pauli-, $H$-, $S$-, $S^{\dagger}$- and CNOT-gate gadgets are implemented by the transversal Pauli, $H$, $S^{\dagger}$, $S$ and CNOT gates acting on the level-$(l-1)$ qubits, respectively.
The fault tolerance of the gate gadgets defined as~\eqref{eq: single-qubit gate A},~\eqref{eq: single-qubit gate B1},~\eqref{eq: single-qubit gate B2},~\eqref{eq: two-qubit gate 2A},~\eqref{eq: two-qubit gate 2B1},~\eqref{eq: two-qubit gate 2B2} and~\eqref{eq: two-qubit gate 2B3} is ensured by transversality.

\begin{figure}
    \includegraphics[width=0.5\textwidth]{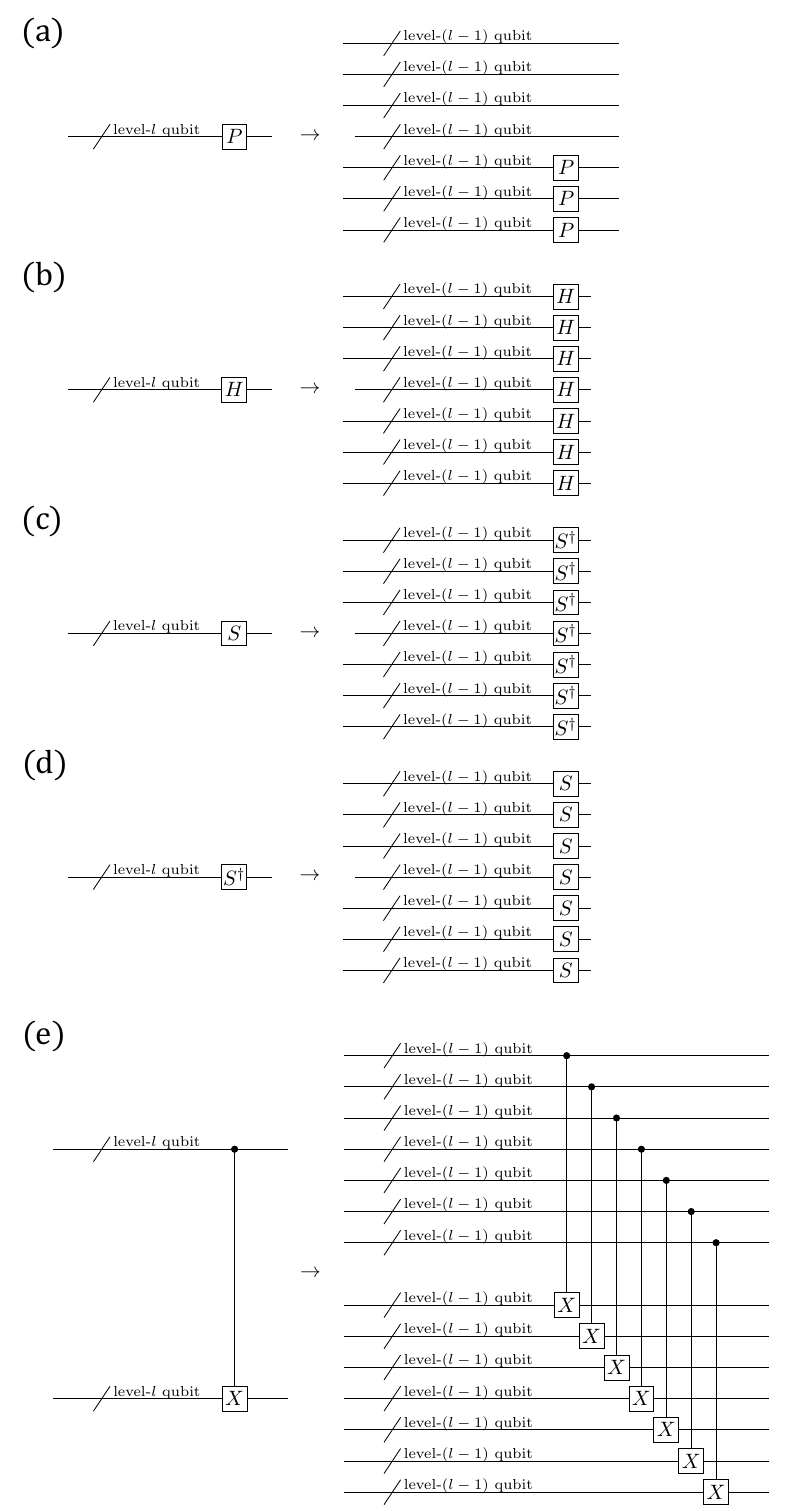}
    \caption{A level-$(l-1)$ circuit of the Pauli-gate gadget (a), the $H$-gate gadget (b), the $S$-gate gadget (c), the $S^{\dagger}$-gate gadget (d), and the CNOT-gate gadget (e).}
    \label{Fig: gate gadgets.}
\end{figure}

\subsubsection{$\ket{0}$-state preparation gadget}

The construction of the $\ket{0}$-state preparation gadget is shown in Fig.~\ref{Fig: 0-state preparation gadget.}~(a).
We first prepare $7$ qubits in the state $\ket{0}^{\otimes 7}$ using the $\ket{0}$-state preparations.
Then, we use an encoding circuit $U_{\mathrm{encode}}^{\ket{0}}$ to transform the state $\ket{0}^{\otimes 7}$ into the logical state of $\ket{0}$.
The encoding circuit $U_{\mathrm{encode}}^{\ket{0}}$ is shown in Fig.~\ref{Fig: 0-state preparation gadget.}~(b), presented in Ref.~\cite{Minimizing_Goto}.

The issue here is that the encoding circuit is non-fault-tolerant, meaning a single fault occurring during the encoding circuit may lead to an error with a weight greater than one by the end of the encoding circuit.
Thus, verification is essential to detect such errors, ensuring that the $\ket{0}$-state preparation gadget satisfies the fault-tolerance conditions specified in~\eqref{eq: Prep A} and~\eqref{eq: Prep B}.
In the verification process, the gadget executes the same encoding circuit $U_{\mathrm{encode}}^{\ket{0}}$ to prepare another set of $7$ qubits in the logical state of $\ket{0}$.
Subsequently, the gadget performs the transversal CNOT gates targeting the newly prepared qubits.
This is followed by the transversal $Z$-basis measurements on these qubits. 
From these $Z$-basis measurements, we obtain the syndrome bits for all $Z$-type stabilizer generators in~\eqref{eq: stabilizer generator for Steane code.} and the measurement outcome of the logical $Z$ operator on the level-$l$ qubit in~\eqref{eq: logical opeartor for Steane code} through classical computation described in the $Z$-basis measurement gadget.
These are crucial for detecting badly propagated $X$-type errors.
The verification passes if and only if all the syndrome bits are trivial and the measurement outcome $m$ of the logical $Z$ operator is $m=0$.
If verification fails, the gadget discards the resulting state and repeats the encoding circuit $U_{\mathrm{encode}}^{\ket{0}}$ without verification in the second run.
Note that the detection of the $Z$-type errors is unnecessary for verifying $\ket{0}$ since multiple $Z$-type errors at the end of the encoding circuit may lead to a logical $Z$ error but do not transform the state, i.e., $Z\ket{0}=\ket{0}$.

We check that the above gadget construction satisfies the fault-tolerance condition defined in~\eqref{eq: Prep A} and~\eqref{eq: Prep B}.
In the case where $s=0$, i.e., no fault occurs during the gadget, the gadget successfully outputs $7$ qubits in the logical state of $\ket{0}$, and thus the conditions~\eqref{eq: Prep A} and \eqref{eq: Prep B} are satisfied.
Next, we discuss the case where $s=1$, i.e., only a single fault occurs in the gadget.
If the verification succeeds, the gadget successfully outputs $7$ qubits in the logical state of $\ket{0}$ in the first run.
If the verification fails in the first run, i.e., an error is detected, then the gadget discards the resulting state.
Then, the gadget restarts and executes the encoding circuit to prepare $7$ qubits in the logical state of $\ket{0}$ in the second run without verification.
In this scenario, since a single fault occurred in the first run, no fault occurred during the part for preparing $7$ qubits in the logical state of $\ket{0}$ in the second run, given that we consider the condition where $s=1$.
Therefore, the fault-tolerance conditions~\eqref{eq: Prep A} and \eqref{eq: Prep B} of the gadget are satisfied.

\begin{figure*}[t]
    \centering
    \includegraphics[width=\textwidth]{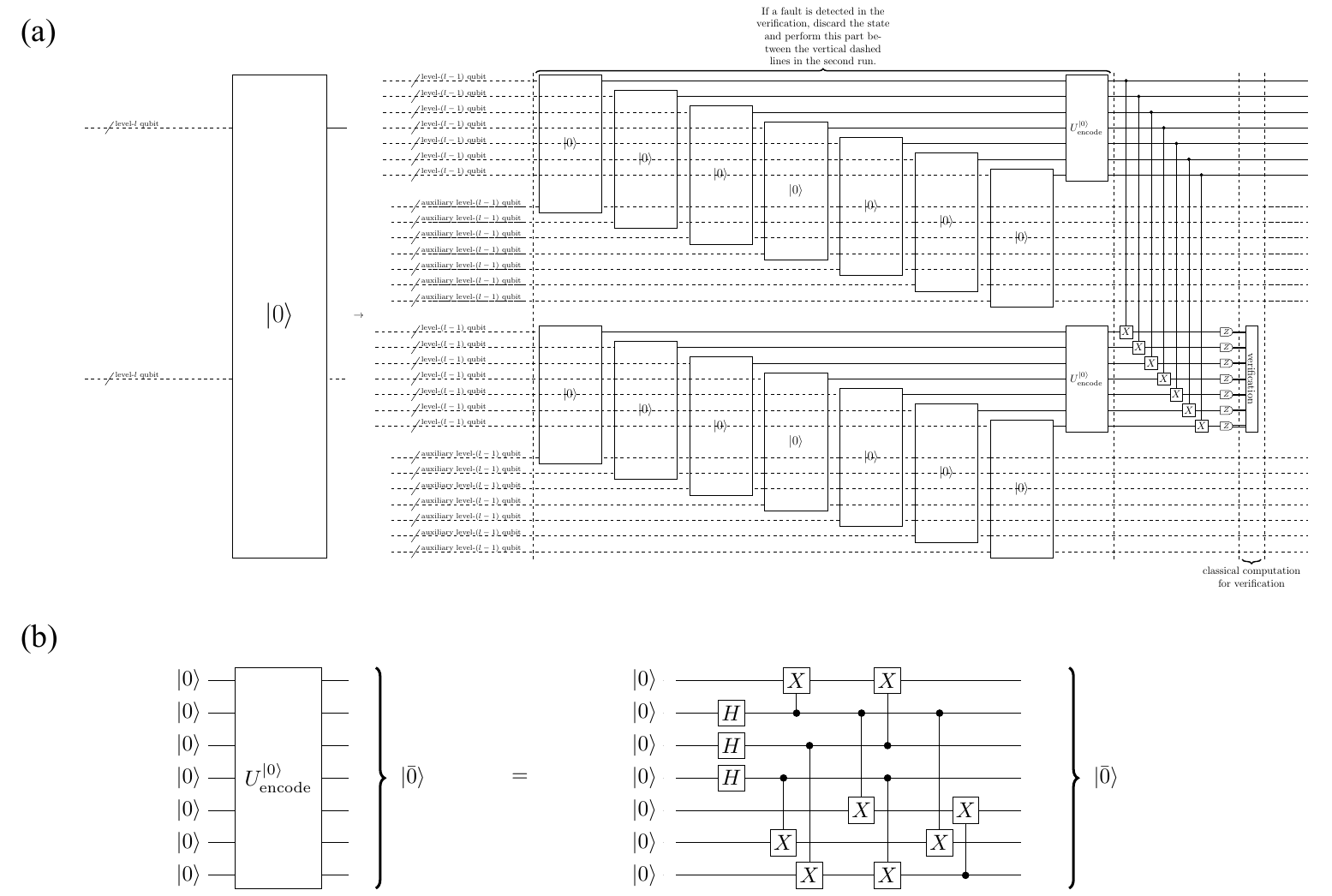}
    \caption{(a) The level-$(l-1)$ circuit of the $\ket{0}$-state preparation gadget.
The gadget starts with the preparation of $\ket{0}^{\otimes 7}$, followed by an encoding circuit $U_{\mathrm{encode}}^{\ket{0}}$, shown in (b), to prepare the logical state $\ket{\Bar{0}}$ of the Steane code from $\ket{0}^{\otimes 7}$.
The encoding circuit $U_{\mathrm{encode}}^{\ket{0}}$ shown in (b) is presented in Ref.~\cite{Minimizing_Goto}.
However, since the encoding circuit is not fault-tolerant, verification is required for the $\ket{0}$-state preparation gadget to satisfy the fault-tolerance conditions in \eqref{eq: Prep A} and \eqref{eq: Prep B}.
In the verification process, the same encoding circuit $U_{\mathrm{encode}}^{\ket{0}}$ is performed to prepare another logical state $\ket{\Bar{0}}$. Next, the gadget performs transversal CNOT gates between the two logical states, followed by transversal $Z$-basis measurements.
These measurements provide syndrome bits of all $Z$-type stabilizers and a measurement outcome of the logical $Z$ operator.
If no error is detected, the verification passes and the gadget outputs the resulting state.
If an error is detected, the gadget discards the state and re-executes the encoding circuit in the second run without verification.
}
    \label{Fig: 0-state preparation gadget.}
\end{figure*}

\subsubsection{$\ket{T}$-state preparation gadget}
\begin{figure*}[p]
    \includegraphics[angle=270,width=0.35\textheight]{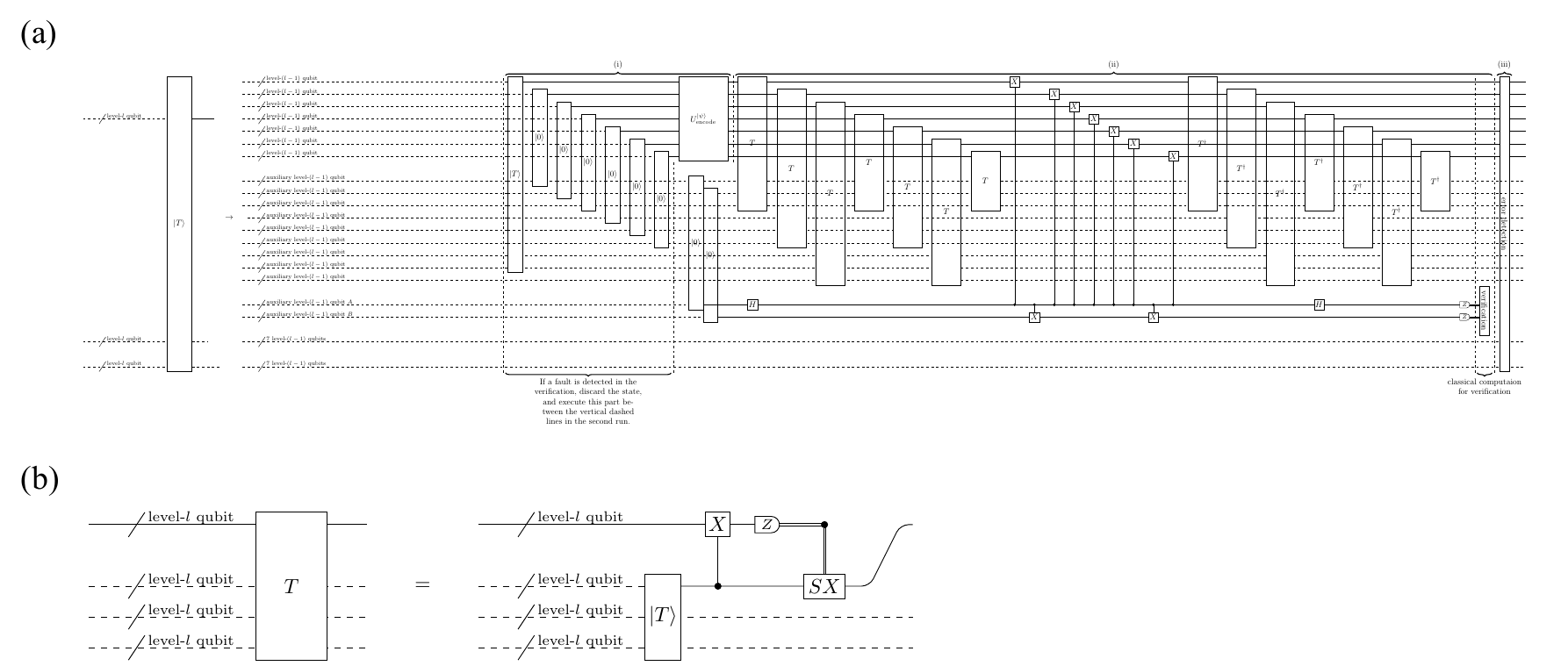}
    \caption{(a) The level-$(l-1)$ circuit of the $\ket{T}$-state preparation gadget. (b) The abbreviation for a sequence of elementary operations used to perform a $T$ gate on a level-$l$ qubit via gate teleportation.
The gadget starts with the preparation of the state $\ket{T}\otimes\ket{0}^{\otimes 6}$, followed by an encoding circuit $U_{\mathrm{encode}}^{\ket{\psi}}$, shown in Fig.~\ref{Fig: decoding gadget.} (b), to prepare the logical state $\ket{\Bar{T}}$ of the Steane code from $\ket{T}\otimes\ket{0}^{\otimes 6}$.
However, since the encoding circuit is not fault-tolerant, verification is required for the $\ket{T}$-state preparation gadget to satisfy the fault-tolerance conditions given in \eqref{eq: Prep A T} and \eqref{eq: Prep B T}.
In the verification process, to measure the logical operator $TXT^{\dagger}$, we perform measurements of the $T^{\dagger}XT$ operator, controlled by the auxiliary qubit $A$, and then measure the auxiliary qubit $A$.
To measure the logical operator $TXT^{\dagger}$ in a fault-tolerant manner, we additionally prepare the $\ket{0}$ state in the auxiliary qubit $B$ and measure the auxiliary qubit $B$.
If no error is detected through the measurement outcomes, the verification passes, and the gadget outputs the resulting state.
If an error is detected, the gadget discards the state and re-executes the encoding circuit in the second run without verification.}
    \label{Fig: t-state-preparation-gadget-knill-steane.}
\end{figure*}

A $\ket{T}$-state preparation gadget is constructed as shown in Fig.~\ref{Fig: t-state-preparation-gadget-knill-steane.}.
The gadget starts with the preparation of the logical state of $\ket{T}\coloneqq TH\ket{0}$ from the $7$ qubits in the state $\ket{T}\otimes\ket{0}^{\otimes 6}$.
We utilize the non-fault-tolerant encoding circuit $U_{\mathrm{encode}}^{\ket{\psi}}$ to transform an arbitrary state $\ket{\psi}$ into an logical state of $\ket{\psi}$, as shown in Fig.~\ref{Fig: decoding gadget.}~(b), presented in Ref.~\cite{Minimizing_Goto}.
A single fault that occurs during or before the encoding circuit $U_{\mathrm{encode}}^{\ket{\psi}}$ leads to an error with a weight greater than one, as well as for the $\ket{0}$-state preparation gadget.
Thus, verification is essential to detect such errors for the $\ket{T}$-state preparation gadget to satisfy the fault-tolerance conditions given by~\eqref{eq: Prep A T} and~\eqref{eq: Prep B T}.
For verification, we perform a measurement of the operator $TXT^{\dagger}$, followed by error detection.
The measurement of the operator $TXT^{\dagger}$ 
is needed for verification, since the $\ket{T}$ state we want to prepare is stabilized by the operator $TXT^{\dagger}$, i.e., 
\begin{equation}
    TXT^{\dagger}\ket{T}=\ket{T}
\end{equation}
Since $TXT^{\dagger}=SX$ is a Clifford operator, we can measure the logical operator $TXT^{\dagger}$ by performing transversal measurements of $T^{\dagger}XT=S^{\dagger}X$.
This measurement can be performed with controlled-$T^{\dagger}XT$ gates controlled by the auxiliary qubit $A$ as shown in Fig.~\ref{Fig: t-state-preparation-gadget-knill-steane.}.
The auxiliary qubit is measured to obtain the measurement outcome of the operator $T^{\dagger}XT$.
If the measurement outcome of the operator $TXT^{\dagger}$ is $m=1$, indicating that the state is projected onto the state orthogonal to $\ket{T}$, then we consider the verification to have failed, and the gadget discards the state.
To measure the logical operator $TXT^{\dagger}$ in a fault-tolerant way, we additionally use a qubit as the flag qubit~\cite{Chao_flag,Chamberland_2019_flag}, where the flag qubit corresponds to the auxiliary qubit $B$ in Fig.~\ref{Fig: t-state-preparation-gadget-knill-steane.}.
A single fault on the auxiliary qubit $A$ may result in an error with a weight greater than one on the code block, which consists of the $7$ qubits. 
However, in such a case, if the measurement outcome of the auxiliary qubit $B$, serving as the flag qubit, is $m=1$ in the $Z$ basis, the verification is deemed to have failed, and the gadget discards the state.
To verify that the state is in the code subspace, we perform error detection at the end of the $\ket{T}$-state preparation gadget based on the gate teleportation protocol as shown in Fig.~\ref{Fig: t-state-preparation-gadget-knill-steane.}.
If at least one non-trivial syndrome bit is detected during the execution of the gate teleport protocol, then we consider the verification to be failed, and the gadget discards the state.
Therefore, the verification is successful only if the measurement outcome $m$ of $TXT^{\dagger}$ is $m=0$, the measurement outcome $m$ of the flag qubit is $m=0$, and all the syndrome bits obtained during the gate teleportation protocol are trivial.
If the verification fails, the gadget discards the resulting state at the point of failure.
Afterwards, the gadget repeats the encoding circuit to prepare the logical state $\ket{T}$ without verification in the second run.

We show that the $\ket{T}$-state preparation gadget satisfies the fault-tolerance condition defined by~\eqref{eq: Prep A T} and~\eqref{eq: Prep B T}.
If $s=0$, indicating no fault in the gadget, the gadget successfully outputs the logical state $\ket{T}$, satisfying conditions~\eqref{eq: Prep A T} and~\eqref{eq: Prep B T} are satisfied.
Next, we consider the case where $s=1$, indicating a single fault in the gadget, and show that condition~\eqref{eq: Prep B T} is satisfied.
We divide the $\ket{T}$-state preparation gadget, which may cause a single fault, into three parts: (1) preparing the logical state of $\ket{T}$, (2) measuring the logical operator $TXT^{\dagger}$, and (3) error detection.
Here, we show that the fault-tolerance condition is satisfied if a single fault occurs in any of the three parts.
\begin{enumerate}
    \item A fault occurs during the preparation of the logical state $\ket{T}$.
    
    In this scenario, since we assume that only a single fault occurs during the gadget, the measurement of the operator $TXT^{\dagger}$ and the error detection are performed ideally.
    If the verification succeeds, the gadget successfully outputs the logical state $\ket{T}$.
    If the verification fails, the gadget discards the state obtained in the first run and then re-executes the encoding circuit in the second run. 
    In this second run, the gadget outputs the logical state $\ket{T}$ without further verification because a single fault occurred in the first run. 
    After completing the encoding circuit in the second run, the gadget successfully outputs the logical state $\ket{T}$. 
    With this gadget construction, the condition \eqref{eq: Prep B T} is satisfied.
    
    \item A fault occurs during the measurement of the logical operator $TXT^{\dagger}$.
    
    In this case, since there is no fault in the first part (i) of the gadget, the encoding circuit outputs the logical state $\ket{T}$.
    The problem arises when a single fault occurs in the auxiliary qubit $A$, as the fault-induced error may propagate to the code block of $7$ qubits, resulting in an error with a weight greater than one.
    However, even if such a fault occurs, the error resulting from the fault is also propagated to the auxiliary qubit $B$, causing a measurement outcome of $m=1$, which can be detected through verification.
    For example, if a fault occurs on the auxiliary qubit $A$, resulting in a Pauli-$X$ or -$Y$ error, this error affects the code block, leading to an error with multiple weights.
    Subsequently, this error propagates to the auxiliary qubit $B$ through the CNOT gates, resulting in a measurement outcome of $m=1$ on qubit $B$.
    Therefore, a Pauli error that produces an error with multiple weights on the code block can be effectively detected through the verification process.    
    If the verification process fails, the gadget discards the state obtained from the first run and performs the encoding circuit again in the second run without verification.
    This construction of the gadget ensures that the condition \eqref{eq: Prep B T} is satisfied.

    \item A fault occurs during error detection.
    
    Since there are no faults in (i) and (ii) of the gadget, the state before executing the error detection is a logical state of $\ket{T}$.
    If a single fault occurs that can be detected by the error detection procedure, the gadget discards the state. 
    In the second run, it executes the encoding circuit again, satisfying the condition \eqref{eq: Prep B T}.
    Even if a single fault occurs that cannot be detected through error detection, from~\eqref{eq: ecB2}, the condition \eqref{eq: Prep B T} is also satisfied.
\end{enumerate}
Therefore, the fault tolerance of the $\ket{T}$-state preparation gadget defined as \eqref{eq: Prep A T} and \eqref{eq: Prep B T} is satisfied.

\subsubsection{Error-correction gadget}

A level-$(l-1)$ circuit of an error-correction gadget is shown in Fig.~\ref{Fig: knill-error-correction-steane-simple.}.
This gadget is based on Knill's error correction~\cite{knill2005scalable,Knill_2005, gottesman2010introduction}.
The gadget uses two sets of $7$ auxiliary qubits initialized to a logical state of $\ket{0}$ using the $\ket{0}$-state preparation gadget to perform error correction on an input logical state of $7$ qubits.
Using these two sets of $7$ auxiliary qubits, we perform quantum teleportation of the input logical state.
The two $7$-bit outcomes obtained from transvarsal $Z$-basis measurements are fed to the decoding algorithm to calculate measurement outcomes of the logical $Z$ operator of the Steane code, which is performed in the same way explained in the $Z$-basis measurement gadget in Sec.~\ref{sec: Z-basis measurement gadget steane}. 
The logical measurement outcomes are used to determine the Pauli correction operation for quantum teleportation.
The fault tolerance of the EC gadget defined by conditions~\eqref{eq: ec A},~\eqref{eq: ecB1}, and~\eqref{eq: ecB2} follows from transversality.

\begin{figure*}[t]
    \centering
    \includegraphics[width=\textwidth]{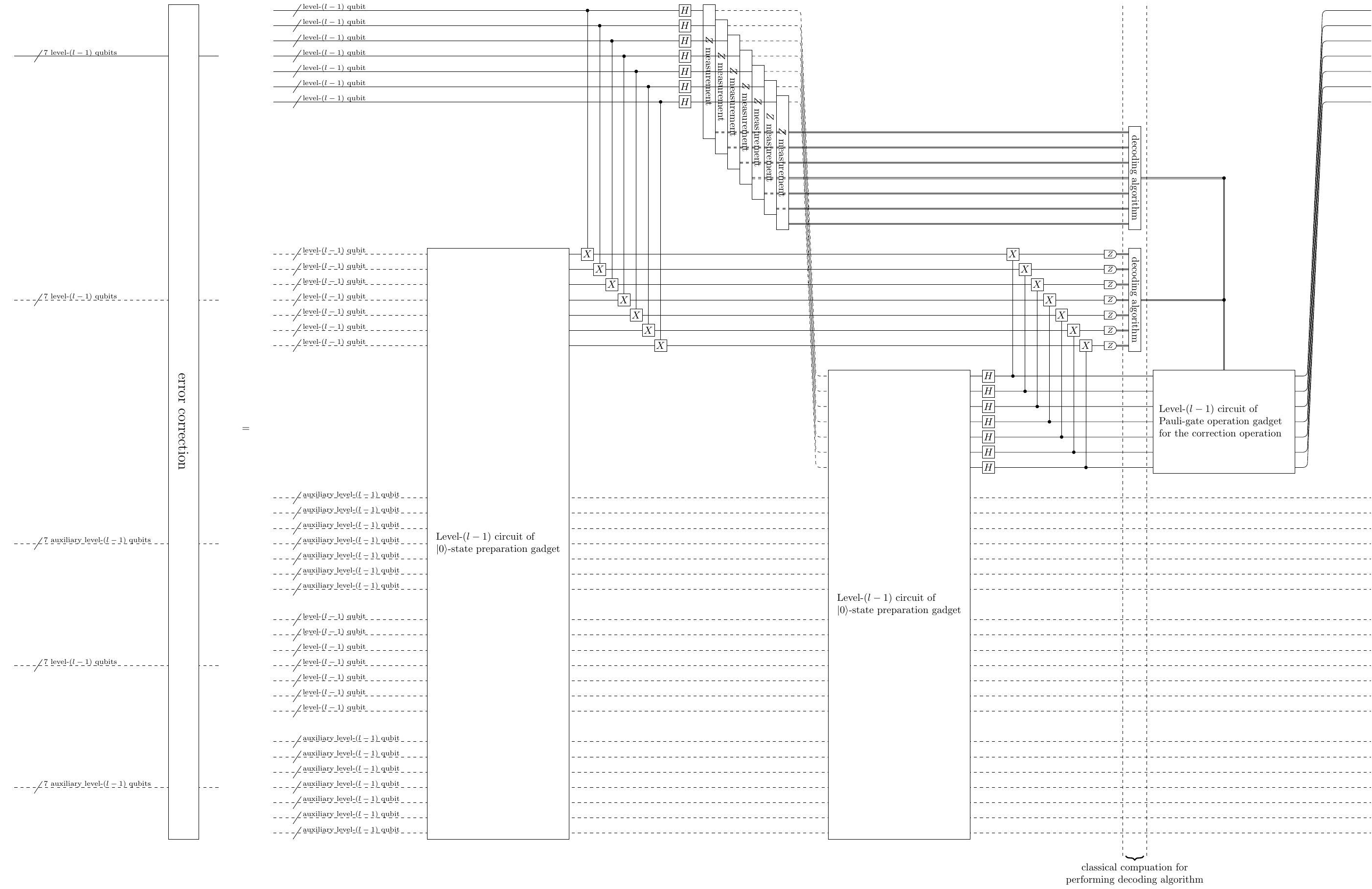}
    \caption{The level-$(l-1)$ circuit of the EC gadget based on Knill's error correction is shown. The gadget uses two sets of $7$ auxiliary qubits, each initialized to a logical state of $\ket{0}$ prepared by the $\ket{0}$-state preparation gadget in \eqref{Fig: 0-state-preparation}, to perform error correction on the input logical state of $7$ qubits.
    The transversal $Z$-basis measurements provide two $7$-bit outcomes.
    From these measurement outcomes, we calculate measurement outcomes of the logical $Z$ operator of the Steane code.
    This computation can be performed in the same way as explained in the $Z$-basis measurement gadget in \ref{sec: Z-basis measurement gadget steane}.
    Based on the logical measurement outcomes, a Pauli correction operation is applied for quantum teleportation.}
    \label{Fig: knill-error-correction-steane-simple.}
\end{figure*}

\subsubsection{Decoding gadget \label{subsubsec: decoding gadget}}

\begin{figure*}[t]
\centering
\includegraphics[width=\textwidth]{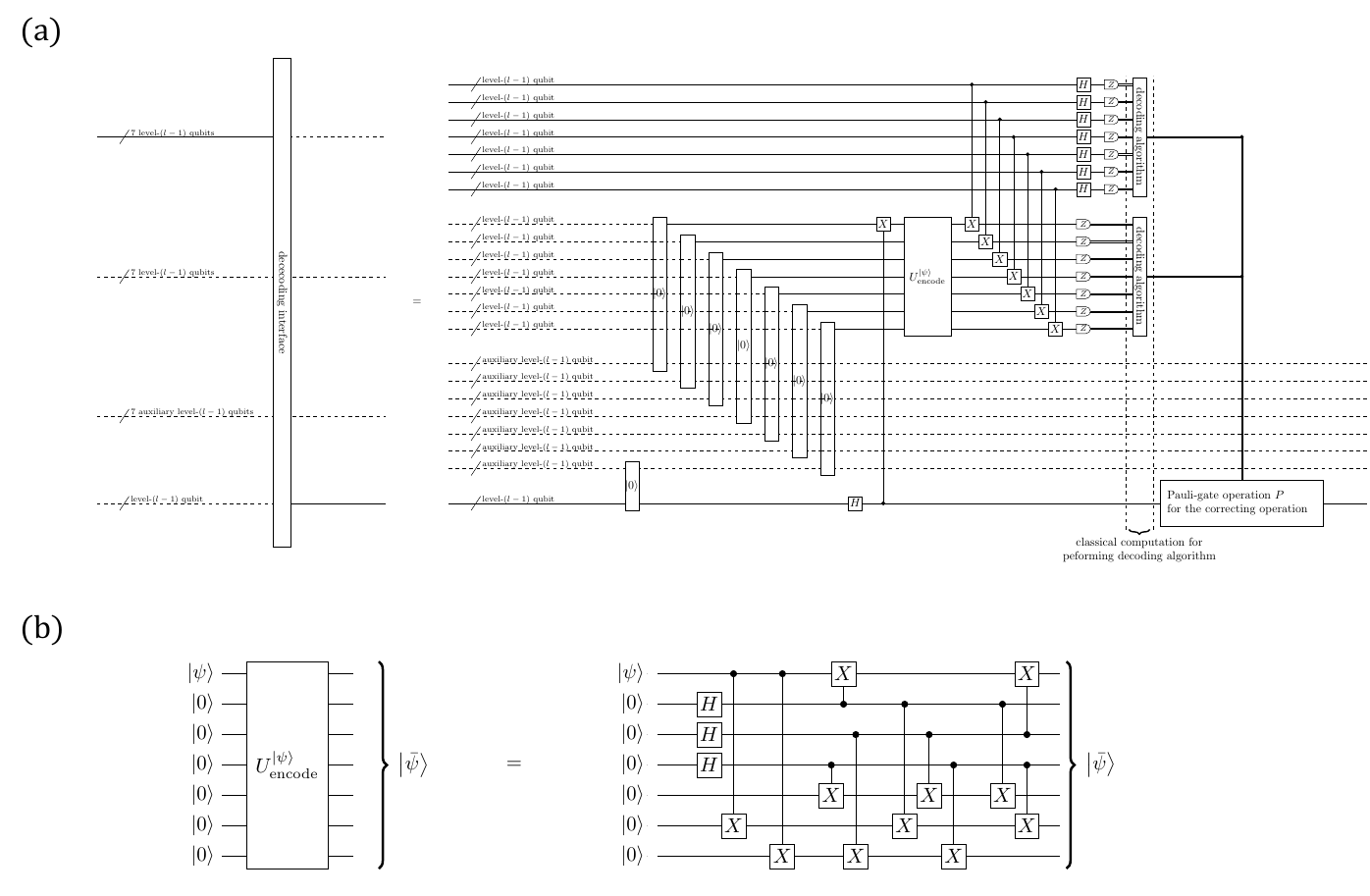}
\caption{(a) The level-$(l-1)$ circuit of the decoding gadget for transforming a logical state encoded in the Steane code to the same state of an unencoded physical qubit. The gadget uses an input logical state and auxiliary qubits prepared in a Bell state, with one qubit encoded in the logical state of the Steane code using the encoding circuit $U_{\mathrm{encode}}^{\ket{\psi}}$ shown in (b). The decoding process involves transversal CNOT gates, $H$ gates, and $Z$-basis measurements. The measurement outcomes are processed classically to determine the Pauli correction operation for successful teleportation of the desired state from the logical qubit to the physical qubit.
(b) The encoding circuit $U_{\mathrm{encode}}^{\ket{\psi}}$ for preparing a Bell state with one qubit encoded in the logical state of the Steane code, which is presented in Ref.~\cite{Minimizing_Goto}}
\label{Fig: decoding gadget.}
\end{figure*}

A decoding gadget is designed to transform a logical state encoded in the Steane code, that is a set of seven qubits, into the same state of an unencoded physical qubit. The construction of the decoding gadget based on quantum teleportation is shown in Fig.~\ref{Fig: decoding gadget.} (a).

In addition to the input logical state that we want to decode, the gadget uses auxiliary qubits prepared in a Bell state. One qubit of the Bell state is encoded in the logical state of the Steane code using an encoding circuit $U_{\mathrm{encode}}^{\ket{\psi}}$, as shown in Fig.~\ref{Fig: decoding gadget.} (b).
The decoding gadget is implemented by applying transversal CNOT gates, followed by transversal $H$ gates, and then transversal $Z$-basis measurements. The $Z$-basis measurements yield a pair of $7$-bit strings as measurement outcomes.
Based on these measurement outcomes, we perform the classical computation to calculate the measurement outcomes of the logical $\Bar{Z}$ operator of the Steane code. This classical computation is performed using the same procedure as explained in the $Z$-basis measurement gadget in Appendix \ref{sec: Z-basis measurement gadget steane}.
The calculated measurement outcomes are used to determine the Pauli correction operation required for quantum teleportation. Once the Pauli correction operation is applied, the gadget outputs the desired state of the qubit, which has been successfully teleported from the logical state to the physical qubit.
Therefore, the decoding gadget satisfies the validity condition in \eqref{eq: decA}.

\subsection{Threshold theorem for open circuits}
\label{sec: Threshold theorem for open circuits}
In this section, we prove the threshold theorem for simulating open circuits.
The threshold theorem for closed original circuits has been established in Refs.~\cite{aliferis2005quantum, gottesman2010introduction}.
However, for open circuits, a different treatment is required because the decoding gadgets are located at the end of level-$0$ circuits. 
The threshold theorem for open circuits described in the following is used without explicit proof in Refs.~\cite{Gottesman2014Constant, Fawzi_2018, grospellier:tel-03364419}.
Although Ref.~\cite{christandl2022faulttolerant} also presents the threshold theorem for open circuits, the noise model considered is the independent and ideally distributed (IID) Pauli error model.
Therefore, the theorem in Ref.~\cite{christandl2022faulttolerant} is not applicable to our more general setting.
For this reason, we provide a complete proof of the theorem based on our protocol and noise model, which is essential for the full proof of the threshold theorem in Theorem~\ref{theorem: threshold theorem for polylog-time- and constant-space-overhead fault-tolerant quantum computation}. 
While the following argument can apply to any open circuits, for consistency with the proof in Theorem~\ref{theorem: threshold theorem for polylog-time- and constant-space-overhead fault-tolerant quantum computation}, we will focus on open circuits used to prepare logical states such as logical $\ket{0}^{\otimes K}$ state, logical Clifford states, and logical magic states of $\mathcal{Q}$, which is a CSS LDPC code.

\begin{theorem}[Threshold theorem for simulating open circuits]\label{Theorem: level-reduction for the circuit that outputs a quantum state}

There exist positive constants 
$\gamma_1$, $\gamma_2$, $c_1$, $c_2$, $c_3$, $p_{\mathrm{loc}}^{\mathrm{th}}$, and $M$
such that 
the circuit  $\tilde{C}(C, L)$ satisfies the following properties for any open circuit $C$ consisting of the elementary operations of \eqref{Fig: 0-state-preparation}-\eqref{Fig: uc-gate} and for any positive integer $L$ representing the concatenation level:

\begin{enumerate}
        \item $W(\tilde{C}(C, L)) \le  c_1W(C) 2^{\gamma_1 L}$.

        \item $D(\tilde{C}(C, L)) \le  c_2 D(C) 2^{\gamma_2 L}$.
        
        \item For any number $p_\loc \in [0, p_{\mathrm{loc}}^{\mathrm{th}})$, the $\delta_3$-reducibility in Def.~\ref{def: delta-reducibility} given by the following diagram holds:
\begin{align}
    \includegraphics[width=0.5\textwidth]{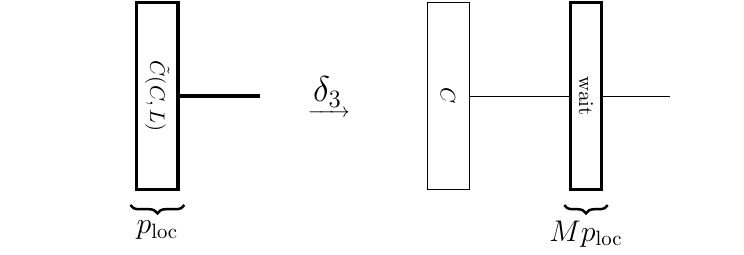}\label{Fig: delta-reducibility},
\end{align}
where 
        \begin{align} 
 \delta_3=c_3|C| {p}^{\mathrm{th}}_{\mathrm{loc}}\left(\frac{p_{\mathrm{loc}}}{{p}^{\mathrm{th}}_{\mathrm{loc}}}\right)^{2^L}.
        \end{align}
    \end{enumerate}
\end{theorem}

\begin{proof}
For a given original open circuit $C$ and the concatenation level $L$, we can obtain the level-$0$ circuit $\tilde{C}(C, L) = C^{(0)}$ that simulates $C$, as described in Sec.~\ref{sec: Compilation from original open circuit to fault-tolerant circuit}.

At each concatenation level $l$, a level-$l$ qubit is replaced with a set of seven level-$(l-1)$ qubits.
Additionally, each level-$l$ qubit is allocated with an additional constant number of auxiliary level-$(l-1)$ qubit per level-$l$ qubit.
Let $w_1$ denote the maximum number of level-$(l-1)$ qubits needed per level-$l$ qubit at each level, which is also a constant.
The width at level $l\in [L,\ldots,1]$ satisfies the recursive relation as
\begin{equation}
    W^{(l-1)}\leq w_1 W^{(l)}.
\end{equation}
Since there exists a positive constant $c_1$ such that $W^{(L)}\leq c_1 W(C)$, we obtain 
\begin{equation}
    W(\tilde{C}(C, L)) = W^{(0)}\leq W^{(L)}\qty(w_1)^L\leq c_1W(C)\qty(w_1)^L,
\end{equation}
which leads to
\begin{align}
    W(\tilde{C}(C, L))\leq c_1W(C)2^{\gamma L},
\end{align}
where $\gamma_1=\log_2 w_1$.

Similarly, the depth increases at each level due to the replacement of gadgets.
Let $d_1$ denote the maximum depth increase per elementary operation at each level-$l$, including the depth of EC gadgets.
The depth at level-$l$ satisfies
\begin{equation}
    D^{(l-1)}\leq d_1 D^{(l)}+d_2,
\end{equation}
where $d_2$ accounts for additional depth due to the decoding gadget.
For simplicity, we can bound $D^{(l-1)}\leq w_2 D^{(l)}$ with $w_2=d_1+d_2$.
Since there exists a positive constant $c_2$ such that $D^{(L)}\leq c_2D(C)$, we obtain
\begin{equation}
    D(\tilde{C}(C, L)) = D^{(0)}\leq D^{(L)}\qty(w_2)^L\leq c_2D(C)\qty(w_2)^L,
\end{equation}
which leads to 
\begin{align}
    D(\tilde{C}(C, L))\leq c_2D(C)2^{\gamma_2 L},
\end{align}
where $\gamma_2=\log_2 w_2$.

Next, we consider the level reduction for $C^{(0)}$ illustrated in Fig.~\ref{fig:level_reduction}.
The proof of the threshold theorem for closed circuits in Ref.~\cite{gottesman2010introduction,aliferis2005quantum,Yamasaki_2024} is accomplished by moving ideal decoders from $Z$-basis measurement gadgets located at the end of a level-$l$ circuit $C^{(l)}$ to the front, to obtain a level-$(l+1)$ circuit  $C^{(l+1)}$ implemented at the logical level of the level-$l$ circuit for $l\in\{0,1,\ldots,L-1\}$.
By repeatedly applying this procedure until we obtain a level-$l$ circuit $C^{(L)}$ on level-$L$ qubits, we effectively obtain the level-$L$ circuit $C^{(L)}$ implemented at the logical level with lower error parameters.

In contrast to the protocol for closed circuits, which provides a circuit ending with the $Z$-basis measurement gadgets, to implement the level-$(l+1)$ circuit, our protocol for open circuits compiles it into a level-$l$ circuit $C^{(l)}$ that ends with decoding gadgets.
Then, a single fault at any location in the decoding gadgets of $C^{(l)}$ could result in an error on the unencoded state.
However, since the decoding gadget is a stabilizer circuit (i.e., composed only of Clifford gates and $Z$-basis measurements)~\cite{Aaronson_2004}, each faulty decoding gadget can be transformed into a non-faulty decoding gadget by applying the propagated Pauli errors at the final time step.
To represent these propagated errors, we add a single-depth circuit composed only of wait locations after the final time step of the decoding gadgets in $C^{(l)}$.
After the error propagation, due to the union bound, the error parameter for each wait location placed after the final time step of $C^{(l)}$ is bounded by, due to the union bound,
\begin{equation}
M'p_{\mathrm{loc}}^{(l)}, 
\end{equation}
where $M'$ is a constant representing the number of locations in the decoding gadget.
Therefore, the decoding gadget, denoted by Dec, can be transformed as
\begin{widetext}
\begin{align}
    \includegraphics{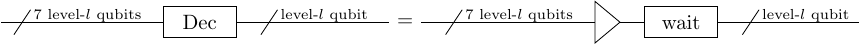}\label{Fig: transform the decoding gadget},
\end{align}
\end{widetext}
where the wait location has the error parameter bounded by $M'p_{\mathrm{loc}}^{(l)}$.

Using the ideal decoder derived from the decoding gadget by~\eqref{Fig: transform the decoding gadget}, the rest of the level reduction procedure for the level-$l$ circuit can be carried out in the same way as in the conventional procedure of level reduction for the protocols with the concatenated codes in Refs.~\cite{gottesman2010introduction,aliferis2005quantum,Yamasaki_2024}, where the ideal decoder is moved from the end to the beginning of the level-$l$ circuit to obtain the level-$(l+1)$ circuit implemented at the logical level.
Due to the argument in Ref.~\cite{gottesman2010introduction}, if $C^{(l)}$ experiences a local stochastic Pauli error with parameter $p_{\mathrm{loc}}^{(l)}$, $C^{(l+1)}$ also experiences a local stochastic Pauli error with parameter
\begin{equation}
\label{eq:p_loc_l}
    p_{\mathrm{loc}}^{(l+1)}\leq A\left(p_{\mathrm{loc}}^{(l)}\right)^2,
\end{equation}
where $A$ is a constant representing the number of pairs of locations in the largest ExRec\@.
Each additional wait location after the final time step of $C^{(l)}$ have the error parameters bounded by $M'p_{\mathrm{loc}}^{(l)}$ as shown in~\eqref{Fig: transform the decoding gadget}.
As a result, by repeating this level-reduction procedure $L$ times, the error parameters of the level-$L$ circuit are determined as shown in Fig.~\ref{fig:level_reduction}.

\begin{figure*}
    \centering
    \includegraphics[width=7.0in]{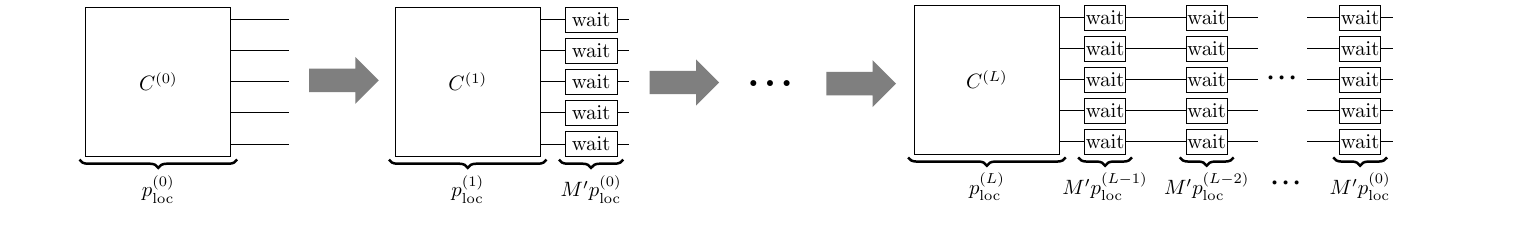}
    \caption{
    The procedure of level reduction. 
Under the compilation procedure presented in Sec.\ref{sec: Compilation from original open circuit to fault-tolerant circuit}, an original open circuit $C$ is compiled to a physical circuit $\tilde{C}(C,\delta_3) = C^{(0)}$. For each $l\in\{0,1,\ldots,L-1\}$, due to the validity condition in \eqref{eq: decA},  each decoding gadget at the end of $C^{(l)}$ can be transformed into an ideal decoder, followed by a wait location with an error parameter $M'p_{\mathrm{loc}}^{(l)}$, where $M'$ is the number of locations in the decoding gadget, as shown in~\eqref{Fig: transform the decoding gadget}.
Following Ref.~\cite{gottesman2010introduction}, we move the ideal decoders backwards in $C^{(l)}$ to provide a level-$(l+1)$ circuit $C^{(l+1)}$ implemented at the logical level of $C^{(l)}$ with each location having error parameters bounded by $p_{\mathrm{loc}}^{(l+1)}$, followed by the wait locations in~\eqref{Fig: transform the decoding gadget}. Repeating this procedure $L$ times, we finally obtain the level-$L$ circuit $C^{(L)}$ (i.e., a circuit equivalent to $C$) implemented at the logical level of $\tilde{C}(C,\delta_3)$, followed by a depth-$L$ circuit of wait locations.}

    \label{fig:level_reduction}
\end{figure*}
    
In this way, we obtain the level-$L$ circuit composed of $C^{(L)}$ implemented at the logical level, which is equivalent to $C$, followed by a depth-$L$ circuit of wait locations arising from obtaining the ideal decoder from the level-$l$ decoding gadgets by~\eqref{Fig: transform the decoding gadget} for $l\in\{1,2,\ldots,L\}$.
Due to~\eqref{eq:p_loc_l}, each location $j$ in $C^{(L)}$ has an error parameter 
\begin{equation}
    q_j=p_{\mathrm{loc}}^{(L)}\leq {p}^{\prime\mathrm{th}}_{\mathrm{loc}}\left(\frac{p^{(0)}_{\mathrm{loc}}}{{p}^{\prime\mathrm{th}}_{\mathrm{loc}}}\right)^{2^L},
\label{eq: doubly exponential supression}
\end{equation}
where ${p}^{\prime\mathrm{th}}_{\mathrm{loc}}\coloneqq 1/A>0$.
In addition, the depth-$L$ circuit of wait locations is equivalent to a single-depth circuit  $C^{\mathrm{wait}}$ composed only of wait locations, where each error parameter $q_j$ of wait location $j$ in $C^{\mathrm{wait}}$ is upper bounded, due to the union bound, by
\begin{align}
    q_j&\leq M'\sum_{l=0}^{L-1}p_{\mathrm{loc}}^{(l)}\\
    &\leq M'\sum_{l=0}^{\infty}p_{\mathrm{loc}}^{(l)}\\
    &\leq M'{p}^{\prime\mathrm{th}}_{\mathrm{loc}}\sum_{k=1}^{\infty}\left(\frac{p^{(0)}_{\mathrm{loc}}}{{p}^{\prime\mathrm{th}}_{\mathrm{loc}}}\right)^{k}\\
    &=\frac{M'p_{\mathrm{loc}}^{(0)}}{1-p_{\mathrm{loc}}^{(0)}/{p}_{\mathrm{loc}}^{\prime\mathrm{th}}}\\
    &< Mp_{\mathrm{loc}},
\end{align}
where we take $M\coloneqq 2M'$ and $p_{\mathrm{loc}}\coloneqq p^{(0)}_{\mathrm{loc}}$ satisfying
\begin{align}
0< p_{\mathrm{loc}}< {p}_{\mathrm{loc}}^{\prime\mathrm{th}}/2\eqqcolon p_{\mathrm{loc}}^{\mathrm{th}}.
\end{align}
Thus, we obtain the circuit
\begin{align}
    C^{(L)}\cup C^{\mathrm{wait}}.
\end{align}

At this stage, we see that a pair $(\tilde{C}(C,\delta_3), \{p_i\}_{i\in\tilde{C}(C,\delta_3)})$ where $p_i\leq p_{\mathrm{loc}}$ for all $i\in\tilde{C}(C,\delta_3)$ is reducible to another pair $\qty(C^{(L)}\cup C^{\mathrm{wait}}, \{q_i\}_{i\in C^{(L)}\cup C^{\mathrm{wait}}})$, where
\begin{align}
    q_j&= p_{\mathrm{loc}}^{(L)}&&\text{if $j \in C^{(L)}$},\\
    q_j&\leq Mp_{\mathrm{loc}}&&\text{if $j \in C^{\mathrm{wait}}$}.
\end{align}
On the other hand, the case of having errors in $C^{(L)}$ can be considered as the case where the pair $(\tilde{C}(C,\delta_3), \{p_i\}_{i\in\tilde{C}(C,\delta_3)})$ fails to reduce to another pair $\qty(C^{(L)}\cup C^{\mathrm{wait}}, \{q_i\}_{i\in C^{(L)}\cup C^{\mathrm{wait}}})$,
where
\begin{align}
    q_j&= 0&&\text{if $j \in C^{(L)}$},\label{eq: error parameter 1 in concatenated code}\\
    q_j&\leq Mp_{\mathrm{loc}}&&\text{if $j \in C^{\mathrm{wait}}$}.
\label{eq: error parameter 2 in concatenated code}
\end{align}
The failure probability of the reduction is bounded by the union bound.
Since there exists a positive constant $c_3$ such that $|C^{(L)}|\leq c_3 W(C) D(C)$, we have
\begin{align} 
    q_j \times |C^{(L)}| &\leq q_j \times c_3 W(C) D(C),\\
    &\leq c_3 |C| {p}^{\mathrm{th}}_{\mathrm{loc}} \left( \frac{p_{\mathrm{loc}}}{{p}^{\mathrm{th}}_{\mathrm{loc}}} \right)^{2^L},
\end{align}
where we used \eqref{eq: doubly exponential supression}.
Thus, a pair $(\tilde{C}(C,\delta_3), \{p_i\}_{i\in\tilde{C}(C,\delta_3)})$ where $p_i\leq p_{\mathrm{loc}}$ for all $i\in\tilde{C}(C,\delta_3)$ is $\delta_3$-reducible to another pair $\qty(C\cup C^{\mathrm{wait}}, \{q_j\}_{j\in C\cup C^{\mathrm{wait}}})$ satisfying \eqref{eq: error parameter 1 in concatenated code} and \eqref{eq: error parameter 2 in concatenated code}, where we use the equivalence of $C$ and $C^{(L)}$ in the sense in Def.~\ref{def: incomplete circuits}.

From the above discussion, we can conclude the proof.
\end{proof}

\bibliography{bibliography}

\end{document}